\documentclass[11pt, oneside]{article} 
\pdfoutput=1

\usepackage{jheppub}

\usepackage{amsmath,amssymb,amsfonts,amsthm}
\usepackage{color}
\usepackage{booktabs}
\definecolor{darkblue}{rgb}{0.1,0.1,.7}
\usepackage[]{amsmath}
\usepackage[]{graphicx}
\usepackage[]{latexsym}
\usepackage{amscd}
\usepackage[all,cmtip]{xy}
\usepackage{mathrsfs}
\usepackage{bbold}
\usepackage{subfigure}
\usepackage[margin=10pt,font=small,labelfont=bf]{caption}
\usepackage{simplewick}
\usepackage{changepage}
\usepackage[]{algorithm2e}
\usepackage{booktabs,multirow}
\usepackage[OT2,OT1]{fontenc}
\usepackage{braket}
\usepackage{tikz}
\usepackage{pgfplots}
\pgfplotsset{compat=1.10}
\usepgfplotslibrary{fillbetween}
\usetikzlibrary{patterns}

\newcommand{\abs}[1]{\left\lvert#1\right\rvert}

\newtheorem{theorem}{Theorem}[section]
\newtheorem{lemma}[theorem]{Lemma}
\newtheorem{corollary}[theorem]{Corollary}

\newtheorem{conjecture}[theorem]{Conjecture}
\theoremstyle{remark}
\newtheorem{remark}[theorem]{Remark}

\makeatletter
\def\@fpheader{\ }
\makeatother

\title{Lightcone Modular Bootstrap and Tauberian Theory\\ \vspace{0.4cm}\large{A Cardy-like Formula for Near-extremal Black Holes}}
\author{Sridip Pal$^{\infty}$, Jiaxin Qiao$^{8}$}
\affiliation{$^\infty$Walter Burke Institute for Theoretical Physics,  California Institute of Technology,  Pasadena, CA, USA  \\  $^8$ Laboratory for Theoretical Fundamental Physics, Institute of Physics, École Polytechnique Fédérale de Lausanne (EPFL), CH-1015 Lausanne, Switzerland\\ }

\abstract{We show that for a unitary modular invariant 2D CFT with central charge $c>1$ and having a nonzero twist gap in the spectrum of Virasoro primaries, for sufficiently large spin $J$,  there always exist spin-$J$ operators with twist falling in the interval $(\frac{c-1}{12}-\varepsilon,\frac{c-1}{12}+\varepsilon)$ with $\varepsilon=O(J^{-1/2}\log J)$.  We establish that the number of Virasoro primary operators in such a window has a Cardy-like i.e. \!$\exp\left(2\pi\sqrt{\frac{(c-1)J}{6}}\right)$ growth. We make further conjectures on potential generalization to CFTs with conserved currents. A similar result is proven for a family of holographic CFTs with the twist gap growing linearly in $c$ and a uniform boundedness condition, in the regime $J\gg c^3\gg1$. From the perspective of near-extremal rotating BTZ black holes (without electric charge), our result is valid when the Hawking temperature is much lower than the ``gap temperature".}

\begin{document}
	
\maketitle

\section{Introduction}
Understanding universal properties is of fundamental importance in the study of physical phenomena. In the realm of two-dimensional conformal field theories (2D CFTs), a famous example is the Cardy formula \cite{Cardy:1986ie}. The Cardy formula establishes an asymptotic relation between the microcanonical entropy $S_\delta(\Delta)$, which counts the number of states with scaling dimensions $\Delta'$ in a window $(\Delta-\delta, \Delta+\delta)$, and the central charge $c$:
\begin{equation}
	\begin{split}
		S_\delta(\Delta) := \log\left(\sum\limits_{\abs{\Delta'-\Delta}<\delta}n_{\Delta'}\right) = 2\pi\sqrt{\frac{c\Delta}{3}}+O(\log\Delta)\quad(\Delta\rightarrow\infty),
	\end{split}
\end{equation}
where $n_{\Delta}$ denotes the number of states with scaling dimension $\Delta$. This formula, along with its generalizations \cite{Verlinde:2000wg,Kutasov:2000td,Shaghoulian:2015kta,Kraus:2016nwo,Das:2017vej, Das:2017cnv,Cardy:2017qhl,Brehm:2018ipf,Collier:2019weq,Pal:2020wwd,Harlow:2021trr,Belin:2021ryy,Anous:2021caj,Lin:2022dhv,Benjamin:2023qsc},  provides a universal connection between the high energy spectrum of a CFT and its central charge.
It has played a significant role in the black hole physics, such as black hole microstate counting, the study of the Hawking-Page phase transition and checking AdS/CFT correspondence \cite{Strominger:1996sh,Strominger:1997eq,Sen:2007qy,Hellerman:2009bu,Sen:2012dw,Hartman:2014oaa}.

It has been recently realized that the Cardy formula holds only on average \cite{Mukhametzhanov:2019pzy,Pal:2019zzr}(see also appendix C of \cite{Das:2017vej}),  and its precise validity requires a more rigorous treatment using Tauberian theory \cite{korevaar2004tauberian} (as explained in \cite{Pappadopulo:2012jk,Qiao:2017xif}). In this paper, we aim to prove a Cardy-like formula
\begin{equation}\label{def:Micentropy}
	\begin{split}
		S^{Vir}_\kappa(J)=2\pi\sqrt{\frac{c-1}{6}J}+O(\log(J))\quad (J\rightarrow\infty),
	\end{split}
\end{equation}
for Virasoro primaries near a ``twist accumulation point". Here, the microcanonical entropy $S^{Vir}_\kappa(J)$ is the logarithm of the number of spin-$J$ Virasoro primaries $\mathcal{O}_{\Delta,J}$ in a shrinking window of twist $\tau$:
\begin{equation}
	\begin{split}
		\abs{\tau-\frac{c-1}{12}}<2\kappa J^{-1/2}\log J,\quad\tau:=\Delta-J\,,  \ \kappa>0\,.
	\end{split}
\end{equation}
Here the lower bound of the allowed $\kappa$ is proportional to $\tau_{\rm gap}^{-1}$, where $\tau_{\rm gap}$ is the twist gap in the spectrum of Virasoro primaries, and the factor of $2$ is just a convention.

Our primary focus is on the unitary modular invariant 2D CFTs with central charge $c>1$ and having a twist gap $\tau_{\rm gap}>0$ in the spectrum of Virasoro primaries. We will study the CFT torus partition function in the so-called \emph{double lightcone limit}. In our recent work \cite{Pal:2022vqc}, we conducted a rigorous analysis of the torus partition function for these 2D CFTs under this limit. The outcome of this analysis led to the establishment of a theorem that confirm certain well-known claims previously discussed in the modular bootstrap literature \cite{Collier:2016cls,Kusuki:2018wpa,Collier:2018exn,Benjamin:2019stq,Afkhami-Jeddi:2017idc}. In particular, it was shown in \cite{Pal:2022vqc} that the theory must include a family of Virasoro primary operators $\mathcal{O}_{\Delta,J}$ with 
\begin{equation}
	\begin{split}
		\Delta,J\rightarrow\infty,\quad\tau\equiv\Delta-J\rightarrow2A\left(\equiv\frac{c-1}{12}\right),
	\end{split}
\end{equation}
where $\Delta=h+\bar h$ is scaling dimension and $J=|h-\bar h|$ is spin. This rigorous framework gives us a powerful tool to investigate more detailed questions about the universality of the fixed twist, large spin spectrum of operators in such CFTs.  One natural question that arises is: how many spin-$J$ Virasoro primary operators have a twist near $\frac{c-1}{12}$? The study of this question leads us to eq.\,\eqref{def:Micentropy}.

The CFTs of aforementioned kind are also known as \textit{irrational} CFTs with Virasoro symmetry only.\footnote{A recent proposal suggests constructing examples of irrational CFTs by weakly relevant deformations from multiple copies of minimal models. See \cite{Antunes:2022vtb} for details.} They are expected to exhibit chaos and/or some form of random-matrix-like statistics in their spectrum of primary operators.  In particular,  the spectrum of primary operators are expected to have dense spacing in appropriate asymptotic sense.  Some of these expectations are byproduct of holography, where we know high energy spectrum of dual CFT capture the black hole microstates.  While we expect the quantum systems dual to black hole to be chaotic,  for example SYK \cite{Saad:2018bqo}, from a CFT perspective,  it is far from obvious to see the imprint of chaos (the recent explorations in this direction\footnote{There have been works probing the signature of chaos in CFT correlation functions with various degrees of rigor and implicit assumptions, for example butterfly effect \cite{Roberts:2014ifa},  ETH like statements as in \cite{Kraus:2016nwo,Das:2017cnv,Basu:2017kzo,Cardy:2017qhl, Brehm:2018ipf,Romero-Bermudez:2018dim,Hikida:2018khg,Collier:2019weq},  however,  to best of our knowledge,  ETH has not been proven in $2$D CFT.  See also EFT approach to chaos \cite{Haehl:2018izb,Anous:2020vtw,Altland:2020ccq,Choi:2023mab}.} include \cite{Collier:2021rsn}, \cite{Haehl:2023tkr} built upon \cite{Benjamin:2021ygh}) in the spectrum of CFT operators.\footnote{Also see \cite{DiUbaldo:2023qli} for a recent discussion on the random matrix behavior of 2d CFTs and AdS$_3$ quantum gravity.} For instance, in the regime of fixed spin and large $\Delta$, it is expected that the asymptotic spacing in $\Delta$ becomes exponentially small with respect to the entropy, scaling as $\sqrt{\Delta}$, and ultimately approaches zero. However, the current best bound in this direction, without assuming a twist gap, is $1$ as established in \cite{Ganguly:2019ksp, Mukhametzhanov:2020swe}, which represents an improvement upon the results of \cite{Mukhametzhanov:2019pzy}. It should be noted that this bound is optimal in the absence of a twist gap but is expected to be sub-optimal when a twist gap is imposed. 

In this paper, with the assumption of a twist gap $\tau_{\rm gap}$, we prove the existence of a ``dense" spectrum, characterized by a large number of Virasoro primary states and powerlaw decreasing spacing of twist, in the vicinity of a specific fixed twist value $\tau=\frac{c-1}{12}$ and for very large spin $J$.\footnote{It is important to note that the term ``dense" used here should not be confused with its typical usage in chaos-related CFT literature. In the chaos context, ``dense" often refers to the property where the spacing between adjacent energy spectra is generally given by $\rho^{-1}$, where $\rho\sim e^{\#\sqrt{E}}$ represents the coarse-grained spectral density. However, in this paper, when we refer to ``dense", we simply mean that as we consider specific powerlaw decreasing windows, the number of spectra within those windows increases as $e^{\#\sqrt{E}}$. It is worth emphasizing that the specific distribution of the spectra within the window remains unknown.}  This result holds for irrational CFTs that possess Virasoro symmetry exclusively.
\begin{itemize}
\item \textit {We prove a refined version of twist accumulation result: there always exist  a Virasoro primary operator with sufficiently large spin $J$ within a narrow window of twist around the twist accumulation point i.e $(c-1)/12$; the width of such window of twist goes to $0$ as $J^{-1/2}\log J$. }

\item \textit{We rigorously establish two-sided bounds for $\mathcal{N}_J(\varepsilon)$, the number of spin-$J$ Virasoro primary operators $\mathcal{O}_{\Delta,J}$ in a window $|\Delta-J-\frac{c-1}{12}|\leqslant 2\varepsilon$, and let $\varepsilon$ scale as $ \varepsilon=\kappa J^{-1/2}\log J$ with $\kappa$ being a fixed positive number of order $O(\tau_{\rm gap}^{-1})$. In the limit $J\to\infty$,  $\mathcal{N}_J$ grows as 
\begin{equation}\label{intro:NJ}
\mathcal{N}_J(\varepsilon\equiv\kappa J^{-1/2}\log J)=e^{4\pi\sqrt{AJ}+O(\log J)}\,,
\end{equation}
which is equivalent to \eqref{def:Micentropy} via the relation
\begin{equation}
	\begin{split}
		S^{Vir}_\kappa(J)\equiv \log\left(\mathcal{N}_J(\varepsilon\equiv\kappa J^{-1/2}\log J) \right).
	\end{split}
\end{equation}
A more precise form of eq.\,\eqref{intro:NJ} is stated in \ref{cor:operatorcount},  a corollary of the main theorem  \ref{theorem:modulartauberianFJ},  that we prove in this paper. }

\item \textit{We make further conjectures on potential generalization of the result for a CFT with conserved currents in section \ref{sec:currents}; also see appendix \ref{app:WNconjecture} for statement about $W_N$ CFTs.}
\end{itemize}

In this paper we further consider a family of unitary modular invariant 2D CFTs, including the large central charge limit $A\equiv\frac{c-1}{24}\to\infty$, such that (a) the lower bound of the twist gap $2T$ grows at least linearly in central charge, i.e. $T/A\geqslant \alpha >0$, and (b) their partition functions satisfy a uniform boundedness condition, inspired by the HKS sparseness condition \cite{Hartman:2014oaa}.  Holographically, this family of CFTs probes the near extremal rotating BTZ black holes, having a nearly $AdS_2\times S^1$ throat. One expects a Schwarzian theory \cite{Maldacena:2016upp} to describe such limit. We have the following main result:
\begin{itemize}
\item \textit{For a such a class of CFTs with sufficiently high central charge,  we rigorously estimate the number of operators $\mathcal{O}_{\Delta,J}$ with sufficiently large spin $J$, scaling dimension $\Delta$  such that  
\begin{equation}\label{result:spacing}
\Delta-J-2A \in \left(-\varepsilon_1,\varepsilon_2 \right) \,, \quad
\left(\varepsilon_1\equiv\frac{1}{\pi\alpha}\sqrt{\frac{A}{J}}\log\left(AJ\right),\ \varepsilon_2\equiv\frac{3}{\kappa}\sqrt{\frac{A}{J}}(2\pi A+\log J)\right).
\end{equation}
Here  $\kappa$ is a positive constant, and its precise value will be clear later.}

\textit{In the limit $A\to \infty$ and $J/A^3\to\infty$, an analogue of \eqref{intro:NJ} reads }
\begin{equation}
 \mathcal{N}_J(\varepsilon_1,\varepsilon_2)=e^{4\pi\sqrt{AJ}+O(A)+O(\log AJ)}\,.
 \end{equation}
 \textit{See the theorem.\!~\ref{theorem:largec} and its corollary  \ref{cor:operatorcount:largec} for the precise version.}
\end{itemize}
This result has a gravitational interpretation in terms of the near-extremal rotating BTZ black holes with angular momentum $J$. The entropy of the near-extremal rotating BTZ black hole is given by the formula
	$$S_{\rm BH}\approx2\pi\sqrt{\frac{c}{6}J}\approx4\pi\sqrt{AJ},\quad c=\frac{3\ell_3}{2G_N}\gg1,$$
	where $\ell_3$ is the radius of AdS$_3$, $G_N$ is Newton's constant and $c=\frac{3\ell_3}{2G_N}$ is the Brown-Henneaux relation \cite{Brown:1986nw}. This formula is known in the standard black hole thermodynamics. Our result supports the thermodynamic description of the near-extremal black holes when the Hawking temperature $T_{\rm H}$, given by $T_{\rm H}=\beta^{-1}$, falls within a certain regime:
	\begin{equation}
		\text{const}\times \frac{\sqrt{c/J}}{\alpha} \leqslant T_{\rm H} \ll 1/c.
	\end{equation}
	In particular, the Hawking temperature is much lower than the ``gap temperature" $c^{-1}$.

In this paper, we leverage existing techniques to analyze the partition function in the lightcone limit with complex $\beta_L$ or $\beta_R$.   While the estimates related to the lightcone bootstrap were already given in \cite{Pal:2022vqc},  they were applicable for partition function evalauted at real $\beta_L$ and $\beta_R$.  The main technical challenge for us is to uplift the aformentioned rigorous estimate as done in \cite{Pal:2022vqc} so that it applies to the partition function for complex $\beta_L$ or $\beta_R$ and we are able to learn about the large spin,  small twist spectra.  We achieve this by using the Tauberian theory techniques developed in \cite{Mukhametzhanov:2020swe, Mukhametzhanov:2019pzy} to analyze the partition function for complex $\beta$ albeit in the high temperature limit (not the light cone limit).  Our main contribution in this paper lies in combining these techniques to analyze the partition function in the lightcone limit with complex $\beta_L$ or $\beta_R$; this leads to the main results of this paper.


We view our results as a stepping stone towards a rigorous understanding of chaotic irrational CFTs, although it has not yet been established in a general $c>1$ irrational CFT with a twist gap. We anticipate that with further effort, a similar analysis can be applied to CFT four-point functions based on \cite{Fitzpatrick:2012yx, Komargodski:2012ek} and \cite{Mukhametzhanov:2018zja, Pal:2022vqc}.\footnote{The analysis of CFT four-point functions using the lightcone bootstrap would be more complicated than the modular bootstrap approach because the conformal blocks of CFT four-point functions do not naturally factorize into left- and right-movers. However, in the double lightcone limit, the conformal blocks exhibit approximate factorization (see \cite{Pal:2022vqc}, appendix A). Based on this observation, we anticipate that a similar Tauberian theorem to theorem \ref{theorem:modulartauberianFJ} can be established using the techniques explained in \cite{Qiao:2017xif}.}

The paper is organized as follows. In section \ref{section:modularbootstrap}, we present the proof of the Cardy-like formula \eqref{def:Micentropy}, and the main results of this section are summarized in theorem \ref{theorem:modulartauberianFJ} and its corollary \ref{cor:operatorcount}. Appendices \ref{app:Ivac:asym} and \ref{appendix:uniformbounds} provide additional technical details related to this section. In section \ref{sec:currents}, we explore potential generalizations for CFTs with conserved currents, {and we include specific examples and leave technical details to appendix \ref{appendix:example}}. Moving on to section \ref{section:largec}, we focus on holographic CFTs and investigate the limit of large central charge, $c\rightarrow\infty$. The main results of this section are summarized in theorem \ref{theorem:largec} and its corollary \ref{cor:operatorcount:largec}, with the proofs presented in appendix \ref{section:estimatelargec}.  Then we discuss the connection between our results and the thermodynamics of near-extremal rotating BTZ black holes.  In section \ref{section:conclusion}, we make conclusions and discuss some potential future directions.

\section{Modular bootstrap}\label{section:modularbootstrap}
\subsection{Setup}\label{section:modularsetup}
We consider a unitary, modular invariant 2D CFT with central charge $c>1$, a (unique) normalizable vacuum and a positive twist gap $\tau_{\rm gap}>0$ in the spectrum of Virasoro primaries. The torus partition function $Z(\beta_L,\beta_R)$ of such a CFT is defined by
\begin{equation}
	Z (\beta_L,\beta_R) \equiv \text{Tr}_{\mathcal{H}_{\text{CFT}}} \left( e^{- \beta_L
		\left( L_0 - \frac{c}{24} \right)} e^{- \beta_R \left( \bar{L}_0 -
		\frac{c}{24} \right)} \right) . \label{Z}
\end{equation}
where $\beta_L$ and $\beta_R$ are the inverse temperatures of the left and right movers, $L_0$ and $\bar{L}_0$ are the standard Virasoro algebra generators and $\mathcal{H}_{\rm CFT}$ is the CFT Hilbert space which is assumed to be the direct sum of Virasoro representations characterized by conformal weights $h$ and $\bar{h}$
\begin{equation}
	\begin{split}
		\mathcal{H}_{\rm CFT}=\bigoplus_{h,\bar{h}} V_{h} \otimes V_{\bar{h}}. \label{def:Hilbertspace}
	\end{split}
\end{equation}
The twist gap assumption means that
\begin{equation}
	\begin{split}
		h,\bar{h}\geqslant\tau_{\rm gap}/2
	\end{split}
\end{equation}
for all representations except the vacuum representation ($h=\bar{h}=0$). 

Using eqs.\,(\ref{Z}) and (\ref{def:Hilbertspace}), the torus partition function can be written as a sum of Virasoro characters $\chi_h(\beta_L)\chi(\beta_R)$ over primaries
\begin{equation}
	\begin{split}
		Z(\beta_L,\beta_R)=\sum\limits_{h,\bar{h}}n_{h,\bar{h}}\ \chi_h(\beta_L)\chi_{\bar{h}}(\beta_R),
	\end{split}
\end{equation}
where $n_{h,\bar{h}}$ counts the degeneracy of the Virasoro primaries with conformal weights $h$ and $\bar{h}$. For $c>1$, the characters of Virasoro unitary representations are given by 
\begin{equation}\label{def:Vircharacter}
	\chi_h (\beta) \equiv \text{Tr}_{V_h} \left( e^{- \beta \left( L_0 -
		\frac{c}{24} \right)}  \right) = \frac{e^{\frac{c - 1}{24} \beta}}{\eta
		(\beta)} \times \begin{cases}
		1 - e^{- \beta} &\text{if } h = 0,\\
		e^{- \beta h}&\text{if } h > 0,
	\end{cases}
\end{equation}
where the Dedekind eta function $\eta(\beta)\equiv e^{-\beta/24}\prod\limits_{n=1}^{\infty}(1-e^{-n \beta})$ accounts for the contribution of descendants. Then we have
\begin{equation}
	Z (\beta_L, \beta_R) = \frac{\tilde{Z} (\beta_L,
		\beta_R)}{\eta (\beta_L) \eta (\beta_R)} ,\label{ZtoZtilde}
\end{equation}
where the reduced partition function $\tilde{Z}$ is given by
\begin{equation}
	\tilde{Z} (\beta_L, \beta_R) = e^{A(\beta_L + \beta_R)}\left[(1 - e^{- \beta_L}) (1 - e^{-
		\beta_R}) + \sum_{h, \bar{h}\geqslant T} n_{h,\bar{h}}\,e^{- \beta_L h - \beta_R \bar{h}}\right]
	\label{def:Ztilde} .
\end{equation}
Here we have denoted $A \equiv \frac{c - 1}{24}$ and $T\equiv\tau_{\rm gap}/2$ for convenience. $n_{h,\bar{h}}$ is the degeneracy of the Virasoro primaries with conformal weights $h$ and $\bar{h}$. The first term in the square bracket corresponds to the contribution from the vacuum state, while the second term represents the total contribution from Virasoro primaries with twists above the twist gap. 

The above formulations assumed a discrete spectrum. The argument below also works for the continuum spectrum, where eq.\,(\ref{def:Ztilde}) is replaced by
\begin{equation}
	\begin{split}
		\tilde{Z} (\beta_L, \beta_R) = e^{A(\beta_L + \beta_R)}\left[(1 - e^{- \beta_L}) (1 - e^{-
			\beta_R}) + \int_{T}^{\infty}d h\int_{T}^{\infty}d\bar{h}\ \rho(h,\bar{h}) e^{- \beta_L h - \beta_R \bar{h}}\right].\label{Ztilde:integral}
	\end{split}
\end{equation}
Here $\rho$ is a non-negative spectral density of Virasoro primaries, which is related to $n_{h,\bar{h}}$ by
\begin{equation}
	\begin{split}
		\rho(h,\bar{h})=\sum\limits_{h',\bar{h}'\geqslant T}n_{h',\bar{h}'}\delta(h-h')\delta(\bar{h}-\bar{h}').
	\end{split}
\end{equation}
We assume that (a) the partition function $Z$ (or equivalently $\tilde{Z}$) for a given CFT is finite when $\beta_L,\beta_R\in(0,\infty)$; (b) $Z$ is modular invariant, i.e.\,$Z(\beta_L,\beta_R)$ is invariant under the transformations generated by
\begin{equation}\label{def:modulartransf}
	\begin{split}
		(\beta_L,\beta_R)\rightarrow&(\beta_L+2\pi i,\beta_R-2\pi i), \\
		(\beta_L,\beta_R)\rightarrow&\left(\frac{4\pi^2}{\beta_L},\frac{4\pi^2}{\beta_R}\right). \\
		\end{split}
\end{equation}
The invariance under the first transformation implies that the spin $J:=\abs{h-\bar{h}}$ of any Virasoro primary state must be an integer. The invariance condition under the second transformation (which is called \emph{S modular transformation}),
\begin{equation}
	Z (\beta_L, \beta_R) = Z_{\text{}} \left( \frac{4 \pi^2}{\beta_L}, \frac{4
		\pi^2}{\beta_R} \right), \label{Smod}
\end{equation}
can be formulated in terms of reduced partition function $\tilde{Z}$ as follows. By (a) and the positivity of the spectral density, the convergence domain of $Z(\beta_L,\beta_R)$ (or equivalently $\tilde{Z}(\beta_L,\beta_R)$) can be extended to the complex domain of $(\beta_L,\beta_R)$ with\footnote{This justifies why the first modular transformation is well-defined on the partition function.}
\begin{equation}
	\begin{split}
		\mathrm{Re}(\beta_L),\mathrm{Re}(\beta_R)\in(0,\infty). \label{domain:complex}
	\end{split}
\end{equation}
Since under S modular transformation, $\eta$ behaves as $\eta(\beta)=\sqrt{\frac{2\pi}{\beta}}\eta(\frac{4\pi^2}{\beta})$ , eqs.\,(\ref{ZtoZtilde}) and (\ref{Smod}) imply that $\tilde{Z}$ transforms as
\begin{equation}\label{modulartransformation}
	\tilde{Z} (\beta_L, \beta_R) = \sqrt{\frac{4 \pi^2}{\beta_L \beta_R}} \tilde{Z}
	\left( \frac{4 \pi^2}{\beta_L}, \frac{4 \pi^2}{\beta_R} \right).
\end{equation}
Notice that the complex domain (\ref{domain:complex}) is preserved by the S modular transformation. Therefore we have two convergent expansions of $\tilde{Z}(\beta_L,\beta_R)$ for $(\beta_L,\beta_R)$ in the domain (\ref{domain:complex}):
\begin{itemize}
	\item Direct channel: expanding l.h.s. of (\ref{modulartransformation}) in terms of (\ref{Ztilde:integral}).
	\item Dual channel: expanding r.h.s. of (\ref{modulartransformation}) in terms of (\ref{Ztilde:integral}) (with $\beta_L,\beta_R$ replaced by $\frac{4\pi^2}{\beta_L},\frac{4\pi^2}{\beta_R}$).
\end{itemize}

\subsection{Review of the twist accumulation point}\label{section:twistaccum}
Under the above setup, one can show that in the theory, there is at least one family of Virasoro primaries $\mathcal{O}_i$ with $h_i\rightarrow A$ and $\bar{h}_i\rightarrow\infty$ \cite{Collier:2016cls,Afkhami-Jeddi:2017idc,Benjamin:2019stq,Pal:2022vqc}. In other words, $(h=A,\bar{h}=\infty)$ is an accumulation point in the spectrum of Virasoro primaries. The same is true with $h$ and $\bar{h}$ interchanged. Here let us briefly explain why it is true. For more technical details, see \cite{Pal:2022vqc}, section 3.

We consider the reduced partition function $\tilde{Z}(\beta_L,\beta_R)$ for real and positive $(\beta_L,\beta_R)$. We take the double lightcone (DLC) limit , defined by\footnote{The ``DLC" limit defined in this context is referred to as the ``M$_*$" limit in \cite{Pal:2022vqc}. In that work, two distinct lightcone bootstrap problems were discussed, and the ``M$_*$" limit, specifically the modular double lightcone limit, was employed to differentiate it from the DLC limit in the other problem.}
	\begin{equation}\label{def:DLClimit}
		\begin{split}
			\text{DLC limit:}\quad \beta_L \rightarrow \infty,
			\quad \beta_R \rightarrow 0,
			\quad \mathfrak{b}(\beta_L,\beta_R):=\frac{4\pi^2T}{A\beta_R} -\beta_L-\frac{3}{A}\log(\beta_L)\rightarrow \infty.
		\end{split}
\end{equation}
The important feature of this limit is that $\beta_{R}$ approaches 0 much faster than $\beta_L$ approaches $\infty$. The introduction of the logarithmic term in $\mathfrak{b}(\beta_L,\beta_R)$ is just for technical reason. One can show that in the DLC limit, the partition function $\tilde{Z}(\beta_L,\beta_R)$ is dominated by the vacuum term (the first term in eq.\,(\ref{def:Ztilde})) in the dual channel, i.e.
\begin{equation}\label{DLCpart}
	\underset{\rm DLC}{\lim} \frac{\tilde{Z} (\beta_L,\beta_R)}{\frac{8\pi^3}{\beta_L^{3/2}\beta_R^{1/2}} e^{\frac{4\pi^2A}{\beta_R}}} = 1. 
\end{equation}
Here the denominator is the asymptotic behavior of the vacuum term in the dual channel:
\begin{equation}\label{Ztilde:asympdual:withoutcurrent}
	\begin{split}
		\sqrt{\frac{4\pi^2}{\beta_L\beta_R}}\tilde{Z}_{vac}\left(\frac{4\pi^2}{\beta_L},\frac{4\pi^2}{\beta_R}\right)\equiv&\sqrt{\frac{4\pi^2}{\beta_L\beta_R}}e^{A\left(\frac{4\pi^2}{\beta_L}+\frac{4\pi^2}{\beta_R}\right)}\left(1-e^{-\frac{4\pi^2}{\beta_L}}\right)\left(1-e^{-\frac{4\pi^2}{\beta_R}}\right) \\
		\sim&\frac{8\pi^3}{\beta_L^{3/2}\beta_R^{1/2}} e^{\frac{4\pi^2A}{\beta_R}}\quad(\beta_{L}\rightarrow\infty,\ \beta_{R}\rightarrow0). \\
	\end{split}
\end{equation}
Now we consider the direct channel (i.e.\,the l.h.s. of eq.\,(\ref{modulartransformation})). Let $\Omega$ be a set of $(h,\bar{h})$ pairs (assuming that $\Omega$ does not contain the vacuum $(0,0)$). We define $\tilde{Z}_{\Omega}$ to be the partial sum of eq.\,(\ref{def:Ztilde}) with $(h,\bar{h})\in\Omega$:
\begin{equation}
	\begin{split}
		\tilde{Z}_{\Omega}(\beta_L,\beta_R):=\sum_{(h,\bar{h})\in\Omega} n_{h,\bar{h}}\,e^{- \beta_L h - \beta_R \bar{h}}
	\end{split}
\end{equation}
In what follows we will only state the conditions of $\Omega$, e.g.~$\tilde{Z}_{h\geqslant A+\varepsilon}$ is the same as $\tilde{Z}_{\Omega}$ with $\Omega=[A+\varepsilon,\infty)\times(0,\infty)$. The claim is that in the DLC limit, the direct channel is dominated by the sum over $h\in(A-\varepsilon,A+\varepsilon)$ and $\bar{h}\geqslant\bar{h}_\ast$:
\begin{equation}
	\begin{split}
		\underset{\rm DLC}{\lim} \frac{\tilde{Z}_{h\in(A-\varepsilon,A+\varepsilon),\bar{h}\geqslant\bar{h}_\ast}
			\left( \beta_L, \beta_R \right)}{\frac{8\pi^3}{\beta_L^{3/2}\beta_R^{1/2}} e^{\frac{4\pi^2A}{\beta_R}}} = 1. \label{modular:crossedchannel}
	\end{split}
\end{equation}
Here $\varepsilon>0$ can be arbitrarily small and $\bar{h}_\ast$ can be arbitrarily large (but they are fixed when we take the DLC limit). To prove this claim, \cite{Pal:2022vqc} demonstrated that in the direct channel, the total contribution from other $(h,\bar{h})$ pairs is suppressed compared to  the dual-channel vacuum term. 

Therefore, $(h=A,\bar{h}=\infty)$ must be an accumulation point in the spectrum. Otherwise one can find sufficiently small $\varepsilon$ and sufficiently large $\bar{h}_\ast$ such that $\tilde{Z}_{h\in(A-\varepsilon,A+\varepsilon),\bar{h}\geqslant\bar{h}_\ast}=0$, contradicting eq.\,(\ref{modular:crossedchannel}). By interchanging the roles of $\beta_L$ and $\beta_{R}$ in the above argument, we can show that $(h=\infty,\bar{h}=A)$ is also an accumulation point in the spectrum.

In terms of scaling dimension $\Delta=h+\bar{h}$ and spin $J=|h-\bar{h}|$, the above argument implies that the theory must include a family of Virasoro primary operators $\mathcal{O}_{\Delta,J}$ with
	\begin{equation}\label{twistaccumulation:DeltaJ}
		\begin{split}
			\Delta,J\rightarrow\infty,\quad\Delta-J\rightarrow2A\left(\equiv\frac{c-1}{12}\right).
		\end{split}
\end{equation}
For general CFTs, (\ref{twistaccumulation:DeltaJ}) is slightly weaker than the existence of both $(h\rightarrow\infty,\bar{h}\rightarrow A)$ and $(h\rightarrow A,\bar{h}\rightarrow\infty)$ families. In a CFT with conserved parity, these two statements are the equivalent.

\subsection{Main theorem}\label{section:modular:idea}
In the previous subsection, we reviewed that in the double lightcone limit, the dominant contribution to the reduced partition function $\tilde{Z}(\beta_L,\beta_R)$ comes from the spectrum with high spin and twist near $2A$ in the direct channel, while the vacuum state ($h=\bar{h}=0$) dominates in the dual channel. This observation implies a connection between the spectral density $\rho(h,\bar{h})$ (as given by eq.\,(\ref{Ztilde:integral})) near the accumulation point $(h=A,\bar{h}=\infty)$ and the vacuum term of the partition function in the dual channel. In fact, by conducting a more thorough analysis of the arguments presented in \cite{Pal:2022vqc}, we can not only establish the existence of an infinite number of operators near the accumulation point $(h=A, \bar{h}=\infty)$ but also estimate how many such operators there are. To achieve this quantitative understanding, we will employ Tauberian theory \cite{korevaar2004tauberian}, building upon similar reasoning presented in \cite{Mukhametzhanov:2019pzy, Mukhametzhanov:2020swe}, and integrate it with the arguments put forth in \cite{Pal:2022vqc}. This combined approach will be the focus of the remaining sections in this paper.

The object we are going to study is the total number of Virasoro primaries with $h$ in the range $(A-\varepsilon, A+\varepsilon)$ and $\bar{h}=h+J$ where the spin $J$ is fixed. This quantity is denoted as $\mathcal{N}_J(\varepsilon)$ and can be expressed as the sum of the degeneracies $n_{h,h+J}$ of Virasoro primaries over the specified range of $h$:
\begin{equation}
	\mathcal{N}_J(\varepsilon) := \sum\limits_{h\in(A-\varepsilon, A+\varepsilon)} n_{h,h+J}.
\end{equation}
Our goal is to derive non-trivial asymptotic two-sided bounds on $\mathcal{N}_J(\varepsilon)$ in the limit $J\rightarrow\infty$ and $\varepsilon\rightarrow0$, under specific constraints between $\varepsilon$ and $J$. However, due to technical limitations, a direct estimate of $\mathcal{N}_J(\varepsilon)$ is not feasible.\footnote{While it is possible to compute $\mathcal{N}_J(\varepsilon)$ directly on a case-by-case basis, our primary objective in this paper is to derive universal behaviors of $\mathcal{N}_J(\varepsilon)$.} To overcome this, we introduce another quantity $\mathcal{A}_J(\beta_L,\varepsilon)$ by assigning a $\beta_{L}$-dependent weight to each degeneracy $n_{h,\bar{h}}$ of Virasoro primaries:
\begin{equation}\label{def:AJ}
	\begin{split}
		\mathcal{A}_J(\beta_L,\varepsilon):=\sum\limits_{h\in(A-\varepsilon,A+\varepsilon)}n_{h,h+J}e^{-(h-A)\beta_L}.
	\end{split}
\end{equation}
Importantly, $\mathcal{N}_J(\varepsilon)$ and $\mathcal{A}_J(\beta_L,\varepsilon)$ are related by the following inequality:
\begin{equation}\label{NAineq}
	\begin{split}
	e^{-\varepsilon\beta_{L}}\mathcal{A}_J(\beta_{L},\varepsilon)\leqslant	\mathcal{N}_J(\varepsilon)\leqslant e^{\varepsilon\beta_{L}}\mathcal{A}_J(\beta_{L},\varepsilon).
	\end{split}
\end{equation}
This inequality provides an upper and lower bound for $\mathcal{N}_J(\varepsilon)$ in terms of $\mathcal{A}_J(\beta_L,\varepsilon)$, with a dependence on the parameter $\varepsilon$ and the inverse temperature $\beta_L$. So our approach involves two main steps. First, we will derive asymptotic two-sided bounds for $\mathcal{A}_J(\beta_L,\varepsilon)$. Then, we will use eq.\,\eqref{NAineq} to obtain corresponding bounds for $\mathcal{N}_J(\varepsilon)$.

To estimate $\mathcal{A}_J(\beta_L,\varepsilon)$, we introduce the DLC$_w$ (double lightcone) limit defined as follows:
\begin{equation}\label{def:DLCwJ}
	\begin{split}
		\mathrm{DLC}_w\ \mathrm{limit}:&\quad\beta_L,\ J\rightarrow\infty,\quad \frac{2\pi T(1-w^2)}{A}\sqrt{\frac{J}{A}}-\beta_L\rightarrow\infty\,,\\
		&\quad \beta_L^{-1}\log{J}\to 0\,.
	\end{split}
\end{equation}
The reason we still refer to it as the ``DLC" limit, similar to \eqref{def:DLClimit}, will become clearer later. For now, a brief explanation is that by introducing the additional identification
$$\beta_{R}=2\pi\sqrt{\frac{A}{J}},$$
\eqref{def:DLCwJ} becomes a slightly stronger form of \eqref{def:DLClimit}. We will revisit this point later around \eqref{def:DLCwlimit}.

With the aforementioned setup, we present our main theorem as follows:
\begin{theorem}\label{theorem:modulartauberianFJ}
	Take any unitary, modular invariant 2D CFT with central charge $c>1$ (i.e.\,$A\equiv\frac{c-1}{24}>0$), a unique normalizable vacuum and a twist gap $\tau_{\rm gap}\equiv2T>0$ in the spectrum of nontrivial Virasoro primaries. 
	
	Then for any $w\in\left(\frac{1}{2},1\right)$ fixed, and $\varepsilon$ within the range
	\begin{equation}\label{epsilon:choice:DLCJ}
		\begin{split}
			&\varepsilon_{\rm min}(\beta_{L},J)\leqslant\varepsilon\leqslant1-\frac{1}{2w}, \\
			\varepsilon_{\rm min}(\beta_{L},J)&:=\max\left\{\frac{A^{3/2}}{\pi w^2T}\frac{\log{J}}{\sqrt{J}},\ 
			\frac{3\log{J}}{4\beta_{L}}+\frac{2\log\beta_{L}}{\beta_{L}}\right\},
		\end{split}
	\end{equation}
	the quantity $\mathcal{A}_J$, defined in \eqref{def:AJ}, satisfies the following asymptotic two-sided bounds in the DLC$_w$ limit \eqref{def:DLCwJ}:
	\begin{equation}
		\begin{split}\label{eq:resultFiniteEpsilon}
			\frac{1}{w}\frac{1}{1-\frac{\tan\left(\pi w(1-\varepsilon)\right)}{\pi w(1-\varepsilon)}}\lesssim \frac{\mathcal{A}_J(\beta_L,\varepsilon )}{4\pi^{5/2}\beta_L^{-3/2}J^{-1/2}e^{4\pi\sqrt{AJ}}} \lesssim \frac{1}{w}\frac{2}{1+\frac{\sin(2\pi w\varepsilon)}{2\pi w\varepsilon}},
		\end{split}
	\end{equation}
    which is uniform in $\varepsilon$. Here by $a\lesssim b$ we mean
    \begin{equation}
    	\begin{split}
    		\lim\frac{a}{b}\leqslant1
    	\end{split}
    \end{equation}
    in the considered limit.
\end{theorem}
Let's make some remarks on Theorem \ref{theorem:modulartauberianFJ}:
\begin{remark}\label{Remark3.3}
	(a) In eq.\,(\ref{eq:resultFiniteEpsilon}), the upper bound is always greater than the lower bound for the assumed range of $w$. This is because the upper bound is consistently larger than $\frac{1}{w}$, while the lower bound monotonically decreases with $\varepsilon$ in the interval $\varepsilon\in\left(0,1-\frac{1}{2w}\right)$. Notably, when $\varepsilon=0$, the lower bound is less than or equal to $\frac{1}{w}$. This observation provides a consistency check for the validity of the two-sided bounds.
	
	(b) The gap between the upper and lower bounds in eq.\,\eqref{eq:resultFiniteEpsilon} decreases as we increase $w$ (i.e.\,when a stronger DLC$_w$ limit is imposed) and decrease $\varepsilon$. In the limit $\varepsilon\to0$ and $w\to1$, both the upper and lower bounds converge to 1. 
	
	(c) In the DLC$_w$ limit, the lower bound $\varepsilon_{\rm min}(\beta_L,J)$ for $\varepsilon$ approaches zero. We note that our choice of $\varepsilon_{\rm min}(\beta_L,J)$ is not optimal with respect to the method that we use, in the sense that the coefficients of the logarithms in \eqref{epsilon:choice:DLCJ} can be further improved. But we expect that the current form of $\varepsilon_{\rm min}$ already captures its essential behavior in the double lightcone limit, namely $\varepsilon_{\rm min}=O(J^{-1/2}\log J)$.
\end{remark}
Using theorem \ref{theorem:modulartauberianFJ}, we can obtain an estimate for $\mathcal{N}_J(\varepsilon)$. Let us consider the following constraints, which are compatible with the DLC$_w$ limit \eqref{def:DLCwJ} (when $J\rightarrow\infty$):
\begin{equation}
	\begin{split}
		\beta_{L}=3\kappa^{-1} J^{1/2},\quad \varepsilon=\kappa J^{-1/2}\log{J}\quad\left(\kappa^{-1}<\frac{2\pi T(1-w^2)}{3A^{3/2}}\ {\rm fixed}\right).
	\end{split}
\end{equation}
Substituting these values into eq.\,\eqref{NAineq} and theorem \ref{theorem:modulartauberianFJ}, and choose e.g.\,$w=\frac{3}{4}$, we obtain the following result:
\begin{corollary}\label{cor:operatorcount}
	Given any fixed $\kappa\in\left(\frac{4A^{3/2}}{\pi T},\infty\right)$, we have
	\begin{equation}
		\begin{split}
			\mathcal{N}_J(\varepsilon\equiv\kappa J^{-1/2}\log J)=J^{-5/4}e^{4\pi\sqrt{AJ}+f_{\kappa}(J)},
		\end{split}
	\end{equation}
where the error term $f_{\kappa}(J)$ satisfies the bound
\begin{equation}\label{eq:errorbound}
	\begin{split}
		\abs{f_{\kappa}(J)}\leqslant 3\log(J+1)+C(\kappa),
	\end{split}
\end{equation}
with $C(\kappa)$ being a finite constant.
\end{corollary} 
Before going to the proof, we have three remarks.
\begin{remark}
	(1) Recall the twist accumulation point is given by $\tau=2A\equiv\frac{c-1}{12}$, corollary \ref{cor:operatorcount} tells us that at large spin $J$, the number of states that are very closed to the twist accumulation point grows exponentially as $e^{2\pi\sqrt{\frac{(c-1)}{6}J}}$, with additional slow-growth factors that are bounded by powers of $J$. This implies that the average spacing between adjacent states in this regime is approximately given by $e^{-2\pi\sqrt{\frac{(c-1)}{6}J}}$.  However,  we can not at present rule out the possibility of having all the states piling up near the end points of the interval.  Therefore the rigorous upper bound on spacing is given by the size of window i.e $J^{-1/2}\log J$.
	
	(2) In corollary \ref{cor:operatorcount}, it is crucial to note that the lower bound of $\kappa$ is proportional to $T^{-1}$. This dependence clearly indicates that our analysis will not be valid if the theory does not have a twist gap.
	
	(3) If we further assume that the theory has some critical spin $J_*$, above which there are no Virasoro primaries with twist strictly below $2A\equiv\frac{c-1}{12}$, then all the Virasoro primaries have $h$ greater than or equal to $A$ when $\bar{h}\geqslant h+J_*$. Consequently, considering the exponential term $e^{(A-h)\beta_{L}}\leqslant1$, we find that the number of Virasoro primaries with $h$ in the window $\left[A,A+\kappa J^{-1/2}\log J\right)$ cannot be smaller than $\mathcal{A}_J$. This leads to a more precise lower bound on $\mathcal{N}_J$, given by:
	\begin{equation}\label{NJ:optimallowerbound}
		\begin{split}
			\mathcal{N}_J(\varepsilon\equiv\kappa J^{-1/2}\log J)\geqslant {\rm const} (J+1)^{-5/4}e^{4\pi\sqrt{AJ}},
		\end{split}
	\end{equation}
	where the constant prefactor is strictly positive. Here the power index $-5/4$ is obtained by choosing $\beta_{L}\sim J^{1/2}$ in \eqref{eq:resultFiniteEpsilon}. 
	
	The index of $-5/4$ in $\mathcal{N}_J(\varepsilon)$ can be understood by considering the contribution from the vacuum character in the dual channel. This can be naively reproduced by only taking into account this part of the contribution. To see this, we rewrite the dual vacuum character in terms of the Laplace transform of the modular crossing kernel:
	\begin{equation}
		\begin{split}
			&\sqrt{\frac{2\pi}{\beta}}e^{\frac{4\pi^2A}{\beta}}\left(1-e^{-\frac{4\pi^2}{\beta}}\right) \\
			=&\int_A^\infty\,d h\,\sqrt{\frac{2}{h-A}}\left[\cosh(4\pi\sqrt{A(h-A)})-\cosh\left(4\pi\sqrt{(A-1)(h-A)}\right)\right]e^{-(h-A)\beta}. \\
		\end{split}
	\end{equation}
   Therefore, a naive computation of the``vacuum character" contribution to $\mathcal{N}_J(\varepsilon)$ is as follows:
   \begin{equation}
   	\begin{split}
   		\left[\mathcal{N}_J(\varepsilon)\right]_{\rm naive}=&\int_A^{A+\varepsilon}d h\int_{h+J-1}^{h+J+1}\sqrt{\frac{4}{(h-A)(\bar{h}-A)}} \\
   		&\times\left[\cosh(4\pi\sqrt{A(h-A)})-\cosh\left(4\pi\sqrt{(A-1)(h-A)}\right)\right] \\
   		&\times\left[\cosh(4\pi\sqrt{A(\bar{h}-A)})-\cosh\left(4\pi\sqrt{(A-1)(\bar{h}-A)}\right)\right] \\
   		\sim&{\rm const}\,\varepsilon^{3/2}J^{-1/2}e^{4\pi\sqrt{AJ}}\quad(\varepsilon\ll1,\ J\gg1). \\
   	\end{split}
   \end{equation}
   By choosing $\varepsilon=\kappa J^{-1/2}\log J$, we obtain the correct index of $-5/4$ in \eqref{NJ:optimallowerbound}. We expect that this is the optimal power index of $J$ for the lower bound of $\mathcal{N}_J(\varepsilon\equiv\kappa J^{-1/2}\log J)$, in the sense that the index cannot be larger. One possible approach to verify the optimality is to examine explicit examples of torus partition functions, e.g.\,the one presented in \cite{Benjamin:2020mfz}.
\end{remark}

\subsection{Sketch of the proof}
To derive the two-sided asymptotic bounds \eqref{eq:resultFiniteEpsilon} for $\mathcal{A}_J(\beta_L,\varepsilon)$ in the DLC$_w$ limit, we introduce several tricks as follows.

\begin{figure}
	\centering
	\begin{tikzpicture}
		\draw[->] (0, 0) -- (3, 0) node[right] {$h$}; 
		\draw[->] (0, 0) -- (0, 3) node[above] {$\bar{h}$}; 
		\draw[pink] (1.5,0) -- (1.5,3);
		\draw (1.5,0) node[anchor=north] {$A$};
		\draw (1.5,0) circle[radius=2pt];
		\fill (1.5,0)  circle[radius=2pt];
		\draw[blue] (0,0) -- (3,3);
		\draw[blue] (1,0) -- (3,2);
		\draw[blue] (2,0) -- (3,1);
		\draw[blue] (0,1) -- (2,3);
		\draw[blue] (0,2) -- (1,3);
		\draw [red] (1.4,1.3) rectangle (1.6,1.7);
	\end{tikzpicture}
\begin{tikzpicture}
	\draw[->] (0, 0) -- (3, 0) node[right] {$h$}; 
	\draw[->] (0, 0) -- (0, 3) node[above] {$\bar{h}$}; 
	\draw[pink] (1.5,0) -- (1.5,3);
	\draw (1.5,0) node[anchor=north] {$A$};
	\draw (1.5,0) circle[radius=2pt];
	\fill (1.5,0)  circle[radius=2pt];
	\draw[blue] (0,0) -- (3,3);
	\draw[blue] (1,0) -- (3,2);
	\draw[blue] (2,0) -- (3,1);
	\draw[blue] (0,1) -- (2,3);
	\draw[blue] (0,2) -- (1,3);
	\draw [red] (1.4,0.65) rectangle (1.6,2.35);
\end{tikzpicture}
\caption{\label{fig:window1} Illustration of the idea behind equation \eqref{AJ:prop}. The blue lines represent the allowed positions of the spectrum, constrained by $h-\bar{h}\in\mathbb{Z}$. We aim to count the spectrum around the pink line ($h=A$). We choose two windows (shown in red) with the same width in $h$ but different widths in $\bar{h}$. Due to the integer-spin constraint, the spectrum inside the two windows is the same, as long as the windows intersect with only one of the blue lines.}
\end{figure}
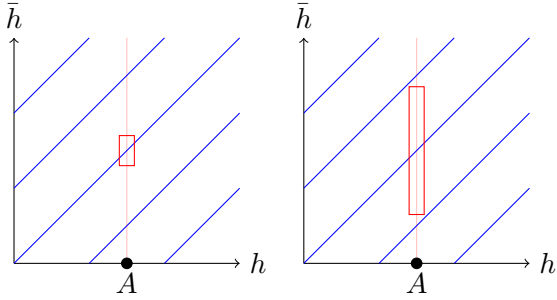
The first trick relies on the fact that only integer spins are allowed (here we only consider bosonic CFTs). This implies that the spectrum is empty for values of $h-\bar{h}$ that are non-integers. Using this property, we can express $\mathcal{A}_J$ in a different form as follows (see figure \ref{fig:window1} for a clearer visual representation):
\begin{equation}\label{AJ:prop}
	\begin{split}
		\mathcal{A}_{J}(\beta_L,\varepsilon)=\mathcal{A}(\beta_L,\bar{H},\varepsilon,\delta)\quad\forall\delta\in(\varepsilon,1-\varepsilon),
	\end{split}
\end{equation}
where
\begin{equation}\label{def:Hbar}
	\begin{split}
		\bar{H}\equiv A+J
	\end{split}
\end{equation}
and $\mathcal{A}(\beta_L,\bar{H},\varepsilon,\delta)$ is defined as
\begin{equation}\label{def:Acal}
	\begin{split}
		\mathcal{A}(\beta_L,\bar{H},\varepsilon,\delta):=&\int_{A-\varepsilon}^{A+\varepsilon}d h\int_{\bar{H}-\delta}^{\bar{H}+\delta} d\bar{h}\,\rho(h,\bar{h})e^{-(h-A)\beta_L}, \\
	\end{split}
\end{equation}
where $\rho(h,\bar{h})$ represents the spectral density of Virasoro primaries in the continuum-spectrum version of $\tilde{Z}$ given by the integral in eq.\,(\ref{Ztilde:integral}).

Now the problem is reduced to obtaining the upper and lower bounds for $\mathcal{A}(\beta_L,\bar{H},\varepsilon,\delta)$ in the DLC$_w$ limit. To achieve this, we express the DLC$_w$ limit \eqref{def:DLCwJ} in terms of $\beta_L$ and $\bar{H}$, taking into account that $\bar{H}=A+J$:
\begin{equation}\label{def:DLCdual}
	\begin{split}
		\mathrm{DLC}_w\ \mathrm{limit}:\quad&\beta_L\rightarrow\infty,\quad\bar{H}\rightarrow\infty,\quad \frac{2\pi T(1-w^2)}{A}\sqrt{\frac{\bar{H}-A}{A}}-\beta_L\rightarrow\infty, \\
		&\beta_L^{-1}\log\bar{H}\rightarrow0, \\
	\end{split}
\end{equation}
where $w$ is the same parameter introduced in \eqref{def:DLCwJ}. 

To proceed, we introduce the next trick which was used in \cite{Mukhametzhanov:2019pzy,Mukhametzhanov:2020swe}. Let us consider two functions $\phi_{\pm}(x)$ satisfying the inequality
\begin{equation}\label{phipm}
	\begin{split}
		\phi_-(x)\leqslant\theta_\delta(x)\leqslant\phi_+(x),\quad
		\theta_\delta(x):=\,\theta(x\in[-\delta,\delta]). 
	\end{split}
\end{equation}
In addition, for technical reasons, we require that $\phi_{\pm}$ are band-limited functions, meaning that their Fourier transforms $\hat{\phi}_\pm$ have compact support:
\begin{equation}\label{phipm:suppcondition}
	\begin{split}
		\phi_{\pm}(x)=&\int d t\  \hat{\phi}_\pm(t)e^{-i x t}, \\
		\mathrm{supp}(\hat{\phi}_\pm)\subset&[-\Lambda,\Lambda]\quad\mathrm{for\ some\ }\Lambda<2\pi w. \\
	\end{split}
\end{equation}
The functions satisfying these conditions exist \cite{Selberg1991}. Later, for the specific range of $w$ we are interested, we will give explicit expression for $\phi_{\pm}$, see \eqref{def:phipm} for $\Lambda=2\pi$ and \eqref{phipm:rescale} for any $\Lambda$. 

Here again, $w$ corresponds to the parameter in eq.\,(\ref{def:DLCdual}). The choice of $\Lambda<2\pi w$ will be clarified at the end of section \ref{section:modular:directnonvac}. By substituting eqs.\,(\ref{def:Acal}) and (\ref{phipm}) into the definition of $\mathcal{A}$, we obtain an upper bound for $\mathcal{A}$ given by:
\begin{equation}\label{A:upperbound}
	\begin{split}
		\mathcal{A}(\beta_L,\bar{H},\varepsilon,\delta)\leqslant&\int_{A-\varepsilon}^{A+\varepsilon}d h\int_0^\infty d\bar{h}\,\rho(h,\bar{h})e^{-(h-A)\beta_L+(\bar{H}+\delta-\bar{h})\beta_R}\theta_\delta(\bar{h}-\bar{H}) \\
		\leqslant&\int_{A-\varepsilon}^{A+\varepsilon}d h\int_0^\infty d\bar{h}\,\rho(h,\bar{h})e^{-(h-A)\beta_L+(\bar{H}+\delta-\bar{h})\beta_R}\phi_+(\bar{h}-\bar{H}) \\
		=&e^{(\bar{H}+\delta)\beta_R}\int_{A-\varepsilon}^{A+\varepsilon}d h\int_0^\infty d\bar{h}\,\rho(h,\bar{h})e^{-(h-A)\beta_L-\bar{h}\beta_R}\int d t \hat{\phi}_+(t)e^{-i(\bar{h}-\bar{H})t} \\
		=&e^{(\bar{H}+\delta-A)\beta_R}\int d t\ \tilde{Z}_{h\in(A-\varepsilon,A+\varepsilon)}(\beta_L,\beta_R+i t)\hat{\phi}_+(t)e^{i(\bar{H}-A)t}. \\
	\end{split}
\end{equation}
In the first line, we used $e^{(\bar{H}+\delta-\bar{h})\beta_R}\geqslant1$ in the support of $\theta_\delta(\bar{h}-\bar{H})$. In the second line, we bounded $\theta_\delta$ by $\phi_+$. In the third line, we rewrote $\phi_+$ as the Fourier transform of $\hat{\phi}_+$. Finally, in the last line we used the definition of $\tilde{Z}_{h\in(A-\varepsilon,A+\varepsilon)}$. 

Similarly, we have the following lower bound for $\mathcal{A}$:
\begin{equation}\label{A:lowerbound}
	\begin{split}
		\mathcal{A}(\beta_L,\bar{H},\varepsilon,\delta)\geqslant e^{(\bar{H}-\delta-A)\beta_R}\int d t\ \tilde{Z}_{h\in(A-\varepsilon,A+\varepsilon)}(\beta_L,\beta_R+i t)\hat{\phi}_-(t)e^{i(\bar{H}-A)t},
	\end{split}
\end{equation}
where $\hat{\phi}_{-}$ is the Fourier transform of $\phi_{-}$. It is worth noting that although the bounds depend on $\beta_R$, the quantity $\mathcal{A}$ itself does not. The final result, given by eq.\,\eqref{eq:resultFiniteEpsilon}, will be obtained by selecting an appropriate value for $\beta_R$. Here we choose $\beta_{R}$ to be\footnote{We will provide a technical explanation for this choice in appendix \ref{app:Ivac:asym}. Additionally, an intuitive argument will be presented in section \ref{section:modulardual}.}
\begin{equation}\label{choice:betaRHbar}
	\begin{split}
		\beta_R=2\pi\sqrt{\frac{A}{\bar{H}-A}}.
	\end{split}
\end{equation}
With this choice, the limit \eqref{def:DLCdual} can be expressed as:
\begin{equation}\label{def:DLCwlimit}
	\begin{split}
		\text{DLC$_w$ limit:}\quad &\beta_L \rightarrow \infty,
		\quad \beta_R \rightarrow 0,
		\quad \mathfrak{b}_w(\beta_L,\beta_R):=\frac{4\pi^2T(1-w^2)}{A\beta_R} -\beta_L\rightarrow \infty, \\
		&\beta_L^{-1}\log\beta_{R}\rightarrow0. \\
	\end{split}
\end{equation}
We observe that \eqref{def:DLCwlimit} is slightly stronger than \eqref{def:DLClimit}: the inclusion of the $w^2$ term in the third equation of \eqref{def:DLCwlimit} is sufficient to eliminate the logarithmic term $\log(\beta_{L})$ present in the third equation of \eqref{def:DLClimit}. The last equation in \eqref{def:DLCwlimit} is introduced for technical reasons.

From now on, we will \textbf{always} assume (\ref{def:Hbar}) and (\ref{choice:betaRHbar}) by default. Consequently, the three formulations of the DLC$_w$ limit, namely (\ref{def:DLCwJ}), (\ref{def:DLCdual}) and (\ref{def:DLCwlimit}), are equivalent. 

In the DLC$_w$ limit, the exponential prefactor $e^{(\bar{H}-A\pm\delta)\beta_{R}}$ in the upper bound \eqref{A:upperbound} and the lower bound \eqref{A:lowerbound} coincide because $\beta_{R}\delta\rightarrow0$. So we have
\begin{equation}
	\begin{split}
		I_{-, h\in(A-\varepsilon,A+\varepsilon)}(A,\bar{H};\beta_{L},\beta_{R})\stackrel{{\rm DLC}_w}{\lesssim}\frac{\mathcal{A}(\beta_{L},\bar{H},\varepsilon,\delta)}{e^{2\pi\sqrt{A(\bar{H}-A)}}}\stackrel{{\rm DLC}_w}{\lesssim}I_{+, h\in(A-\varepsilon,A+\varepsilon)}(A,\bar{H};\beta_{L},\beta_{R})
	\end{split}
\end{equation}
where $I_{\pm}$ are defined in the following way:
\begin{equation}\label{modular:toestimate}
	\begin{split}
		I_{\pm, h\in(A-\varepsilon,A+\varepsilon)}(A,\bar{H};\beta_{L},\beta_{R}):=\int d t\ \tilde{Z}_{h\in(A-\varepsilon,A+\varepsilon)}(\beta_L,\beta_R+i t)\hat{\phi}_\pm(t)e^{i(\bar{H}-A)t}.
	\end{split}
\end{equation}
So our goal reduces to deriving asymptotic behavior of the integral \eqref{modular:toestimate} in the DLC$_w$ limit. To do this, the main idea is to demonstrate that the asymptotic behavior of the integral \eqref{modular:toestimate} remains unchanged in the DLC$_w$ limit if we replace $\tilde{Z}_{h\in(A-\varepsilon,A+\varepsilon)}$ with the vacuum term in the dual channel of the full reduced partition function $\tilde{Z}$:
\begin{equation}
	\begin{split}
		\lim\limits_{{{\mathrm{DLC}_w}}}\frac{\int d t\ \tilde{Z}_{h\in(A-\varepsilon,A+\varepsilon)}(\beta_L,\beta_R+i t)\hat{\phi}_\pm(t)e^{i(\bar{H}-A)t}}{\int d t\ \sqrt{\frac{4\pi^2}{\beta_L(\beta_R+i t)}}\tilde{Z}_{\rm vac}\left(\frac{4\pi^2}{\beta_L},\frac{4\pi^2}{\beta_R+i t}\right)\hat{\phi}_\pm(t)e^{i(\bar{H}-A)t}}=1.\label{modular:idea}
	\end{split}
\end{equation}
To see this, let us consider the integral
\begin{equation}\label{Ipm:direct}
	\begin{split}
		I_\pm(A,\bar{H};\beta_{L},\beta_{R}):=\int d t\ \tilde{Z}(\beta_L,\beta_R+i t)\hat{\phi}_\pm(t)e^{i(\bar{H}-A)t}.
	\end{split}
\end{equation}
Using modular invariance, we evaluate this integral in the dual channel:
\begin{equation}\label{Ipm:dual}
	\begin{split}
		I_\pm(A,\bar{H};\beta_{L},\beta_{R})=\int d t\ \sqrt{\frac{4\pi^2}{\beta_L(\beta_R+it)}}\tilde{Z}\left(\frac{4\pi^2}{\beta_{L}},\frac{4\pi^2}{\beta_R+i t}\right)\hat{\phi}_\pm(t)e^{i(\bar{H}-A)t}.
	\end{split}
\end{equation}
Now we split $I_\pm$ in different ways in the two channels. Using equations \eqref{Ipm:direct} and \eqref{Ipm:dual}, we have:
\begin{equation}\label{modularinv:split}
	\begin{split}
		I_{\pm,\rm vac}+I_{\pm,T\leqslant h\leqslant A-\varepsilon}&+I_{\pm,A-\varepsilon<h<A+\epsilon}+I_{\pm,h\geqslant A+\epsilon}=I^{\rm dual}_{\pm,\rm vac}+I^{\rm dual}_{\pm,\rm nonvac}, \\
		({\rm direct}\ &{\rm channel})\qquad\qquad\qquad\qquad\qquad({\rm dual\ channel}) \\
	\end{split}
\end{equation}
where the integrals are defined as follows:
\begin{equation}\label{def:allIpm}
	\begin{split}
		I_{\pm,\rm vac}=&\int d t\ \tilde{Z}_{\rm vac}(\beta_L,\beta_R+i t)\hat{\phi}_\pm(t)e^{i(\bar{H}-A)t}, \\
		I_{\pm,T\leqslant h\leqslant A-\varepsilon}=&\int d t\ \tilde{Z}_{T\leqslant h\leqslant A-\varepsilon}(\beta_L,\beta_R+i t)\hat{\phi}_\pm(t)e^{i(\bar{H}-A)t}, \\
		I_{\pm,h\in(A-\varepsilon,A+\varepsilon)}=&\int d t\ \tilde{Z}_{h\in(A-\varepsilon,A+\varepsilon)}(\beta_L,\beta_R+i t)\hat{\phi}_\pm(t)e^{i(\bar{H}-A)t}, \\
		I_{\pm,h\geqslant A+\epsilon}=&\int d t\ \tilde{Z}_{h\geqslant A+\varepsilon}(\beta_L,\beta_R+i t)\hat{\phi}_\pm(t)e^{i(\bar{H}-A)t}, \\
		I^{\rm dual}_{\pm,\rm vac}=&\int d t \sqrt{\frac{4\pi^2}{\beta_L(\beta_R+i t)}}\tilde{Z}_{\rm vac}\left(\frac{4\pi^2}{\beta_L},\frac{4\pi^2}{\beta_R+i t}\right)\hat{\phi}_\pm(t)e^{i(\bar{H}-A)t}, \\
		I^{\rm dual}_{\pm,\rm nonvac}=&\int d t \sqrt{\frac{4\pi^2}{\beta_L(\beta_R+i t)}}\tilde{Z}_{h,\bar{h}\geqslant T}\left(\frac{4\pi^2}{\beta_L},\frac{4\pi^2}{\beta_R+i t}\right)\hat{\phi}_\pm(t)e^{i(\bar{H}-A)t}. \\ 
	\end{split}
\end{equation}
Here and below, $I^{\ldots}_{\ldots}$ always refers to $I^{\ldots}_{\ldots}(A,\bar{H};\beta_{L},\beta_{R})$ with the identification $\beta_{R}=2\pi\sqrt{\frac{A}{\bar{H}-A}}$.

In section \ref{section:modulardual}, we will demonstrate that in the DLC$_w$ limit, the dual channel is dominated by $I^{\rm dual}_{\pm,\rm vac}$. This is equivalent to say that
\begin{equation}\label{modular:sketch1}
	\begin{split}
		\lim\limits_{{\mathrm{DLC}_w}}\frac{I^{\rm dual}_{\pm,\rm nonvac}}{I^{\rm dual}_{\pm,\rm vac}}=0.
	\end{split}
\end{equation}
Moving on to section \ref{section:modulardirect}, we will establish that in the DLC$_w$ limit, the direct channel is dominated by $I_{\pm,h\in(A-\varepsilon,A+\varepsilon)}$. Since the vacuum term dominates the dual channel in the DLC$_w$ limit, this is equivalent to say that
\begin{equation}\label{modular:sketch2}
	\begin{split}
		\lim\limits_{\mathrm{DLC}_w}\frac{\left(I_{\pm,\rm vac}+I_{\pm,T\leqslant h\leqslant A-\varepsilon}+I_{\pm,h\geqslant A+\varepsilon}\right)}{I^{\rm dual}_{\pm,\rm vac}}=0.
	\end{split}
\end{equation}
We will establish that each term in the numerator of (\ref{modular:sketch2}) is suppressed by the denominator $I^{\rm dual}_{\pm,\rm vac}$. 

Subsequently, eq.\,(\ref{modular:idea}) follows from eqs.\,(\ref{modularinv:split}), (\ref{def:allIpm}), (\ref{modular:sketch1}) and (\ref{modular:sketch2}).  We can then evaluate the denominator of eq.\,(\ref{modular:idea}), which corresponds to $I^{\rm dual}_{\pm,\rm vac}$, using its precise expression. The result of $I^{\rm dual}_{\pm,\rm vac}$ will be given in section \ref{modular:crossvac}, and the technical details will be given in appendix \ref{app:Ivac:asym}. It provides us with the desired estimate of the asymptotic upper and lower bounds on $\mathcal{A}(\beta_L,\bar{H},\varepsilon,\delta)$ in the DLC$_w$ limit (see section \ref{section:modular:tauberian}).

In section \ref{section:epsilonwindow}, we will demonstrate that, while satisfying the aforementioned estimates, $\varepsilon$ can effectively approach zero in the DLC$_w$ limit, as long as it remains bounded from below by $\varepsilon_{\rm min}(\beta_{L},J)$ (defined in eq.\,(\ref{epsilon:choice:DLCJ})). This justification will support the final part of theorem \ref{theorem:modulartauberianFJ}.

Returning to $\mathcal{A}_J(\beta_{L},\varepsilon)$ using eq.\,(\ref{AJ:prop}), we note that the spectral density $\rho$ is positive, implying that $\mathcal{A}(\beta_L,\bar{H},\varepsilon,\delta)$ is monotonically increasing in $\delta$. To obtain optimal bounds for $\mathcal{A}_J$, we choose the smallest $\delta$ for the upper bound on $\mathcal{A}$ and the largest $\delta$ for the lower bound on $\mathcal{A}$, yielding
\begin{equation}\label{AJ:twosideinA}
	\begin{split}
		\lim\limits_{\delta\rightarrow(1-\varepsilon)^{-}}\mathcal{A}(\beta_L,\bar{H},\varepsilon,\delta)\leqslant\mathcal{A}_{J}(\beta_L,\varepsilon)\leqslant\lim\limits_{\delta\rightarrow\varepsilon^{+}}\mathcal{A}(\beta_L,\bar{H},\varepsilon,\delta).
	\end{split}
\end{equation}
These inequalities provide the two-sided bounds stated in eq.\,(\ref{eq:resultFiniteEpsilon}).

\subsection{Dual channel}\label{section:modulardual}
\subsubsection{Dual channel: vacuum}
Consider the vacuum part of $I_\pm$ in the dual channel
\begin{equation}\label{int:dualvac}
	\begin{split}
		I^{\rm dual}_{\pm,\rm vac}\equiv&\int d t \sqrt{\frac{4\pi^2}{\beta_L(\beta_R+i t)}}\tilde{Z}_{\rm vac}\left(\frac{4\pi^2}{\beta_L},\frac{4\pi^2}{\beta_R+i t}\right)\hat{\phi}_\pm(t)e^{i(\bar{H}-A)t}, \\
	\end{split}
\end{equation}
where $\tilde{Z}_{\rm vac}$ is the vacuum part of the partition function, given by
\begin{equation}
	\begin{split}
		\tilde{Z}_{\rm vac}\left(\beta,\bar{\beta}\right)=e^{A\left(\beta+\bar{\beta}\right)}\left(1-e^{-\beta}\right)\left(1-e^{-\bar{\beta}}\right).
	\end{split}
\end{equation}
By definition $I^{\rm dual}_{\pm,\rm vac}$ depends on $A$, $\bar{H}$, $\beta_L$ and $\beta_R$. We will choose a proper $\beta_R$ to optimize the asymptotic behavior of $I^{\rm dual}_{\pm,\rm vac}$ in the limit (\ref{def:DLCdual}). As mentioned in section \ref{section:modular:idea}, here we choose $\beta_R=2\pi\sqrt{\frac{A}{\bar{H}-A}}$. We leave the technical reason of this choice to appendix \ref{app:Ivac:asym}. Here we would like to give the following intuitive explanation why this is a good choice. 

For simplicity let us fix $\beta_L$ and take the limit $\beta_R\rightarrow0$. Using modular invariance, one can show that $\tilde{Z}(\beta_L,\beta_{R})$ is dominated by the vacuum term in the dual channel:
\begin{equation}
	\begin{split}
		\tilde{Z}(\beta_L,\beta_{R})=&\sqrt{\frac{4\pi^2}{\beta_L\beta_{R}}}\tilde{Z}_{vac}\left(\frac{4\pi^2}{\beta_L},\frac{4\pi^2}{\beta_R}\right)\left[1+O\left(e^{-\frac{4\pi^2 T}{\beta_{R}}}\right)\right] \\
		\sim&\frac{8\pi^3}{\beta_L^{3/2}\beta_R^{1/2}} e^{\frac{4\pi^2A}{\beta_R}}, \\
	\end{split}
\end{equation}
where $T$ is the twist gap. To reproduce this asymptotic behavior, we make a naive guess on the large-$\bar{h}$ behavior of the spectral density $\rho(h,\bar{h})$ by performing inverse Laplace transform on $\beta_{R}^{-1/2}e^{\frac{4\pi^2A}{\beta_{R}}}$:
\begin{equation}
		\begin{split}
			\rho(h,\bar{h})\stackrel{guess}{\sim} F(h)\frac{e^{4\pi\sqrt{A(\bar{h}-A)}}}{\sqrt{\bar{h}-A}} \quad(\bar{h}\rightarrow\infty).
		\end{split}
\end{equation}
This statement can actually be proven in a rigorous way using the argument in \cite{Mukhametzhanov:2019pzy}. Then we have
\begin{equation}
		\begin{split}
			\tilde{Z}(\beta_L,\beta_R)\equiv&\int d h d\bar{h}\,\rho(h,\bar{h}) e^{(A-h)\beta_L+(A-\bar{h})\beta_R} \\
			\sim&\mathcal{L}(F)(\beta_L)e^{A\beta_L}\int d\bar{h}\frac{e^{4\pi\sqrt{A(\bar{h}-A)}}}{\sqrt{\bar{h}-A}}e^{(A-\bar{h})\beta_R}\quad(\beta_R\rightarrow0),
		\end{split}
\end{equation}
where $\mathcal{L}(F)$ is the Laplace transform of $F$. Now let us focus on the $\beta_R$-related part:
\begin{equation}
		\begin{split}
			\int d\bar{h}\frac{e^{4\pi\sqrt{A(\bar{h}-A)}}}{\sqrt{\bar{h}-A}}e^{(A-\bar{h})\beta_R}=2e^{\frac{4\pi^2A}{\beta_R}}\int dx\ e^{-\beta_R\left(x-\frac{2\pi\sqrt{A}}{\beta_R}\right)^2}.
		\end{split}
\end{equation}
Here, we introduce the variable change $x=\sqrt{\bar{h}-A}$. Observing the integrand, we notice that it reaches its maximum value at $x=\frac{2\pi\sqrt{A}}{\beta_R}$, which implies $\beta_{R}=2\pi\sqrt{\frac{A}{\bar{h}-A}}$ (see figure \ref{figure:window}). As we aim to extract information about the spectrum within the window $\bar{h}\in(\bar{H}-\delta,\bar{H}+\delta)$, it seems natural to choose eq.\,(\ref{choice:betaRHbar}) as the relation between $\beta_R$ and $\bar{H}$. This completes the intuitive explanation.
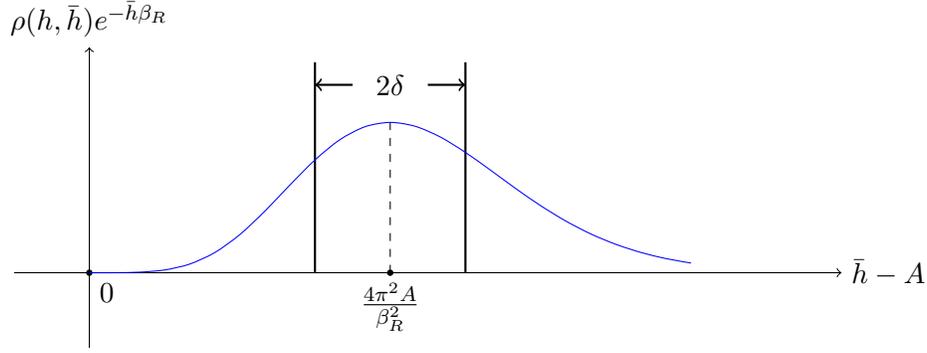
\begin{figure}
	\centering
	\begin{tikzpicture}
		\draw[->] (-1, 0) -- (10, 0) node[right] {$\bar{h}-A$}; 
		\draw[->] (0, -1) -- (0, 3) node[above] {$\rho(h,\bar{h})e^{-\bar{h}\beta_R}$}; 
		\draw[fill] (0,0) circle [radius=1pt] node[below right] {$0$}; 
		\draw[dashed] (4,0) -- (4,2); 
		\draw[thick,black] (3,0) -- (3,2.8); 
		\draw[thick,black] (5,0) -- (5,2.8); 
		\draw[thick , ->] (3.5,2.5) -- (3,2.5); 
		\draw[thick , ->] (4.5,2.5) -- (5,2.5); 
		\node (A) at (4,2.5) {$2\delta$}; 
		\draw[fill] (4,0) circle [radius=1pt] node[below] {$\frac{4\pi^2A}{\beta_R^2}$}; 
		\draw[scale=1, domain=0:8, smooth, variable=\x, blue] plot ({\x},{2*exp(-4*pow(pow(\x,1/2)-2,2))}); 
	\end{tikzpicture}
	\caption{A typical behavior of $\rho(h,\bar{h})e^{-\bar{h}\beta_{R}}$ (with $h$ fixed) at small $\beta_R$.}
	\label{figure:window}
\end{figure}

After identifying $\beta_{R}$ and $\bar{H}$ using constraint (\ref{choice:betaRHbar}), we have\footnote{By $A\sim B$ we mean that $\frac{A}{B}\rightarrow1$ in the considered limit.}
\begin{equation}
	\begin{split}
		I^{\rm dual}_{\pm,\rm vac}\sim \left(\dfrac{4\pi^2}{\beta_L}\right)^{3/2}\sqrt{\frac{\pi}{\bar{H}}}e^{2\pi\sqrt{A\bar{H}}}\hat{\phi}_\pm(0)  \label{Ipmvac:asymp:inH}
	\end{split}
\end{equation}
in the DLC$_w$ limit. We leave the detailed derivation of eq.\,(\ref{Ipmvac:asymp:inH}) to appendix \ref{app:Ivac:asym}. In order to compare with the contribution from other parts in (\ref{modularinv:split}), it is convenient to rewrite the asymptotic behavior of $I^{\rm dual}_{\pm,\rm vac}$ as
\begin{equation}
	\begin{split}
		I^{\rm dual}_{\pm,\rm vac}\stackrel{{\rm DLC}_w}{\sim} \frac{4\pi^{5/2}\beta_R}{\beta_L^{3/2}A^{1/2}}e^{\frac{4\pi^2A}{\beta_R}}\hat{\phi}_\pm(0).  \label{Ipmvac:asymp}
	\end{split}
\end{equation}

\subsubsection{Dual channel: non-vacuum}\label{section:modular:directnonvac}
Consider the non-vacuum part of $I_\pm$ in the dual channel, defined as follows:
\begin{equation}\label{def:dualnonvac}
	\begin{split}
		I^{\rm dual}_{\pm,\rm nonvac}\equiv&\int d t \sqrt{\frac{4\pi^2}{\beta_L(\beta_R+i t)}}\tilde{Z}_{h,\bar{h}\geqslant T}\left(\frac{4\pi^2}{\beta_L},\frac{4\pi^2}{\beta_R+i t}\right)\hat{\phi}_\pm(t)e^{i(\bar{H}-A)t}, \\
	\end{split}
\end{equation}
where $\tilde{Z}_{h,\bar{h}\geqslant T}$ is defined by
\begin{equation}
	\begin{split}
		\tilde{Z}_{h,\bar{h}\geqslant T}\left(\beta,\bar{\beta}\right)=\int_{T}^{\infty}d h\int_{T}^{\infty} d\bar{h}\,\rho(h,\bar{h})e^{-(h-A)\beta-(\bar{h}-A)\bar{\beta}}.
	\end{split}
\end{equation}
For technical reasons, we split $I^{\rm dual}_{\pm,\rm nonvac}$ into two parts
\begin{equation}
	\begin{split}
		I^{\rm dual}_{\pm,\rm nonvac}=I^{\rm dual}_{\pm,T\leqslant\bar{h}<A}+I^{\rm dual}_{\pm, \bar{h}\geqslant A},
	\end{split}
\end{equation}
where the subscripts denote the regimes of $(h,\bar{h})$ that contribute. We begin with the following inequality:
\begin{equation}
	\begin{split}
		\abs{I^{\rm dual}_{\pm,\rm nonvac}}\leqslant&\sqrt{\frac{4\pi^2}{\beta_L\beta_R}}\max\limits_{x}\abs{\hat{\phi}_\pm(x)}\int_{-\Lambda}^{\Lambda} d t \tilde{Z}_{h,\bar{h}\geqslant T}\left(\frac{4\pi^2}{\beta_L},\frac{4\pi^2\beta_R}{\beta_R^2+t^2}\right) \\
		=&\sqrt{\frac{4\pi^2}{\beta_L\beta_R}}\max\limits_{x}\abs{\hat{\phi}_\pm(x)} \\
		&\times\int_{-\Lambda}^{\Lambda} d t \left[\tilde{Z}_{T\leqslant \bar h<A}\left(\frac{4\pi^2}{\beta_L},\frac{4\pi^2\beta_R}{\beta_R^2+t^2}\right)+\tilde{Z}_{ \bar h\geqslant A}\left(\frac{4\pi^2}{\beta_L},\frac{4\pi^2\beta_R}{\beta_R^2+t^2}\right)\right] \\
		\leqslant&\sqrt{\frac{4\pi^2}{\beta_L\beta_R}}\max\limits_{x}\abs{\hat{\phi}_\pm(x)} \\
		&\times2\Lambda\left[\tilde{Z}_{T\leqslant \bar h<A}\left(\frac{4\pi^2}{\beta_L},\frac{4\pi^2}{\beta_R}\right)+\tilde{Z}_{\bar h\geqslant A}\left(\frac{4\pi^2}{\beta_L},\frac{4\pi^2\beta_R}{\beta_R^2+\Lambda^2}\right)\right]. \\
	\end{split}
\end{equation}
Here in the first line, we use the inequality $\abs{\sqrt{\frac{1}{\beta_R+i t}}}\leqslant\sqrt{\frac{1}{\beta_R}}$, the identity $\abs{e^z}=e^{\mathrm{Re}(z)}$, and the fact that $\rm supp(\hat{\phi}_\pm)\subset[-\Lambda,\Lambda]$. In the second line, we split $\tilde{Z}_{h,\bar{h}\geqslant T}$ into two parts. In the last line, we make use of the inequalities $e^{-(\bar{h}-A)\frac{4\pi^2\beta_R}{\beta_R^2+t^2}}\leqslant e^{-(\bar{h}-A)\frac{4\pi^2}{\beta_R}}$ for $\bar{h}<A$ and $e^{-(\bar{h}-A)\frac{4\pi^2\beta_R}{\beta_R^2+t^2}}\leqslant e^{-(\bar{h}-A)\frac{4\pi^2\beta_R}{\beta_R^2+\Lambda^2}}$ for $\bar{h}\geqslant A$ and $\abs{t}\leqslant\Lambda$. Here, we use the maximum notation $\max_{x}$ to indicate that we are taking the maximum of $\abs{\hat{\phi}_\pm(x)}$ over all $x$. 

The estimate above is for $I^{\rm dual}_{\pm,\rm nonvac}$, but it is straightforward to obtain the following individual inequalities for $I^{\rm dual}_{\pm,T\leqslant\bar{h}<A}$ and $I^{\rm dual}_{\pm, \bar{h}\geqslant A}$:
\begin{equation}\label{estimate:Inonvac}
	\begin{split}
		\abs{I^{\rm dual}_{\pm,T\leqslant\bar{h}<A}}\leqslant&2\Lambda \sqrt{\frac{4\pi^2}{\beta_L\beta_R}}\max\limits_{x}\abs{\hat{\phi}_\pm(x)} \tilde{Z}_{T\leqslant \bar h<A}\left(\frac{4\pi^2}{\beta_L},\frac{4\pi^2}{\beta_R}\right), \\
		\abs{I^{\rm dual}_{\pm,\bar h\geqslant A}}\leqslant&2\Lambda \sqrt{\frac{4\pi^2}{\beta_L\beta_R}}\max\limits_{x}\abs{\hat{\phi}_\pm(x)} \tilde{Z}_{\bar h\geqslant A}\left(\frac{4\pi^2}{\beta_L},\frac{4\pi^2\beta_{R}}{\beta_R^2+\Lambda^2}\right). \\
	\end{split}
\end{equation}
To demonstrate that $I^{\rm dual}_{\pm,T\leqslant\bar{h}<A}$ and $I^{\rm dual}_{\pm, \bar{h}\geqslant A}$ are suppressed by $I^{\rm dual}_{\pm,\rm vac}$ in the DLC$_w$ limit, we can establish some upper bounds on $\tilde{Z}_{T\leqslant \bar h<A}$ and $\tilde{Z}_{\bar h\geqslant A}$. We present a useful lemma below:
\begin{lemma}\label{lemma:Ztildebound}
	Let $\beta_0\in(0,\infty)$ be a fixed number. The partition function satisfies the following upper bound:
	\begin{equation}
		\begin{split}
			\tilde{Z}(\beta_L,\beta_R)\leqslant&\kappa(\beta_0)e^{A(\beta_L+\beta_R)}\quad(\beta_L,\beta_R\geqslant\beta_0), \\
		\end{split}
	\end{equation}
where
\begin{equation}\label{def:kappa0}
	\begin{split}
		\kappa(\beta_0)\equiv1+\frac{\tilde{Z}(\beta_0,\beta_0)}{\tilde{Z}_{\rm vac}(\beta_0,\beta_0)}.
	\end{split}
\end{equation}
\end{lemma}
\begin{proof}
	For $\beta_L,\beta_R\geqslant\beta_0$ we have
	\begin{equation}
		\begin{split}
			\tilde{Z}(\beta_L,\beta_R)\leqslant&\tilde{Z}_{\rm vac}(\beta_L,\beta_R)+\tilde{Z}_{\rm nonvac}(\beta_L,\beta_R) \\
			\leqslant&\tilde{Z}_{\rm vac}(\beta_L,\beta_R)+e^{(A-T)(\beta_L+\beta_R-2\beta_0)}\tilde{Z}_{\rm nonvac}(\beta_0,\beta_0), \\
		\end{split}
	\end{equation}
where we split the partition function into the vacuum part and the non-vacuum part and used the fact that $h,\bar{h}\geqslant T$ for each term in the non-vacuum part. 

The vacuum part is bounded as follows:
\begin{equation}\label{bound:Zvac}
	\begin{split}
		\tilde{Z}_{\rm vac}(\beta_L,\beta_R)\equiv e^{A(\beta_L+\beta_R)}\left(1-e^{-\beta_L}\right)\left(1-e^{-\beta_R}\right)\leqslant e^{A(\beta_L+\beta_R)}.
	\end{split}
\end{equation}
We also have
\begin{equation}
	\begin{split}
		\tilde{Z}_{\rm nonvac}(\beta_0,\beta_0)\leqslant\frac{\tilde{Z}(\beta_0,\beta_0)}{\tilde{Z}_{\rm vac}(\beta_0,\beta_0)}\tilde{Z}_{\rm vac}(\beta_0,\beta_0)\leqslant\frac{\tilde{Z}(\beta_0,\beta_0)}{\tilde{Z}_{\rm vac}(\beta_0,\beta_0)} e^{2A\beta_0},
	\end{split}
\end{equation}
where we bounded $\tilde{Z}_{\rm nonvac}$ by the full partition function $\tilde{Z}$ and used (\ref{bound:Zvac}).

Putting everything together we get
\begin{equation}
	\begin{split}
		\tilde{Z}(\beta_L,\beta_R)\leqslant e^{A(\beta_L+\beta_R)}\left[1+e^{-T(\beta_L+\beta_R-2\beta_0)}\frac{\tilde{Z}(\beta_0,\beta_0)}{\tilde{Z}_{\rm vac}(\beta_0,\beta_0)}\right]\leqslant\kappa(\beta_0)e^{A(\beta_L+\beta_R)}
	\end{split}
\end{equation}
for $\beta_L,\beta_R\geqslant\beta_0$, where $\kappa(\beta_0)$ is given by (\ref{def:kappa0}). This completes the proof.
\end{proof}
Consider the expression $\tilde{Z}_{T\leqslant \bar {h}<A}\left(\frac{4\pi^2}{\beta_L},\frac{4\pi^2}{\beta_R}\right)$ in the regime where $\beta_R\leqslant\beta_0\leqslant\beta_L$. We have the following inequalities:
\begin{equation}\label{Ztilde:dualLessA:derivation}
	\begin{split}
		\tilde{Z}_{T\leqslant \bar h<A}\left(\frac{4\pi^2}{\beta_L},\frac{4\pi^2}{\beta_R}\right)\leqslant&e^{(A-T)\left(\frac{4\pi^2}{\beta_R}-\frac{4\pi^2}{\beta_0}\right)}\tilde{Z}_{T\leqslant \bar h<A}\left(\frac{4\pi^2}{\beta_L},\frac{4\pi^2}{\beta_0}\right) \\
		\leqslant&e^{(A-T)\left(\frac{4\pi^2}{\beta_R}-\frac{4\pi^2}{\beta_0}\right)}\sqrt{\frac{\beta_L\beta_0}{4\pi^2}}\tilde{Z}(\beta_L,\beta_0) \\
		\leqslant&\kappa(\beta_0)\sqrt{\frac{\beta_L\beta_0}{4\pi^2}}e^{(A-T)\left(\frac{4\pi^2}{\beta_R}-\frac{4\pi^2}{\beta_0}\right)+A(\beta_L+\beta_0)},\\
	\end{split}
\end{equation}
In the first line, we use the fact that $e^{(A-\bar{h})\frac{4\pi^2}{\beta_R}}\leqslant e^{(A-T)\left(\frac{4\pi^2}{\beta_R}-\frac{4\pi^2}{\beta_0}\right)}e^{(A-\bar{h})\frac{4\pi^2}{\beta_0}}$ for $h\geqslant T$ and $\beta_R\leqslant\beta_0$. In the second line, we bound $\tilde{Z}_{T\leqslant \bar h<A}$ by the full partition function and use modular invariance. Finally, in the last line, we use Lemma \ref{lemma:Ztildebound}.

Using (\ref{Ipmvac:asymp}), \eqref{Ztilde:dualLessA:derivation} and the first inequality of \eqref{estimate:Inonvac}, we obtain the following asymptotic inequality in the DLC$_w$ limit:
\begin{equation}
	\begin{split}\label{eq:dualLessA}
		\abs{\frac{I^{\rm dual}_{\pm,T\leqslant\bar{h}<A}}{I^{\rm dual}_{\pm,{\rm vac}}}}\stackrel{{\rm DLC}_w}{\lesssim}&C^{(1)}_\pm(\beta_0,A,T)\left(\frac{\beta_L}{\beta_{R}}\right)^{3/2}e^{-\frac{4\pi^2T}{\beta_R}+A\beta_L}, \\
	\end{split}
\end{equation}
where $C^{(1)}\pm(\beta_0,A,T)$ is a finite constant for fixed $\beta_0$, $A$, and $T$, given by
\begin{equation}\label{def:C1}
	\begin{split}
		C^{(1)}_\pm(\beta_0,A,T)\equiv& \frac{\kappa(\beta_0) A^{1/2}\beta_0^{1/2}\Lambda}{2\pi^{5/2}}\max\limits_{x}\abs{\frac{\hat{\phi}_\pm(x)}{\hat{\phi}_\pm(0)}}e^{-A\left((1-\frac{T}{A})\frac{4\pi^2}{\beta_0}-\beta_0\right)}. \\
	\end{split}
\end{equation}
We can apply the DLC$_w$ limit (defined in (\ref{def:DLCwlimit})) to the remaining factor in the r.h.s. of (\ref{eq:dualLessA}), obtaining:
\begin{equation}
\begin{split}\label{eq:dualLessAprime}
\left(\frac{\beta_L}{\beta_{R}}\right)^{3/2}e^{-\frac{4\pi^2T}{\beta_R}+A\beta_L}& = e^{-\frac{4\pi^2T}{\beta_R}(1-w^2)+A\beta_L}\times\left(\frac{\beta_L}{\beta_{R}}\right)^{3/2}e^{-\frac{4\pi^2w^2T}{\beta_R}}\\
&\leqslant \left[e^{-\frac{4\pi^2T}{\beta_R}(1-w^2)+A\beta_L}\right]\left[\left(4\pi^2\right)^{3/2}\beta_R^{-3}e^{-\frac{4\pi^2w^2T}{\beta_R}}\right]\quad (\beta_L\beta_R \leqslant 4\pi^2) \\
&\stackrel{{\rm DLC}_w}{\longrightarrow}0. \\
\end{split}
\end{equation}
Here we used the fact that $\beta_L\beta_R\leqslant4\pi^2$ eventually holds in the DLC$_w$ limit. As a result, we conclude that $I^{\rm dual}_{\pm,T\leqslant\bar{h}<A}$ is suppressed by $I^{\rm dual}_{\pm,{\rm vac}}$ in the DLC$_w$ limit.

Then let us consider $\tilde{Z}_{\bar h\geqslant A}\left(\frac{4\pi^2}{\beta_L},\frac{4\pi^2\beta_R}{\beta_R^2+\Lambda^2}\right)$ in the regime $\beta_L,\frac{\Lambda^2}{\beta_R}\geqslant\beta_0\ {\rm and}\ \beta_R\leqslant\Lambda$. We have the following bound:
\begin{equation}\label{Ztilde:dualGreatA:derivation}
	\begin{split}
		\tilde{Z}_{\bar h\geqslant A}\left(\frac{4\pi^2}{\beta_L},\frac{4\pi^2\beta_R}{\beta_R^2+\Lambda^2}\right)\leqslant&\sqrt{\frac{\beta_L(\beta_R^2+\Lambda^2)}{4\pi^2\beta_R}}\tilde{Z}\left(\beta_L,\beta_R+\frac{\Lambda^2}{\beta_R}\right) \\
		\leqslant&\kappa(\beta_0)\sqrt{\frac{\beta_L\Lambda^2}{2\pi^2\beta_R}}e^{A\left(\beta_L+\beta_R+\frac{\Lambda^2}{\beta_R}\right)}. \\
	\end{split}
\end{equation}
Here in the first line we bounded $\tilde{Z}_{\bar h\geqslant A}$ by the full partition function and used modular invariance, and in the second line we applied lemma \ref{lemma:Ztildebound} in the specific regime of $(\beta_L,\beta_R)$. By using the above bound and (\ref{Ipmvac:asymp}), we obtain the following estimate: 
\begin{equation}
	\begin{split}\label{eq:dualGreatA}
		\abs{\frac{I^{\rm dual}_{\pm,\bar h\geqslant A}}{I^{\rm dual}_{\pm,{\rm vac}}}}\stackrel{{\rm DLC}_w}{\lesssim}&C^{(2)}_\pm(\beta_0,A,T)\frac{\beta_L^{3/2}}{\beta^2_R} e^{-A\left(\frac{4\pi^2-\Lambda^2}{\beta_R}-\beta_L-\beta_R\right)}, \\
		C^{(2)}_\pm(\beta_0,A,T)\equiv&\sqrt{\frac{A\Lambda^3}{2\pi^5}}\kappa(\beta_0)\max\limits_{x}\abs{\frac{\hat{\phi}_\pm(x)}{\hat{\phi}_\pm(0)}}. \\
	\end{split}
\end{equation}
We note that $C^{(2)}_\pm(\beta_0,A,T)$ is a finite constant for fixed $\beta_0$, $A$ and $T$. Since we have chosen $\phi_{\pm}$ with $\Lambda<2\pi w$ (see (\ref{phipm:suppcondition})), in the regime $\beta_L\beta_R \leqslant 4\pi^2$ (which eventually holds in the DLC$_w$ limit), we can write
\begin{equation}
\frac{\beta_L^{3/2}}{\beta^2_R}e^{-A\left(\frac{4\pi^2-\Lambda^2}{\beta_R}-\beta_L\right)}\leqslant (4\pi^2)^{3/2} e^{-A\left(\frac{4\pi^2}{\beta_R}(1-w^2)-\frac{A}{T}\beta_L\right)}\times\left[  \beta_R^{-7/2}e^{-A\frac{4\pi^2}{\beta_R}\left(w^2-\frac{\Lambda^2}{4\pi^2}\right)}\right].
\end{equation}
In the DLC$_w$ limit, the first exponential factor goes to zero by (\ref{def:DLClimit}) and the second factor $[..]$ also also goes to zero because of $\Lambda<2\pi w$ (this is the reason why we made such a choice of $\Lambda$ in (\ref{phipm:suppcondition})). Therefore, we conclude that $I^{\rm dual}_{\pm,T\leqslant\bar{h}<A}$ is suppressed by $I^{\rm dual}_{\pm,\rm vac}$ in the DLC$_w$ limit.

\subsection{Direct channel}\label{section:modulardirect}
Now we consider $I_{\pm}$ in the direct-channel. According to the dual-channel results in the previous subsection, we know that $I_{\pm}$ has the asymptotic behavior
\begin{equation}\label{eq:mainasymp}
	\begin{split}
		I_\pm \sim I^{\rm dual}_{\pm,\rm vac}\sim \frac{4\pi^{5/2}\beta_R}{\beta_L^{3/2}A^{1/2}}e^{\frac{4\pi^2A}{\beta_R}}\hat{\phi}_\pm(0)
	\end{split}
\end{equation}
in the DLC$_w$ limit with the identification $\beta_R=2\pi\sqrt{\frac{A}{\bar{H}-A}}$. In this section, we would like to show that $I_\pm$ is dominated by $I_{\pm,h\in(A-\varepsilon,A+\varepsilon)}$ in the direct channel (in the same limit), i.e.
\begin{equation}
	\begin{split}
		I_{\pm,h\in(A-\varepsilon,A+\varepsilon)}\stackrel{{\rm DLC}_w}{\sim}  \frac{4\pi^{5/2}\beta_R}{\beta_L^{3/2}A^{1/2}}e^{\frac{4\pi^2A}{\beta_R}}\hat{\phi}_\pm(0).
	\end{split}
\end{equation}
We will argue this by showing that
\begin{equation}
	\begin{split}
		\lim\limits_{\mathrm{DLC}_w}\frac{I_{\pm,\rm vac}}{I^{\rm dual}_{\pm,\rm vac}}=\lim\limits_{\mathrm{DLC}_w}\frac{I_{\pm,T\leqslant h\leqslant A-\varepsilon}}{I^{\rm dual}_{\pm,\rm vac}}=\lim\limits_{\mathrm{DLC}_w}\frac{I_{\pm,h\geqslant A+\varepsilon}}{I^{\rm dual}_{\pm,\rm vac}}=0.
	\end{split}
\end{equation}

\subsubsection{Direct channel: vacuum}\label{modular:crossvac}
Let us consider the vacuum term $I_{\pm,\rm vac}$ in the direct channel of $I_\pm$:
\begin{equation}
	\begin{split}
		I_{\pm,\rm vac}\equiv&\int d t\,e^{A(\beta_L+\beta_R+i t)}\left(1-e^{-\beta_L}\right)\left(1-e^{-\beta_R-i t}\right)\,\hat{\phi}_{\pm}(t)e^{i(\bar{H}-A)t} \\
		=&e^{A(\beta_L+\beta_R)}\left(1-e^{-\beta_L}\right)\left[\phi_{\pm}(-\bar{H})-e^{-\beta_R}\phi_{\pm}(1-\bar{H})\right].
	\end{split}
\end{equation}
So $I_{\pm,\rm vac}$ has the following upper bound
\begin{equation}\label{Ipmdrect:upperbound}
	\begin{split}
		I_{\pm,\rm vac}\leqslant 2\max\limits_{x}\abs{\phi_\pm(x)} e^{A(\beta_L+\beta_R)}.
	\end{split}
\end{equation}
Compare (\ref{Ipmdrect:upperbound}) with \eqref{eq:mainasymp}, we see that the ratio $I_{\pm,\rm vac}/I_{\pm,\rm vac}^{\rm dual}$ is asymptotically bounded as follows in $\mathrm{DLC}_w$ limit :
\begin{equation}\label{eq:vac}
\begin{split}
\abs{\frac{I_{\pm,\rm vac}}{I^{\rm dual}_{\pm,\rm vac}}}\stackrel{\mathrm{DLC}_w}{\lesssim}& C^{(3)}_\pm(A)\frac{\beta_L^{3/2}}{\beta_R} e^{-A\left(\frac{4\pi^2}{\beta_R}-\beta_L-\beta_R\right)}, \\
C^{(3)}_\pm(A)\equiv&\frac{A^{1/2}}{2\pi^{5/2}}\max\limits_{x}\abs{\frac{\phi_\pm(x)}{\hat{\phi}_\pm(0)}}. \\
\end{split}
\end{equation}
$C^{(3)}_\pm(A)$ is a finite constant. The rest part of (\ref{eq:vac}) is bounded as follows
\begin{equation}
	\begin{split}
		\frac{\beta_L^{3/2}}{\beta_R} e^{-A\left(\frac{4\pi^2}{\beta_R}-\beta_L-\beta_R\right)}=&\frac{\beta_L^{3/2}}{\beta_R}e^{-\frac{4\pi^2 w^2}{\beta_R}} e^{-A\left(\frac{4\pi^2(1-w^2)}{\beta_R}-\beta_L-\beta_R\right)} \\
		\stackrel{{\rm DLC}_w}{\lesssim}&\frac{8\pi^3}{\beta_R^{5/2}}e^{-\frac{4\pi^2 w^2}{\beta_R}} \times e^{-A\left(\frac{4\pi^2T(1-w^2)}{A\beta_R}-\beta_L-\beta_R\right)} \\
		\rightarrow&0.
	\end{split}
\end{equation}
Here in the second line we used the fact that $\beta_L\beta_R\leqslant 4\pi^2$ eventually in the DLC$_w$ limit, and the third line follows from the definition of the DLC$_w$ limit (recall eq.\,(\ref{def:DLCwlimit})). Therefore, $I_{\pm,\rm vac}$ is suppressed by $I^{\rm dual}_{\pm,\rm vac}$ in the DLC$_w$ limit.

\subsubsection{Direct channel: high twist and low twist}\label{modular:crosshighlow}
Then let us consider the non-vacuum terms in the direct channel of $I_\pm$: $I_{\pm,T\leqslant h\leqslant A-\varepsilon}$ and $I_{\pm,h\geqslant A+\varepsilon}$, given by (\ref{def:allIpm}). Integrating over $t$ in (\ref{def:allIpm}) for $I_{\pm,T\leqslant h\leqslant A-\varepsilon}$ and $I_{\pm,h\geqslant A+\varepsilon}$, we get
\begin{equation}
	\begin{split}
		I_{\pm,T\leqslant h\leqslant A-\varepsilon}=&\,e^{A(\beta_L+\beta_R)}\int_{T}^{A-\varepsilon}d h\int_{T}^{\infty}d\bar{h}\,\rho(h,\bar{h})e^{-h\beta_L-\bar{h}\beta_R}\,\phi_{\pm}(h-\bar{H}), \\
		I_{\pm,h\geqslant A+\varepsilon}=&\,e^{A(\beta_L+\beta_R)}\int_{A+\varepsilon}^{\infty}d h\int_{T}^{\infty}d\bar{h}\,\rho(h,\bar{h})e^{-h\beta_L-\bar{h}\beta_R}\,\phi_{\pm}(h-\bar{H}). \\
	\end{split}
\end{equation}
Bounding $\phi_\pm(h-\bar{H})$ by its maximal value, we get
\begin{equation}\label{Ipmlowhigh:obviousbound}
	\begin{split}
		\abs{I_{\pm,T\leqslant h\leqslant A-\varepsilon}}\leqslant&\max\limits_{x}\abs{\phi_{\pm}(x)}\tilde{Z}_{T\leqslant h\leqslant A-\varepsilon}(\beta_L,\beta_R), \\
		\abs{I_{\pm,h\geqslant A+\varepsilon}}\leqslant&\max\limits_{x}\abs{\phi_{\pm}(x)}\tilde{Z}_{h\geqslant A+\varepsilon}(\beta_L,\beta_R). \\
	\end{split}
\end{equation}
Now it suffices to show that $\tilde{Z}_{T\leqslant h\leqslant A-\varepsilon}(\beta_L,\beta_R)$ and $\tilde{Z}_{h\geqslant A+\varepsilon}(\beta_L,\beta_R)$ are suppressed by $I^{\rm dual}_{\pm,\rm vac}$ in the DLC$_w$ limit. This follows from the same analysis as in \cite{Pal:2022vqc}, section 3.\footnote{Here a quick way to see the suppression is to rewrite the asymptotic behavior of $I^{\rm dual}_{\pm,\rm vac}$ as
\begin{equation}
		\begin{split}
			I^{\rm dual}_{\pm,\rm vac}\stackrel{{\rm DLC}_w}{\sim}\sqrt{\frac{\beta_{R}^3}{4\pi A}}\hat{\phi}_{\pm}(0)\times\sqrt{\frac{4\pi^2}{\beta_{L}\beta_{R}}}\tilde{Z}_{\rm vac}\left(\frac{4\pi^2}{\beta_{L}},\frac{4\pi^2}{\beta_{R}}\right).
		\end{split}
\end{equation}
The second factor is exactly the dual-channel vacuum term of the partition function $\tilde{Z}$. It is known in \cite{Pal:2022vqc} $\frac{\tilde{Z}_{T\leqslant h\leqslant A-\varepsilon}(\beta_L,\beta_R)}{\sqrt{\frac{4\pi^2}{\beta_{L}\beta_{R}}}\tilde{Z}_{\rm vac}\left(\frac{4\pi^2}{\beta_{L}},\frac{4\pi^2}{\beta_{R}}\right)}$ and $\frac{\tilde{Z}_{h\geqslant A+\varepsilon}(\beta_L,\beta_R)}{\sqrt{\frac{4\pi^2}{\beta_{L}\beta_{R}}}\tilde{Z}_{\rm vac}\left(\frac{4\pi^2}{\beta_{L}},\frac{4\pi^2}{\beta_{R}}\right)}$ decays exponentially fast in the M$^*$ limit. The DLC$_w$ limit in this paper is a stronger version of the M$^*$ limit, and it is sufficient to kill the extra slow-growing factors $\beta_{R}^{-3/2}$.
}

Let us derive an upper bound on $\tilde{Z}_{h\geqslant A+\varepsilon}(\beta_L,\beta_R)$ first. We chose some fixed $\beta_0\in(0,\infty)$ and consider the regime $\beta_R\leqslant \frac{4\pi^2}{\beta_0}\leqslant\beta_L$. We have
\begin{equation}\label{Zhigh:bound}
	\begin{split}
		\tilde{Z}_{h\geqslant A+\varepsilon}(\beta_L,\beta_R)\leqslant& e^{-\varepsilon\left(\beta_L-\frac{4\pi^2}{\beta_0}\right)}\tilde{Z}_{h\geqslant A+\varepsilon}\left(\frac{4\pi^2}{\beta_0},\beta_R\right) \\
		\leqslant&e^{-\varepsilon\left(\beta_L-\frac{4\pi^2}{\beta_0}\right)}\sqrt{\frac{\beta_0}{\beta_{R}}}\tilde{Z}\left(\beta_0,\frac{4\pi^2}{\beta_{R}}\right) \\
		\leqslant&e^{-\varepsilon\left(\beta_L-\frac{4\pi^2}{\beta_0}\right)}\sqrt{\frac{\beta_0}{\beta_{R}}}\kappa(\beta_0)e^{A\left(\beta_0+\frac{4\pi^2}{\beta_{R}}\right)}.
	\end{split}
\end{equation}
Here in the first line we used $e^{(A-h)\beta_L}\leqslant e^{-\varepsilon\left(\beta_L-\frac{4\pi^2}{\beta_0}\right)}e^{(A-h)\frac{4\pi^2}{\beta_0}}$ for $\beta_{L}\geqslant\frac{4\pi^2}{\beta_0}$ and $h\geqslant A+\varepsilon$, in the second line we bounded $\tilde{Z}_{h\geqslant A+\varepsilon}$ by the full partition function $\tilde{Z}$ and used modular invariance (\ref{modulartransformation}), and in the last line we used lemma \ref{lemma:Ztildebound} in the regime $\frac{4\pi^2}{\beta_{R}}\geqslant\beta_0$. 

By (\ref{Ipmvac:asymp}), (\ref{Ipmlowhigh:obviousbound}) and (\ref{Zhigh:bound}) we get
\begin{equation}\label{eq:high}
	\begin{split}
		\abs{\frac{I_{\pm,h\geqslant A+\varepsilon}}{I^{\rm dual}_{\pm,\rm vac}}}\stackrel{{\rm DLC}_w}{\lesssim}&C^{(4)}_\pm(A,\beta_0)\left(\frac{\beta_{L}}{\beta_{R}}\right)^{3/2} e^{-\varepsilon\left(\beta_L-\frac{4\pi^2}{\beta_0}\right)+A\beta_0}\\
		C^{(4)}_\pm(A,\beta_0)=&\kappa(\beta_0)\sqrt{\frac{A\beta_0}{16\pi^5}}\max\limits_{x}\abs{\frac{\phi_{\pm}(x)}{\hat{\phi}_\pm(0)}} \\
	\end{split}
\end{equation}
$C^{(4)}_\pm(A,\beta_0)$ is a finite constant. The rest part of (\ref{eq:high}) is bounded as follows
\begin{equation}\label{eq:highprime}
	\begin{split}
		\left(\frac{\beta_{L}}{\beta_{R}}\right)^{3/2}e^{-\varepsilon\left(\beta_L-\frac{4\pi^2}{\beta_0}\right)+A\beta_0}=\beta_{L}^{3/2}e^{-\varepsilon\beta_{L}/2+A\beta_0+\frac{4\pi^2\varepsilon}{\beta_0}}\times\beta_{R}^{-3/2}e^{-\varepsilon\beta_{L}/2}\stackrel{{\rm DLC}_w}{\longrightarrow}0.
	\end{split}
\end{equation}
Here the first factor obviously vanishes as $\beta_{L}\rightarrow\infty$, and the second factor vanishes because of last condition of DLC$_w$. 

Then let us derive an upper bound on $\tilde{Z}_{T\leqslant h\leqslant A-\varepsilon}(\beta_L,\beta_R)$. We introduce an auxiliary variable $\beta_L'$. Then we have the following upper bound on $\tilde{Z}_{T\leqslant h\leqslant A-\varepsilon}(\beta_L,\beta_R)$:
\begin{equation}\label{ineq:lowtwist}
	\begin{split}
		\tilde{Z}_{T\leqslant h\leqslant A-\varepsilon}(\beta_L,\beta_R)\leqslant e^{-\varepsilon(\beta_L'-\beta_L)}\tilde{Z}_{T\leqslant h\leqslant A-\varepsilon}(\beta_L',\beta_R)\quad(0<\beta_L\leqslant\beta_L').
	\end{split}
\end{equation} 
This bound follows from the fact that $e^{(A-h)\beta_L}\leqslant e^{-\varepsilon(\beta_L'-\beta_L)}e^{(A-h)\beta_L'}$ for $h\leqslant A-\varepsilon$ and $\beta_L\leqslant\beta_L'$. We choose
\begin{equation}\label{def:auxbeta}
	\begin{split}
		\beta_L'=\frac{4\pi^2 T(1-w^2/2)}{A\beta_R}\left(=\frac{2\pi T(1-w^2/2)}{A}\sqrt{\frac{\bar{H}-A}{A}}\right).
	\end{split}
\end{equation}
Then in the DLC$_w$ limit, we have
\begin{equation}\label{betaprimebeta}
	\begin{split}
		\beta_L'-\beta_L\geqslant\frac{4\pi^2T(1-w^2)}{A\beta_{R}}-\beta_{L}\rightarrow\infty.
	\end{split}
\end{equation}
We see that $\beta_L'\geqslant\beta_L$ eventually in the DLC$_w$ limit so eq.\,(\ref{ineq:lowtwist}) holds, and
\begin{equation}\label{ineq:Zlowtwist}
	\begin{split}
		\tilde{Z}_{T\leqslant h\leqslant A-\varepsilon}(\beta_L',\beta_R) \leqslant&\sqrt{\frac{4\pi^2}{\beta_L'\beta_R}}\tilde{Z}\left(\frac{4\pi^2}{\beta_L'},\frac{4\pi^2}{\beta_R}\right) \\
		=&\sqrt{\frac{4\pi^2}{\beta_L'\beta_R}}\left[\tilde{Z}_{\rm vac}\left(\frac{4\pi^2}{\beta_L'},\frac{4\pi^2}{\beta_R}\right)+\tilde{Z}_{h,\bar{h}\geqslant T}\left(\frac{4\pi^2}{\beta_L'},\frac{4\pi^2}{\beta_R}\right)\right] \\
		\leqslant&\sqrt{\frac{4\pi^2}{\beta_L'\beta_R}}\left[\frac{4\pi^2}{\beta_{L}'}e^{A\left(\frac{4\pi^2}{\beta_L'}+\frac{4\pi^2}{\beta_R}\right)}+\kappa(\beta_0)\sqrt{\frac{\beta_L'\beta_0}{4\pi^2}}e^{(A-T)\left(\frac{4\pi^2}{\beta_R}-\frac{4\pi^2}{\beta_0}\right)+A(\beta_L'+\beta_0)}\right] \\
		=&\left(\frac{4\pi^2}{\beta_{L}'}\right)^{3/2}\beta_{R}^{-1/2}e^{A\left(\frac{4\pi^2}{\beta_L'}+\frac{4\pi^2}{\beta_R}\right)} \\
		&\times\left[1+\kappa(\beta_0)\beta_0^{1/2}\left(\frac{\beta_L'}{4\pi^2}\right)^{3/2}e^{A\left(\beta_L'-\frac{4\pi^2}{\beta_{L}'}+\beta_0-\frac{4\pi^2}{\beta_0}\right)-T\left(\frac{4\pi^2}{\beta_R}-\frac{4\pi^2}{\beta_0}\right)}\right] \\
	\end{split}
\end{equation}
Here in the first line we bounded $\tilde{Z}_{T\leqslant h\leqslant A-\varepsilon}$ by the full partition function $\tilde{Z}$ and used modular invariance (\ref{modulartransformation}), in the second line we rewrote $\tilde{Z}$ as $\tilde{Z}_{\rm vac}+\tilde{Z}_{h,\bar{h}\geqslant T}$, in the third line we used $\tilde{Z}_{\rm vac}(\beta,\bar{\beta})\leqslant \beta e^{A(\beta+\bar{\beta})}$ and lemma \ref{lemma:Ztildebound}, and the last line is just a rewriting of the third line. The second term in $\left[\ldots\right]$ vanishes in the DLC$_w$ limit because
\begin{equation}\label{Zlowtwist:2ndterm}
	\begin{split}
		\left(\frac{\beta_L'}{4\pi^2}\right)^{3/2}e^{A\left(\beta_L'-\frac{4\pi^2}{\beta_{L}'}+\beta_0-\frac{4\pi^2}{\beta_0}\right)-T\left(\frac{4\pi^2}{\beta_R}-\frac{4\pi^2}{\beta_0}\right)}=&\left(\frac{(1-w^2/2)T}{A\beta_R}\right)^{3/2}e^{-\frac{2\pi^2w^2T}{\beta_R}-\frac{A^2\beta_{R}}{(1-w^2/2)T}} \\
		&\times e^{A\left(\beta_0-\frac{4\pi^2}{\beta_0}+\frac{4\pi^2 T}{A\beta_0}\right)} \\
		\stackrel{{\rm DLC}_w}{\longrightarrow}&0. \\
	\end{split}
\end{equation}
By (\ref{Ipmvac:asymp}), (\ref{ineq:lowtwist}), (\ref{ineq:Zlowtwist}) and (\ref{Zlowtwist:2ndterm}) we get
\begin{equation}\label{eq:low}
	\begin{split}
		\abs{\frac{I_{\pm,T\leqslant h\leqslant A-\varepsilon}}{I^{\rm dual}_{\pm,\rm vac}}}\stackrel{{\rm DLC}_w}{\lesssim}&C^{(5)}_{\pm}(A)\left(\frac{4\pi^2\beta_{L}}{\beta_{L}'\beta_{R}}\right)^{3/2}e^{-\varepsilon(\beta_L'-\beta_L)+\frac{4\pi^2A}{\beta_{L}'}}, \\
		C^{(5)}_{\pm}(A)=&\sqrt{\frac{A}{16\pi^5}}\max\limits_{x}\abs{\frac{\phi_{\pm}(x)}{\hat{\phi}_\pm(0)}}. \\
	\end{split}
\end{equation}
$C^{(4)}_\pm(A,\beta_0)$ is a finite constant. The rest part of (\ref{eq:low}) is bounded as follows
\begin{equation}\label{eq:lowprime}
	\begin{split}
		\left(\frac{4\pi^2\beta_{L}}{\beta_{L}'\beta_{R}}\right)^{3/2}e^{-\varepsilon(\beta_L'-\beta_L)+\frac{4\pi^2A}{\beta_{L}'}}=&\left(\frac{A\beta_{L}}{(1-w^2/2)T}\right)^{3/2}e^{-\frac{2\pi^2Tw^2\varepsilon}{A\beta_{R}}}\times e^{-\varepsilon\left(\frac{4\pi^2T(1-w^2)}{A\beta_{R}}-\beta_{L}\right)+\frac{A^2\beta_{R}}{(1-w^2/2)T}} \\
		\leqslant&\left(\frac{4\pi^2(1-w^2)}{(1-w^2/2)\beta_{R}}\right)^{3/2}e^{-\frac{2\pi^2Tw^2\varepsilon}{A\beta_{R}}} \quad({\rm eventually\ in\ DLC}_w)\\
		&\times e^{-\varepsilon\left(\frac{4\pi^2T(1-w^2)}{A\beta_{R}}-\beta_{L}\right)+\frac{A^2\beta_{R}}{(1-w^2/2)T}} \\
		\stackrel{{\rm DLC}_w}{\longrightarrow}&0.
	\end{split}
\end{equation}
Here the first line is just a rewriting, in the second line we used the fact that $\beta_{L}\leqslant\frac{4\pi^2T(1-w^2)}{A\beta_{R}}$ eventually in DLC$_w$ limit, and in the last line we used fact that both two factors vanishes in the DLC$_w$ limit.

\subsection{Summary of the estimates, two-sided bounds for fixed $\varepsilon$}\label{section:modular:tauberian}
\renewcommand{\arraystretch}{1.5}
\begin{table}[h]
	\centering
	\begin{tabular}{@{}ll@{}ll@{}}\toprule
		& $I_{\pm,\ldots}/I^{\rm dual}_{\pm,\rm vac}$ &\quad Eq.\\
		\midrule
		dual channel, $T\leqslant \bar h\leqslant A$ & $\left(\frac{\beta_L}{\beta_R}\right)^{3/2}e^{-\frac{4\pi^2T}{\beta_R}+A\beta_L}$  &\quad \eqref{eq:dualLessA} \\
		dual channel, $\bar h\geqslant A$ & $\frac{\beta_L^{3/2}}{\beta_R^2}e^{-A\left(\frac{4\pi^2-\Lambda^2}{\beta_R}-\beta_L\right)}$ &\quad  \eqref{eq:dualGreatA} \\
		\midrule
		direct channel, $h=0$ & $\frac{\beta_L^{3/2}}{\beta_R}e^{-A\left(\frac{4\pi^2}{\beta_R}-\beta_L\right)}$  &\quad \eqref{eq:vac} \\
		direct channel, $h\geqslant A+\varepsilon$ & $\left(\frac{\beta_L}{\beta_R}\right)^{3/2}e^{-\varepsilon\beta_L}$&\quad \eqref{eq:high} \\
		direct channel, $T\leqslant h\leqslant A-\varepsilon$ & $\beta_L^{3/2}e^{-\varepsilon\left(\frac{4\pi^2T(1-w^2/2)}{A\beta_R}-\beta_L\right)}$  &\quad \eqref{eq:low}\\
		\bottomrule
	\end{tabular}
	\caption{\label{table:modularestimate} Estimates of the suppressed contributions in the $\mathrm{DLC}_w$ limit. Here we have ignored the constant factors.}
\end{table}
We summarize our estimates on various terms in (\ref{modularinv:split}) in table \ref{table:modularestimate}. We conclude that in the $\mathrm{DLC}_w$ limit, $I_{\pm,h\in(A-\varepsilon,A+\varepsilon)}$ dominates the direct channel and $I^{\rm dual}_{\pm,\rm vac}$ dominates the dual channel. So we get $I_{\pm,h\in(A-\varepsilon,A+\varepsilon)}\stackrel{{\rm DLC}_w}{\sim}I^{\rm dual}_{\pm,\rm vac}$, which justifies eq.\,(\ref{modular:idea}). 

On the other hand, eq.\,(\ref{Ipmvac:asymp:inH}) gives the asymptotic behavior of $I^{\rm dual}_{\rm vac}$ in the $\mathrm{DLC}_w$ limit. Together with (\ref{A:upperbound}), (\ref{A:lowerbound}) and (\ref{choice:betaRHbar}), we conclude that
\begin{equation}
		\begin{split}\label{eq:resultraw}
			\hat{\phi}_{-}(0) \lesssim \frac{\mathcal{A}(\beta_L,\bar{H},\varepsilon,\delta)}{\left(\frac{4\pi^2}{\beta_L}\right)^{3/2}\sqrt{\frac{\pi}{\bar{H}}}e^{4\pi\sqrt{A\bar{H}}}} \lesssim \hat{\phi}_{+}(0)
		\end{split}
\end{equation}
in the $\mathrm{DLC}_w$ limit. Here we would like to emphasize that the above equation is valid only when $\hat{\phi}_{+}(0)\neq0$ (the lower bound is trivial when $\hat{\phi}_{-}(0)=0$).
    
Recall that we would like to derive the bounds on $\mathcal{A}_J(\varepsilon,\beta_{L})$. We use (\ref{AJ:prop}), i.e.\,$\mathcal{A}_J(\beta_L,\varepsilon)$ is exactly the same as $\mathcal{A}(\beta_{L},\bar{H},\varepsilon,\delta)$ for $\delta\in(\varepsilon,1-\varepsilon)$. This fact allows us to choose different $\delta$ for the upper and lower bounds, say $\delta_+$ and $\delta_-$ respectively.  Therefore,  we conclude from eqs.\,\eqref{AJ:prop} and \eqref{eq:resultraw} that
\begin{equation}\label{AJ:twosided1}
		\begin{split}
			\hat{\phi}_{-,\delta_{-}}(0)\lesssim \frac{\mathcal{A}_J(\beta_L,\varepsilon )}{8\pi^{7/2}\beta_L^{-3/2}J^{-1/2}e^{4\pi\sqrt{AJ}}} \lesssim \hat{\phi}_{+,\delta_+}(0).
		\end{split}
\end{equation}
in the DLC$_w$ limit, with $\delta_\pm\in(\varepsilon,1-\varepsilon)$. Here we added the extra subscript $\delta_{\pm}$ to $\phi_\pm$, which simply means $\phi_{\pm}$ in (\ref{phipm}) with $\delta=\delta_{\pm}$.

Now the question is: given fixed $\Lambda$ and $\delta_{\pm}$, what are the optimal values of $\hat{\phi}_{\pm,\delta_{\pm}}(0)$? This problem was studied in \cite{Mukhametzhanov:2020swe} for $\Lambda=2\pi$. In that case, when $\delta_{+}$ very close to 0 and $\delta_{-}$ very close to 1, the optimal functions $\phi_\pm$ are given by
\begin{equation}\label{def:phipm}
	\begin{split}
		\phi_{+,\delta_+}(x)&=\frac{16 \delta_+ ^2 \left[x \cos \left(\pi\delta_+\right) \sin \left(\pi  x\right)-\delta_+  \sin \left(\pi\delta_+\right) \cos \left(\pi x\right)\right]^2}{\left(x^2-\delta_+ ^2\right)^2 (2\pi\delta_+  +\sin (2\pi\delta_+))^2},\\
		\phi_{-,\delta_-}(x)&=\frac{ \delta_- ^2 \left(x \cos \left(\pi x\right)-\delta_-  \cot \left(\pi\delta_-  \right) \sin \left(\pi  x\right)\right)^2}{x^2 (\delta_-^2 -x^2) \left(\pi\delta_- \cot \left(\pi\delta_-\right)-1\right)^2}. \\
	\end{split}
\end{equation}
This choice of $\phi_\pm$ gives\footnote{These values are obtained from the first equation of eq.\,(86) and eq.\,(88) in \cite{Mukhametzhanov:2020swe}.  These function appeared in \cite{littmann2013quadrature} in mathematics literature.}
\begin{equation}
	\begin{split}
		\hat{\phi}_{+,\delta_+}(0)=\frac{1}{2\pi}\frac{2}{1+\frac{\sin(2\pi\delta_{+})}{2\pi\delta_{+}}},\quad\hat{\phi}_{-,\delta_-}(0)=\frac{1}{2\pi}\frac{1}{1-\frac{\tan(\pi\delta_{-})}{\pi\delta_{-}}}.
	\end{split}
\end{equation}
The case with arbitrary $\Lambda$ can easily be obtained by doing scaling:
\begin{equation}\label{phipm:scaling}
	\begin{split}
		\phi_{\pm,\delta_\pm}(x)\quad(\Lambda=2\pi)\quad\longrightarrow\quad\phi^{\Lambda}_{\pm,\delta_\pm}(x):=\phi_{\pm,\frac{\Lambda}{2\pi}\delta_\pm}\left(\frac{\Lambda x}{2\pi}\right).
	\end{split}
\end{equation}
Under scaling, $\hat{\phi}_{\pm,\delta_{\pm}}(0)$ is given by
\begin{equation}
	\begin{split}
		\hat{\phi}^\Lambda_{\pm,\delta_{\pm}}(0)=\frac{2\pi}{\Lambda}\hat{\phi}_{\pm,\frac{\Lambda}{2\pi}\delta_\pm}(0).
	\end{split}
\end{equation}
This gives us
\begin{equation}\label{hatphipm:optimalR}
	\begin{split}
		\hat{\phi}_{+,\delta_+}(0)=\frac{1}{\Lambda}\frac{2}{1+\frac{\sin(\Lambda\delta_+)}{\Lambda\delta_+}},\quad \hat{\phi}_{-,\delta_-}(0)=\frac{1}{\Lambda}\frac{1}{1-\frac{2\tan\left(\frac{\Lambda\delta_-}{2}\right)}{\Lambda\delta_-}}\,.
	\end{split}
\end{equation}
Here we neglect the superscript $\Lambda$ in the expression.

Now we insert the optimal values in (\ref{hatphipm:optimalR}) into \eqref{AJ:twosided1}. But we need to be careful because the following constraints on $\delta_{\pm}$, $w$ and $\Lambda$ should be satisfied:
\begin{itemize}
	\item According to (\ref{AJ:prop}) we need $\delta_{\pm}\in(\varepsilon,1-\varepsilon)$.
	\item According to (\ref{phipm:suppcondition}) we must have $\Lambda<2\pi w$.
	\item The choice of $\phi_{+,\delta_+}$ requires that $0<\Lambda\delta_{+}<\pi$ (see eq.\,(86) of \cite{Mukhametzhanov:2020swe}).
	\item The choice of $\phi_{-,\delta_{-}}$ requires that $\pi<\Lambda\delta_-<2\pi$ (see eq.\,(88) of \cite{Mukhametzhanov:2020swe}).
\end{itemize}	
To make the above constraints consistent we also need that
\begin{equation}
	\begin{split}
		w>\frac{\Lambda}{2\pi}>\frac{1}{2\delta_{-}}>\frac{1}{2(1-\varepsilon)}\quad\Rightarrow\quad\varepsilon<1-\frac{1}{2w}.
	\end{split}
\end{equation}
Then the condition $\varepsilon>0$ implies $w>1/2$. Under the above constraints, we choose $\delta_-$ to be arbitrarily close to $1-\varepsilon$, $\delta_+$ to be arbitrarily close to $\varepsilon$ and $\Lambda$ to be arbitrarily close to $2\pi w$. So $\mathcal{A}_J$ has the following asymptotic two-sided bounds
\begin{equation}\label{AJ:twosided2}
	\begin{split}
		\frac{1}{w}\frac{1}{1-\frac{\tan\left(\pi w(1-\varepsilon)\right)}{\pi w(1-\varepsilon)}}-\alpha_-\lesssim \frac{\mathcal{A}_J(\beta_L,\varepsilon )}{4\pi^{5/2}\beta_L^{-3/2}J^{-1/2}e^{4\pi\sqrt{AJ}}} \lesssim \frac{1}{w}\frac{2}{1+\frac{\sin(2\pi w\varepsilon)}{2\pi w\varepsilon}}+\alpha_+.
	\end{split}
\end{equation}
Here $\alpha_\pm$ are defined by
\begin{equation}
	\begin{split}
		\alpha_+:=2\pi\left(\hat{\phi}^{\Lambda}_{+,\delta_{+}}(0)-\hat{\phi}^{2\pi w}_{+,\varepsilon}(0)\right),\quad\alpha_-:=2\pi\left(\hat{\phi}^{2\pi w}_{-,1-\varepsilon}(0)-\hat{\phi}^{\Lambda}_{-,\delta_{-}}(0)\right).
	\end{split}
\end{equation}
They are positive but can be arbitrarily small in the limit $\delta_{+}\rightarrow\varepsilon$ and $\delta_{-}\rightarrow1-\varepsilon$. However $\mathcal{A}_J(\beta_{L},\varepsilon)$ does not depend on $\delta_{\pm}$ and $\Lambda$ because $\phi_{\pm}$ are just auxiliary functions for our analysis. So we arrive at \eqref{eq:resultFiniteEpsilon} for $w\in(\frac{1}{2},1)$ and $\varepsilon\in\left(0,1-\frac{1}{2w}\right)$ fixed.
\begin{remark}
There are also other choices of $\hat{\phi}_{\pm,\delta_{\pm}}$, which give non-optimal but simpler expressions of upper and lower bounds than \eqref{eq:resultFiniteEpsilon}. For example, one can choose $\phi_{\pm,\delta}$ at $\Lambda=2\pi$ to be the Beurling-Selberg functions \cite{BeurlingArne1989Tcwo, Selberg1991}, denoted by $\phi^{\rm BS}_{\pm,\delta}$.\footnote{For the explicit constructions and technical details, see \cite{Mukhametzhanov:2020swe}, eqs.\,(38) - (42) and appendix C.} The B-S functions give $\hat{\phi}^{\rm BS}_{\pm,\delta}(0)=\frac{1}{2\pi}(2\delta\pm1)$. Then after the same rescaling procedure as \eqref{phipm:scaling}, one gets
\begin{equation}
	\begin{split}
		\hat{\phi}_{\pm,\delta_\pm}=\frac{1}{2\pi}\left(2\delta_{\pm}\pm\frac{2\pi}{\Lambda}\right).
	\end{split}
\end{equation}
By taking extremal values of $\delta_{\pm}$, the upper and lower bounds in \eqref{eq:resultFiniteEpsilon} becomes $\frac{1}{w}+2\varepsilon$ and $2-\dfrac{1}{w}-2\varepsilon$. We note that in the limit $\varepsilon\rightarrow0$, the B-S function already gives the optimal upper bound, while it gives the optimal lower bound only when we also take the limit $w\rightarrow1$.
\end{remark}
So far, the statement in theorem \ref{theorem:modulartauberianFJ} has been established with for fixed $\varepsilon$. As a final step, we would like to let $\varepsilon$ also go to zero in the DLC$_w$ limit. This will be the subject of the next subsection.

\subsection{Shrinking the $(A-\varepsilon,A+\varepsilon)$ window}\label{section:epsilonwindow}
In this subsection, we would like to establish the final part of theorem \ref{theorem:modulartauberianFJ}, which allows for the vanishing of $\varepsilon$ in the DLC$_w$ limit, provided it remains larger than $\varepsilon_{\rm min}(\beta_{L},J)$ as defined in \eqref{epsilon:choice:DLCJ}.

In our analysis, the $\varepsilon$-dependence arises in three key aspects:
\begin{enumerate}
	\item In appendix \ref{app:Ivac:asym}, we used the dominated convergence theorem, which requires the condition that $\left|\phi_{\pm}\right|$, which depends on $\delta$, must be bounded by an integrable function. This condition is automatically satisfied when $\varepsilon$ is fixed since, in that case, we work with fixed functions $\phi_{\pm}$ that are integrable by themselves. However, as $\varepsilon\rightarrow 0$, the functions $\phi_{\pm}$ are no longer fixed, and it becomes crucial to ensure that the family of $\phi_{\pm}$ functions we consider remains uniformly bounded by certain integrable functions that are independent of $\delta$. This will ensure the applicability of the dominated convergence theorem in the limit as $\varepsilon\rightarrow 0$.
	
	\item The ratios $\max\limits_{x}\abs{\frac{\hat{\phi}_{\pm}(x)}{\hat{\phi}_\pm(0)}}$ (as seen in \eqref{def:C1} and \eqref{eq:dualGreatA}) and $\max\limits_{x}\abs{\frac{\phi_\pm(x)}{\hat{\phi}_\pm(0)}}$ (as seen in \eqref{eq:vac}, \eqref{eq:high}, and \eqref{eq:low}) depend on the choice of $\phi_{\pm}$. Ultimately, we selected $\delta_{+}=\varepsilon\rightarrow0$ for $\phi_{+}$ and $\delta_{-}=1-\varepsilon\rightarrow1$ for $\phi_{-}$. Therefore, similarly to point 1, we need to derive some uniform bounds on the ratios.
	
	\item The bounds on the high- and low-twist contributions in the direct channel of $I_\pm$ incorporate exponential factors that rely on $\varepsilon$, as presented in Table \ref{table:modularestimate}.
\end{enumerate}
To address the concerns raised in the first and second points, we establish upper bounds on the quantities
\begin{equation*}
	\begin{split}
		\abs{\phi_{\pm,\delta_{\pm}}(x)},\quad\max\limits_{x}\abs{\frac{\phi_{\pm,\delta_{\pm}}(x)}{\hat{\phi}_{\pm,\delta_{\pm}}(0)}}\quad{\rm and}\ \max\limits_{x}\abs{\frac{\hat{\phi}_{\pm,\delta_{\pm}}(x)}{\hat{\phi}_{\pm,\delta_{\pm}}(0)}}
	\end{split}
\end{equation*}
for the chosen $\phi_{\pm,\delta_{\pm}}$. These bounds remain uniform in $\delta_{\pm}$ within the regimes that permit the limits $\delta_{+}\rightarrow0$ and $\delta_{-}\rightarrow1$. The precise statements of these bounds are presented in lemma \ref{lemma:phiplusbounds} and lemma \ref{lemma:phiminusbounds}, while the detailed proofs can be found in appendix \ref{appendix:uniformbounds}. Consequently, the first and second concerns are effectively resolved through the establishment of these uniform bounds.

To address the third point, we examine the conditions required for the DLC$_w$ limit as $\varepsilon$ approaches 0, as shown in table \ref{table:modularestimate}. These conditions can be expressed as follows:
\begin{equation}\label{epsiloncondition}
	\begin{split}
		\beta_L^{3/2}e^{-\varepsilon\left(\frac{4\pi^2T}{A\beta_R}(1-w^2/2)-\beta_L\right)},\left(\frac{\beta_L}{\beta_R}\right)^{3/2}e^{-\varepsilon\beta_L}\rightarrow0. \\
	\end{split}
\end{equation}
I.e.\ the exponential factors must decay rapidly enough to render the power-law factors negligible.

For the first term in \eqref{epsiloncondition}, we use the DLC$_w$ condition \eqref{def:DLCwlimit} and get the following inequality:
\begin{equation}
	\begin{split}
		\beta_L^{3/2}e^{-\varepsilon\left(\frac{4\pi^2T}{A\beta_R}(1-w^2/2)-\beta_L\right)}\leqslant\left(\frac{4\pi^2T(1-w^2)}{A\beta_R}\right)^{3/2}e^{-\varepsilon\frac{2\pi^2 w^2T}{A\beta_R}}.
	\end{split}
\end{equation}
In the DLC$_w$ limit, where $\beta_R\rightarrow 0$, we need the r.h.s.\ of the above inequality to vanish. This can be achieved if we impose the following sufficient condition:
\begin{equation}\label{epsilon:conditionlow}
	\begin{split}
		\varepsilon\frac{2\pi^2 w^2T}{A\beta_R}\geqslant\left(\frac{3}{2}+\alpha\right)\log(1/\beta_{R}),
	\end{split}
\end{equation}
where $\alpha$ is an arbitrary fixed positive constant.

For the second term in \eqref{epsiloncondition}, we rewrite it as follows:
\begin{equation}
	\begin{split}
		\left(\frac{\beta_L}{\beta_R}\right)^{3/2}e^{-\varepsilon\beta_L}=e^{-\left(\varepsilon-\frac{3\log(1/\beta_{R})}{2\beta_{L}}\right)\beta_{L}+\frac{3}{2}\log\beta_{L}}.
	\end{split}
\end{equation}
In order for this term to vanish in the DLC$_w$ limit, we require the following sufficient condition:
\begin{equation}\label{epsilon:conditionhigh}
	\begin{split}
		\left(\varepsilon-\frac{3\log(1/\beta_{R})}{2\beta_{L}}\right)\beta_{L}\geqslant \left(\frac{3}{2}+\alpha\right)\log\beta_{L},
	\end{split}
\end{equation}
where $\alpha$ is an arbitrary fixed positive constant.

Now we select the same $\alpha$ for both \eqref{epsilon:conditionlow} and \eqref{epsilon:conditionhigh} for simplicity, and determine the smallest value of $\varepsilon$ that satisfies these conditions. Specifically, we have:
\begin{equation}\label{epsilon:choice:DLC}
	\begin{split}
		\varepsilon\geqslant\max\left\{\left(\frac{3}{2}+\alpha\right)\frac{A}{2\pi^2w^2T}\beta_{R}\log(1/\beta_{R}),\ 
		\frac{3\log(1/\beta_{R})}{2\beta_{L}}+\left(\frac{3}{2}+\alpha\right)\frac{\log\beta_{L}}{\beta_{L}}\right\}
	\end{split}
\end{equation}
It can be verified that, in the DLC$_w$ limit, the choice of $\varepsilon$ given by \eqref{epsilon:choice:DLC} tends to zero. The last part of theorem \ref{theorem:modulartauberianFJ} follows by choosing $\alpha=\frac{1}{2}$ in \eqref{epsilon:choice:DLC} and using the identification \eqref{choice:betaRHbar} (which implies $\beta_{R}\sim J^{-1/2}$ in the DLC$_w$ limit). This finishes the whole proof of theorem \ref{theorem:modulartauberianFJ}.

\section{Adding conserved currents}\label{sec:currents}
In our previous analysis, we assumed that the 2D CFT has a nonzero twist gap $\tau_{\rm gap}>0$, and that the vacuum state is the only twist-0 primary state. However, we believe that the argument can be straightforwardly generalized to include 2D CFTs with nontrivial twist-0 primaries, which correspond to conserved chiral currents.\footnote{We thank Nathan Benjamin for a relevant discussion on this point.}

To account for the presence of these twist-0 primaries, we modify our general ansatz for the partition function. While maintaining the assumption of a nonzero twist gap $\tau_{\rm gap}$ above the twist-0 primaries, the modified partition function ansatz becomes:
\begin{equation}\label{Zansatz:withcurrents}
	\begin{split}
		Z(\beta_L,\beta_R)=&\chi_0(\beta_L)\chi_0(\beta_R)+\sum\limits_{h,\bar{h}\geqslant\tau_{\rm gap}/2}n_{h,\bar{h}}\ \chi_h(\beta_L)\chi_{\bar{h}}(\beta_R) \\
		&+\sum\limits_{j=1}^{\infty}D(j)\left[\chi_{j}(\beta_L)\chi_0(\beta_R)+\chi_0(\beta_L)\chi_{j}(\beta_R)\right], \\
	\end{split}
\end{equation}
where $j$ represents the spin of the current, and $D(j)$ denotes the number of left/right currents with spin $j$ . We assume parity symmetry for simplicity, which implies that for each $j$, there are an equal number of left and right spin-$j$ currents.

We aim to provide separate comments on the cases of finitely many currents and infinitely many currents, and propose conjectures that generalize, or partially generalize, the results presented in section \ref{section:modularbootstrap}. Although the following arguments lack rigor, we anticipate that they can be established using a similar analysis to section \ref{section:modularbootstrap}.

So far we are not aware of any known examples of CFTs that can be used to verify the consistency of (a) the case of finitely many currents (section \ref{section:finitecurrents}) or (b) the case of infinitely many currents with a "slow growth" of the current degeneracy $D(j)$ (section \ref{section:infinitenoshift}). However, many examples exist where we can examine the last case: infinitely many currents with a "critical growth" of the current degeneracy $D(j)$ (section \ref{section:shifttwist}). We provide three examples for the last case, and the technical details can be found in appendix \ref{appendix:example}.

Nevertheless, a comprehensive investigation is still required in order to fully understand and validate these conjectures. We leave it for future work.

\subsection{Finite number of currents}\label{section:finitecurrents}
Consider the case where the 2D CFT has finitely many twist-0 primaries (chiral currents). In this scenario, the modified partition function, given by eq.\,\eqref{Zansatz:withcurrents}, reduces to:
\begin{equation}
	\begin{split}
		Z(\beta_L,\beta_R)=&\chi_0(\beta_L)\chi_0(\beta_R)+\sum\limits_{h,\bar{h}\geqslant\tau_{\rm gap}/2}n_{h,\bar{h}}\ \chi_h(\beta_L)\chi_{\bar{h}}(\beta_R) \\
		&+\sum\limits_{n=1}^{N}\left[\chi_{j_n}(\beta_L)\chi_0(\beta_R)+\chi_0(\beta_L)\chi_{j_n}(\beta_R)\right], \\
	\end{split}
\end{equation}
where $j_n\in\mathbb{Z}_+$ represent the spins of the currents, and $N$ denotes the total number of currents.

In comparison to the asymptotic behavior of the partition function without currents, the essential difference here is that the current characters, instead of the vacuum character, dominate the dual channel in the DLC$_w$ limit (\ref{def:DLCwlimit}). This change arises because in the dual channel, the $\beta_L$-dependent part of the current character approaches 1, while that of the vacuum character approaches zero in the double lightcone limit. Specifically, we have:
\begin{equation}
	\begin{split}
		{\rm vacuum:}\  e^{\frac{4\pi^2A}{\beta_{L}}}\left(1-e^{-\frac{4\pi^2}{\beta_{L}}}\right)\sim\frac{4\pi^2}{\beta_{L}}\quad vs\quad{\rm current:}\ e^{(A-j)\frac{4\pi^2}{\beta_{L}}}\sim1.
	\end{split}
\end{equation}
Consequently, in the dual channel, we find that:
\begin{equation}
	\begin{split}
		\tilde{Z}(\beta_L,\beta_R)\stackrel{{\rm DLC}_w}{\sim}&\sqrt{\frac{4\pi^2}{\beta_{L}\beta_{R}}}e^{A(\frac{4\pi^2}{\beta_{L}}+\frac{4\pi^2}{\beta_{R}})}\sum\limits_{n=1}^{N}e^{-\frac{4\pi^2j_n}{\beta_{L}}}\left(1-e^{-\frac{4\pi^2}{\beta_{R}}}\right) \\
		\sim&\frac{2\pi N}{(\beta_L\beta_R)^{1/2}}e^{\frac{4\pi^2A}{\beta_R}}. \\
	\end{split}
\end{equation}
This indicates that only the slow-growth factor (the power-growth factor) changes when finitely many currents are added to the partition function. Therefore, the argument presented in section \ref{section:twistaccum} remains valid. The point $(h=A,\bar{h}=\infty)$ continues to be an accumulation point in the spectrum of Virasoro primaries, and the same holds true with $h$ and $\bar{h}$ interchanged.

Moreover, the analysis conducted in sections \ref{section:modular:idea} - \ref{section:epsilonwindow} still applies, but results in table \ref{table:modularestimate} are slightly modified: the exponential factors remain the same, but the power-law indices in the estimates will change due to the presence of the currents.

Therefore, in this case, we expect similar results to theorem \ref{theorem:modulartauberianFJ}, with two modifications:
\begin{itemize}
	\item The arguments in the two-sided bounds \eqref{eq:resultFiniteEpsilon} are modified as follows:
	\begin{equation}
		\begin{split}
			\frac{\mathcal{A}_J(\beta_L,\varepsilon )}{4\pi^{5/2}\beta_L^{-3/2}J^{-1/2}e^{4\pi\sqrt{AJ}}}\ \longrightarrow&\ \frac{\mathcal{A}_J(\beta_L,\varepsilon )}{N\sqrt{\frac{\pi}{\beta_L J}}e^{4\pi\sqrt{AJ}}}. \\
		\end{split}
	\end{equation}
    \item The allowed lower bound on $\varepsilon$ has a similar structure, but the coefficients in front of the logarithmic terms in \eqref{epsilon:choice:DLCJ} will change.
\end{itemize}
Then for $\mathcal{N}_J(\varepsilon)$ in the large spin limit, we still have
\begin{equation}
	\begin{split}
		\mathcal{N}_J(\varepsilon\equiv\kappa J^{-1/2}\log J)\stackrel{J\rightarrow\infty}{\sim}e^{4\pi\sqrt{AJ}+f_\kappa(J)},\quad \abs{f_\kappa(J)}\leqslant C_1\log( J+1)+C_2(\kappa),
	\end{split}
\end{equation}
where the allowed $\kappa$ and the constants $C_1$ and $C_2$ will be different from corollary \ref{cor:operatorcount}.

\subsection{Infinitely many currents}\label{sec:Infinite currents}
Now let us consider a CFT with infinitely many currents. In this case, the partition function is given by \eqref{Zansatz:withcurrents} with an infinite sum over the current spin $j$. For convenience, we define
\begin{equation}\label{def:Fcurrent}
	\begin{split}
		\mathcal{F}_{\rm current}(\beta):=\sum\limits_{j=1}^{\infty}D(j)e^{-j \beta}\equiv\sum\limits_{j=1}^{\infty}D(j)q^j\quad(q=e^{-\beta}).
	\end{split}
\end{equation}
We can rewrite $\tilde{Z}$ using \eqref{Zansatz:withcurrents} as:
\begin{equation}
	\begin{split}
		\tilde{Z}(\beta_L,\beta_R)=&e^{A(\beta_{L}+\beta_{R})}\Big{[}\left(1-e^{-\beta_{L}}\right)\left(1-e^{-\beta_{R}}\right)+\sum\limits_{h,\bar{h}\geqslant\tau_{\rm gap}/2}e^{-h\beta_{L}-\bar{h}\beta_{R}} \\
		&+\mathcal{F}_{\rm current}(\beta_{L})\left(1-e^{-\beta_{R}}\right)+\left(1-e^{-\beta_{L}}\right)\mathcal{F}_{\rm current}(\beta_{R})\Big{]}. \\
	\end{split}
\end{equation}
It is important to note that the growth of $D(j)$, the number of left/right currents with spin $j$, cannot be arbitrarily fast. Here we make the following ansatz:
\begin{equation}\label{D:ansatz}
	\begin{split}
		D(j)=f(j)e^{4\pi a j^b}\quad(a>0,b\geqslant0),
	\end{split}
\end{equation}
where $f(j)$ represents the component of slow growth compared to the exponential factor.\footnote{Here, by ``slow growth" of $f(j)$ we mean that for any $\delta>0$, $f(j)e^{-\delta\,j^b}\rightarrow0$ in the limit $j\rightarrow\infty$.} Additionally, for simplicity, we assume that $f(j)$ is bounded from below by some positive constant.

The range of the allowed $a$ and $b$ is constrained by the Cardy growth \cite{Cardy:1986ie}, which states that:
\begin{equation}\label{D:cardybound}
	\begin{split}
		D(j)\lesssim g(j) e^{4\pi\sqrt{Aj}}\quad(j\rightarrow\infty),
	\end{split}
\end{equation}
where $g(j)$ is some factor of slow growth. This bound can be derived from modular invariance by considering reduced partition function $\tilde{Z}(\beta_{L},\beta_{R})$ in the limit $\beta_{L}\rightarrow\infty$ with $\beta_{R}$ fixed. It follows that $b$ should be in the range:
\begin{equation}
	\begin{split}
		b\leqslant\frac{1}{2},
	\end{split}
\end{equation}
In addition, at criticality $b=\frac{1}{2}$, the parameter $a$ should be in the range:
\begin{equation}
	\begin{split}
		a\leqslant\sqrt{A}.
	\end{split}
\end{equation}
\begin{remark}
	Using the ansatz \eqref{D:ansatz}, neglecting the slow-growing factor $f(j)$ and performing a saddle-point approximation in $j$ in \eqref{def:Fcurrent}, we find that in the limit $\beta_{L}\rightarrow\infty$, $\mathcal{F}_{\text{current}}\left(\frac{4\pi^2}{\beta_L}\right)$ exhibits the following growth behavior:
	\begin{equation}\label{Fcurrent:asymp:b}
		\begin{split}
			\mathcal{F}_{\rm current}\left(\frac{4\pi^2}{\beta_{L}}\right)\sim\sum\limits_{j=1}^{\infty}e^{-j \frac{4\pi^2}{\beta_{L}}+4\pi aj^b}\sim e^{\#\beta_{L}^{\frac{b}{1-b}}}\quad(\beta_L\rightarrow\infty).
		\end{split}
	\end{equation}
	Using \eqref{def:Fcurrent} together with the derived asymptotic behavior \eqref{Fcurrent:asymp:b}, we can also argue why $b>\frac{1}{2}$ is not allowed. Firstly, for $b\geqslant1$, the sum in \eqref{def:Fcurrent} diverges when $\beta<4\pi a$. However, such a divergence contradicts the requirement of a well-defined torus partition function. Secondly, for $b \in \left(\frac{1}{2},1\right)$, we have $\frac{b}{1-b} > 1$ in \eqref{Fcurrent:asymp:b}. Consequently, if we consider the limit $\beta_{L}\rightarrow\infty$ with $\beta_{R}$ fixed for $\tilde{Z}(\beta_{L},\beta_{R})$, the contribution solely from the left currents in the dual channel grows faster in $\beta_L$ than the total contribution in the direct channel, which is at most $e^{A\beta_{L}}$. This contradiction rules out the possibility of $b\in\left(\frac{1}{2},1\right)$.\footnote{This argument holds irrespective of the twist-gap condition $\tau_{\text{gap}}>0$ and is also valid for the case $\tau_{\text{gap}}=0$.}
\end{remark}
In the subsequent subsections, we will analyze the following three cases separately:
\begin{itemize}
	\item $0 \leqslant b < \frac{1}{2}$;
	\item $b=\frac{1}{2}$ and $0<a<\sqrt{A}$;
	\item $b = \frac{1}{2}$ and $a=\sqrt{A}$.
\end{itemize}
We will see that they exhibit different behaviors in the double lightcone limit.

\subsubsection{Case $0\leqslant b<\frac{1}{2}$}\label{section:infinitenoshift}
In the case when $0\leqslant b<\frac{1}{2}$, in the double lightcone limit, we expect the partition function $\tilde{Z}(\beta_{L},\beta_{R})$ to be dominated by the left-current characters in the dual channel:
\begin{equation}\label{Ztilde:dualasymp:withcurrent}
	\begin{split}
		\tilde{Z}(\beta_{L},\beta_{R})\sim\sqrt{\frac{4\pi^2}{\beta_{L}\beta_{R}}}\mathcal{F}_{\rm current}\left(\dfrac{4\pi^2}{\beta_{L}}\right)e^{\frac{4\pi^2}{\beta_{R}}},
	\end{split}
\end{equation}
where the asymptotic behavior of the current factor $\mathcal{F}_{\rm current}\left(\dfrac{4\pi^2}{\beta_{L}}\right)$ is given by (\ref{Fcurrent:asymp:b}) (up to a slow-growing factor). By comparing (\ref{Ztilde:dualasymp:withcurrent}) with (\ref{Ztilde:asympdual:withoutcurrent}), we note that only the growth in $\beta_{L}$ changes, while the exponential factor $e^{\frac{4\pi^2}{\beta_{R}}}$, which dominates the growth in the double lightcone limit, remains the same.

In the direct channel of $\tilde{Z}(\beta_{L},\beta_{R})$, our expectation is that the contributions from the vacuum, currents, low twists, and high twists are still subleading compared to the r.h.s. of \eqref{Ztilde:dualasymp:withcurrent}. Thus, the analysis presented in section \ref{section:twistaccum} remains applicable, indicating that $(h,\bar{h})=(A,\infty)$ continues to be an accumulation point in the spectrum of Virasoro primaries.

To estimate the quantity $\mathcal{A}_J(\beta_{L},\varepsilon)$, we can follow a similar analysis as described in sections \ref{section:modular:idea} - \ref{section:epsilonwindow}. However, due to the presence of infinite currents, we expect additional factors to appear in the estimates:
$$e^{\#\beta_{L}^{\frac{b}{1-b}}},\quad e^{\#\beta_{R}^{-\frac{b}{1-b}}},\quad e^{\#\beta_{L}^{-\frac{b}{1-b}}},\quad e^{\#\beta_{R}^{\frac{b}{1-b}}}.$$
While the first two terms may exhibit unbounded growth, it is important to note that $\frac{b}{1-b}<1$ for the range $b<\frac{1}{2}$. Hence, we anticipate that these subexponential-growing factors will have a negligible impact on the estimates, except when considering the shrinking of the $(A-\varepsilon,A+\varepsilon)$ range towards zero.

Moreover, it is necessary to establish that the contributions from (a) right-current characters in the dual channel and (b) left- and right-current characters in the direct channel also become subleading in the double lightcone limit. 

Based on the aforementioned argument, we expect the following key differences compared to table \ref{table:modularestimate}:
\begin{itemize}
	\item The dominant term is now $I^{\rm dual}_{\pm,{\rm left}}$, which corresponds to the contribution from left currents in the dual channel.
	\item The power indices in the prefactors will change.
	\item There are additional estimates on $I^{\rm dual}_{\pm,{\rm right}}$, $I_{\pm,{\rm left}}$ and $I_{\pm,{\rm right}}$. These contributions are expected to be subleading.
	\item In the estimate of the low-twist contribution, an additional factor $e^{\#\beta_{R}^{-\frac{b}{1-b}}}$ will appear in the final result. This arises from the contribution of the left currents $\mathcal{F}_{\rm currents}\left(\frac{4\pi^2}{\beta_{L}'}\right)$ in the analysis of \eqref{ineq:Zlowtwist}.
\end{itemize}
The first and second differences are already present in the case of finitely many currents. However, the third difference, which accounts for the growth of the current degeneracy $D(j)$, significantly impacts the lower bound of $\varepsilon$ allowed by theorem \ref{theorem:modulartauberianFJ}. Specifically, the decay behavior of $\varepsilon_{\rm min}$ changes from $J^{-1/2}\log J$ to $J^{-\frac{1-2b}{2(1-b)}}$, which decays more slowly as $J$ increases.

Therefore, we propose the following conjecture that generalizes theorem \ref{theorem:modulartauberianFJ} and corollary \ref{theorem:modulartauberianFJ}:
\begin{conjecture}\label{conjecture:infinitespin:noshift}
	For a torus partition function of the form (\ref{Zansatz:withcurrents}), assuming the presence of at least one chiral current, if the current density $D(j)$ satisfies the upper bound
	\begin{equation}\label{Dbound:noshift}
		\begin{split}
			D(j)\leqslant N_0e^{a j^b}
		\end{split}
	\end{equation}
where $a>0$, $0\leqslant b<\frac{1}{2}$, and $N_0<+\infty$, then the following statements hold:
\begin{enumerate}
	\item There exists a family of Virasoro primaries $\left\{(h_n,\bar{h}_n)\right\}_{n\in\mathbb{N}}$ with $h_n\rightarrow A$ and $\bar{h}_n\rightarrow\infty$.
	\item Theorem \ref{theorem:modulartauberianFJ} still holds with the following modifications:
	\begin{itemize}
		\item The argument in the two-sided bounds \eqref{eq:resultFiniteEpsilon} is modified to
		\begin{equation}
			\frac{\mathcal{A}_J(\beta_L,\varepsilon )}{4\pi^{5/2}\beta_L^{-3/2}J^{-1/2}e^{4\pi\sqrt{AJ}}}\ \longrightarrow\ \frac{\mathcal{A}_J(\beta_L,\varepsilon)}{\mathcal{F}_{\rm current}\left(\frac{4\pi^2}{\beta_{L}}\right)\sqrt{\frac{\pi}{\beta_{L}J}}e^{4\pi\sqrt{AJ}}}.
		\end{equation}
	\item The lower bound for the allowed $\varepsilon$ is modified to
	\begin{equation}\label{epsilonbound:noshift}
		\begin{split}
			\varepsilon\geqslant\max\left\{\alpha_1 J^{-\frac{1-2b}{2(1-b)}}+\alpha_2 J^{-1/2}\log{J},\ 
			\alpha_3\beta_{L}^{-1}\log{J}+\alpha_4\beta_{L}^{-1}\log\beta_{L}\right\}.
		\end{split}
	\end{equation}
    Here the constants $\alpha_i$ are finite. They depend on $A$, $T$, $w$, and the parameters in \eqref{Dbound:noshift} (i.e.\,$N_0$, $a$, and $b$).
	\end{itemize}
    \item Corollary \ref{cor:operatorcount} still holds with the following modifications:
    \begin{itemize}
    	\item For $\mathcal{N}_J(\varepsilon)$, the choice of $\varepsilon$ is modified to $\varepsilon=\kappa J^{-\frac{1-2b}{2(1-b)}}$, where $\kappa$ has a finite, positive lower bound that depends on $A$, $T$, $N_0$, $a$ and $b$.
    	\item We have
    	\begin{equation}\label{NJ:result:noshift}
    		\begin{split}
    			\mathcal{N}_J\left(\varepsilon\equiv\kappa J^{-\frac{1-2b}{2(1-b)}}\right)=J^{-3/4}e^{4\pi\sqrt{AJ}+f_\kappa(J)},
    		\end{split}
    	\end{equation}
        where the error term $f\kappa(J)$ satisfies the bound
        \begin{equation}\label{error:result:noshift}
        	\begin{split}
        		f_\kappa(J)\leqslant C(\kappa,A,T,a,b)(J+1)^{\frac{b}{2(1-b)}},
        	\end{split}
        \end{equation}
       and the constant $C$ is finite.
    \item When $b=0$, an additional $\log J$ factor should be included in the choice of $\varepsilon$ and the bound on the error term $f_\kappa(J)$.
    \end{itemize}
\end{enumerate}
\end{conjecture}
\begin{remark}
This conjecture also encompasses the case of finite currents ($b=0$), where we have
	\begin{equation}
		\begin{split}
			 \mathcal{F}_{\rm currents}\left(\frac{4\pi^2}{\beta_{L}}\right)\stackrel{\beta_{L}\rightarrow\infty}{\longrightarrow}\mathcal{F}_{\rm currents}(0)={\rm number\ of\ left\ currents}.
		\end{split}
	\end{equation}
\end{remark}
The proof of conjecture \ref{conjecture:infinitespin:noshift} is left for future work.

\subsubsection{Case $b=\frac{1}{2}$ preview: the shifted twist accumulation point}\label{section:shifttwist}
Now we turn our attention to the critical case: $b=\frac{1}{2}$. In this case, the ansatz \eqref{D:ansatz} of $D(j)$ becomes:
\begin{equation}\label{Nj:ansatz:shift}
	\begin{split}
		D(j)=f(j)e^{4\pi a\sqrt{j}},
	\end{split}
\end{equation}
where $a$ is a strictly positive parameter. We still assume that $f(j)$ grows slowly compared to any exponential growth of the form $e^{\#\sqrt{j}}$. It is worth noting that there exist many known examples of CFTs that fall into this category (see sections \ref{section:twistshift1} and \ref{section:twistshift2} below, and appendix \ref{appendix:example} for more details).
\begin{remark}
	Before proceeding with the discussion, it is important to note that the ansatz \eqref{Nj:ansatz:shift} does not encompass all possible cases that remain from conjecture \ref{conjecture:infinitespin:noshift}. Recall that the conjecture covers situations where the growth of the spin degeneracy $D(j)$ is bounded by $e^{a\,j^b}$ for some $b<\frac{1}{2}$. However, there may be other scenarios that fall outside the scope of \eqref{Nj:ansatz:shift}. For instance, there could exist CFTs with $D(j)\sim e^{\sqrt{j}/\log j}$. Such cases cannot be accommodated within the framework of \eqref{Nj:ansatz:shift} with a positive $a$. At present, we do not have a systematic approach to address the most general situation, and we leave it for future work.
\end{remark}
For CFTs with a spin degeneracy $D(j)$ described by the form \eqref{Nj:ansatz:shift}, an intriguing phenomenon occurs: the twist accumulation point is shifted by $a^2$:
\begin{equation}\label{h:shift}
	\begin{split}
		h=A-a^2,\quad\bar{h}=\infty.
	\end{split}
\end{equation}
This shift arises because in the double lightcone limit, there is an additional factor $e^{a^2\beta_{L}}$ in the dual channel of $\tilde{Z}(\beta_{L},\beta_{R})$, which corresponds to the cumulative contribution of current operators:
\begin{equation}
	\begin{split}
		\mathcal{F}_{\rm current}\left(\frac{4\pi^2}{\beta_{L}}\right)=&\sum\limits_{j=1}^{\infty}D(j)e^{-\frac{4\pi^2 j}{\beta_{L}}} \\
		\stackrel{\beta_{L}\rightarrow\infty}{\sim}&\int_0^{\infty}d j\ e^{4\pi a\sqrt{j}-\frac{4\pi^2j}{\beta_{L}}} \\
		\sim&e^{a^2\beta_{L}}. \\
	\end{split}
\end{equation}
The estimate above is not precise as we neglected the slow-growing factor $f(j)$. Taking into account $f(j)$, we expect $\mathcal{F}_{\rm current}$ to have the following asymptotic behavior:
\begin{equation}\label{Fcurrent:asymp:withshift}
	\begin{split}
		\mathcal{F}_{\rm current}\left(\beta\right)\stackrel{\beta\rightarrow0}{\sim}e^{\frac{4\pi^2a^2}{\beta}+o\left(\frac{1}{\beta}\right)}.
	\end{split}
\end{equation}
In the double lightcone limit, the contribution from left currents dominates the dual channel of $\tilde{Z}(\beta_{L},\beta_{R})$, leading to:
\begin{equation}
	\begin{split}
		\tilde{Z}(\beta_{L},\beta_{R})\stackrel{{\rm DLC}_w}{\sim}\sqrt{\frac{4\pi^2}{\beta_{L}\beta_{R}}}e^{a^2\beta_{L}+\frac{4\pi^2A}{\beta_{R}}+o(\beta_{L})}.
	\end{split}
\end{equation}
Because of the presence of the extra factor $e^{a^2\beta_{L}}$, to make the contribution from the left currents dominant, the double lightcone limit is not necessarily as strong as \eqref{def:DLClimit}. Here we modify the DLC$_w$ limit \eqref{def:DLCwlimit} to:
\begin{equation}\label{redef:DLCw}
	\begin{split}
		{\rm DLC}_w\ {\rm limit}:\quad&\beta_{L}\rightarrow\infty,\quad\beta_{R}\rightarrow0,\quad\frac{4\pi^2T(1-w^2)}{(A-a^2)\beta_{R}}-\beta_{L}\rightarrow\infty, \\
		&\beta_{L}^{-1}\log\beta_{R}\rightarrow0.
	\end{split}
\end{equation}
Then to match the correct exponential growth in $\beta_{L}$ in the direct channel of $\tilde{Z}(\beta_{L},\beta_{R})$, we need contribution from operators with
\begin{equation}
	\begin{split}
		e^{(A-h)\beta_{L}}\sim e^{a^2\beta_{L}}\quad(\beta_{L}\rightarrow\infty).
	\end{split}
\end{equation}
This leads to a rough but quick guess of the position of the twist accumulation point, given by \eqref{h:shift}. We anticipate that by employing a similar technique as in \cite{Pal:2022vqc} and assuming that $\mathcal{F}_{\rm current}(\beta)$ exhibits asymptotic behavior \eqref{Fcurrent:asymp:withshift}, it may be possible to rigorously prove the existence of the twist accumulation point at $(h=A-a^2,\bar{h}=\infty)$.

Based on \eqref{h:shift} and the unitarity constraint ($h\geqslant0$), we find that $a$ cannot exceed $\sqrt{A}$. This is consistent with the "Cardy growth" argument presented around eq.\,\eqref{D:cardybound} and gives rise to two possibilities: (1) $0<a<\sqrt{A}$ and (2) $a=\sqrt{A}$ (as mentioned at the beginning of section \ref{sec:Infinite currents}), which will be discussed separately in the next two subsections. 

\subsubsection{Case $b=\frac{1}{2}$ and $0<a<\sqrt{A}$}\label{section:twistshift1}
In the case of $0<a<\sqrt{A}$, we can choose a sufficiently small (but nonzero) $\varepsilon$ such that the interval $(A-a^2-\varepsilon,A-a^2+\varepsilon)$ does not contain 0. Consequently, for Virasoro primaries with $h\in(A-a^2-\varepsilon,A-a^2+\varepsilon)$ and large $\bar{h}$ (i.e.\,large spin), their Virasoro characters correspond to the second case of \eqref{def:Vircharacter} without any subtraction. Therefore, in this scenario, we can essentially apply the same arguments as in section \ref{section:modularbootstrap}, but with slight modifications in the definitions of $\mathcal{N}_J$, $\mathcal{A}_J$:
\begin{equation}\label{redef:AJ}
	\begin{split}
		\mathcal{N}_J(\varepsilon):=&\sum\limits_{h\in(A-a^2-\varepsilon,A-a^2+\varepsilon)}n_{h,h+J}, \\ \mathcal{A}_J(\beta_L,\varepsilon):=&\sum\limits_{h\in(A-a^2-\varepsilon,A-a^2+\varepsilon)}n_{h,h+J}e^{-(h-A+a^2)\beta_L}, \\
	\end{split}
\end{equation}
and the double lightcone limit which is compatible with the current version \eqref{redef:DLCw}:
\begin{equation}\label{redef:DLCwJ}
	\begin{split}
		\mathrm{DLC}_w\ \mathrm{limit}:&\quad\beta_L,J\rightarrow\infty,\quad \frac{2\pi T(1-w^2)}{A-a^2}\sqrt{\frac{J}{A}}-\beta_L\rightarrow\infty\,,\\
		&\quad \beta_L^{-1}\log J\to 0\,.
	\end{split}
\end{equation}
Here we propose the following conjecture:
\begin{conjecture}\label{conjecture:infinitespin:shift}
	For a partition function of the form (\ref{Zansatz:withcurrents}), where the current degeneracy $D(j)$ is given by (\ref{Nj:ansatz:shift}) with $a\in(0,\sqrt{A})$, the following statements hold:
	\begin{enumerate}
		\item There exists a family of Virasoro primaries $\left\{(h_n,\bar{h}_n)\right\}_{n\in\mathbb{N}}$ with $h_n\rightarrow A-a^2$ and $\bar{h}_n\rightarrow\infty$.
		\item For any $w\in\left(\frac{1}{2},1\right)$ and $\varepsilon\in(0,A-a^2)$ fixed, the quantity $\mathcal{A}_J(\beta_L,\varepsilon)$ redefined in \eqref{redef:AJ} satisfies the following asymptotic two-sided bounds in the redefined DLC$_w$ limit \eqref{redef:DLCwJ}:\footnote{The denominator in the r.h.s. is actually $\mathcal{F}_{\rm current}\left(\frac{4\pi^2}{\beta_{L}}\right)\sqrt{\frac{\pi}{\beta_{L}(J-a^2)}}e^{4\pi\sqrt{A(J-a^2)}}$, but it makes no difference in the limit $J\rightarrow\infty$.}
		\begin{equation}\label{replace:shift}
			\begin{split}
				\frac{1}{w}\frac{1}{1-\frac{\tan\left(\pi w(1-\varepsilon)\right)}{\pi w(1-\varepsilon)}}\lesssim \frac{\mathcal{A}_J(\beta_L,\varepsilon)}{f_{\rm current}\left(\beta_{L}\right)\sqrt{\frac{\pi}{\beta_{L}J}}e^{4\pi\sqrt{AJ}}} \lesssim \frac{1}{w}\frac{2}{1+\frac{\sin(2\pi w\varepsilon)}{2\pi w\varepsilon}},
			\end{split}
		\end{equation}
	    where $f_{\rm current}(\beta)$ is the slow-growing factor in $\mathcal{F}_{\rm current}$:
	    \begin{equation}
	    	\begin{split}
	    		\mathcal{F}_{\rm current}\left(\frac{4\pi^2}{\beta}\right)=f_{\rm current}(\beta)e^{a^2\beta}.
	    	\end{split}
	    \end{equation}
	\end{enumerate}
\end{conjecture}
We have not made a conjecture regarding the lower bound of the allowed $\varepsilon$ since it depends on the slow-growing factor $f(j)$ in \eqref{Nj:ansatz:shift}. Additionally, $f(j)$ will affect the $\varepsilon$-window of $\mathcal{N}_J(\varepsilon)$. We expect that by choosing proper $\varepsilon\rightarrow0$ limit according to the growth of $f(j)$, $\mathcal{N}_J(\varepsilon)$ will behave as
\begin{equation}\label{operatorcounting:shift}
	\begin{split}
		\mathcal{N}_J(\varepsilon)=e^{4\pi\sqrt{AJ}+\ldots},
	\end{split}
\end{equation}
where ``$\ldots$" denotes the subleading part in the exponent.

To examine that our conjecture is reasonable, we would like to provide two examples.
\newline
\newline\textbf{Example 1: decoupled irrational CFTs}
\newline There are very few known examples for the case of theorem \ref{theorem:modulartauberianFJ} (e.g.\,the construction of partition function in \cite{Benjamin:2020mfz}). However, given at least one of such examples, we can make infinitely many examples for the case of conjecture \ref{conjecture:infinitespin:shift}. The construction is simply several decoupled copies of unitary CFTs:
$${\rm CFT}={\rm CFT}^{(1)}\otimes{\rm CFT}^{(2)}\otimes\ldots\otimes{\rm CFT}^{(N)}.$$
We assume that each CFT$^{(i)}$ has central charge $c^{(i)}>1$, a unique vacuum and a nonzero twist gap $\tau^{(i)}_{\rm gap}$. Then the big CFT has the following features:
\begin{itemize}
	\item The central charge is given by $c=\sum\limits_{i=1}^{N}c^{(i)}$ (i.e.\,$A=\sum\limits_{i=1}^{N}A^{(i)}+\frac{N-1}{24}$ where $A^{(i)}=\frac{c^{(i)}-1}{24}$).
	\item The theory has a unique vacuum and infinitely many conserved chiral currents.
	\item The theory has a twist gap for the Virasoro primaries with nonzero twists, given by $\tau_{\rm gap}=\min\{4,\tau_{\rm gap}^{(1)},\ldots,\tau_{\rm gap}^{(N)}\}$.
	\item When $j$ is large, the growth of the current degeneracy is given by $D(j)\sim e^{4\pi\sqrt{\frac{N-1}{24}j}}$ up to a slow-growing factor. I.e. this example corresponds to $b=\frac{1}{2}$ and $a=\sqrt{\frac{N-1}{24}}$.
	\item It has a twist accumulation point of Virasoro primaries, given by
	\begin{equation}
		\begin{split}
			h=\sum\limits_{i=1}^{N}A^{(i)}\quad\left(A^{(i)}\equiv\frac{c^{(i)}-1}{24}\right),\quad\bar{h}=\infty.
		\end{split}
	\end{equation}
    \item The quantity $\mathcal{A}_J(\varepsilon,\beta_{L})$, defined in \eqref{redef:AJ}, is expected to have the following growth in the redefined DLC$_w$ limit \eqref{redef:DLCwJ}:
    \begin{equation}
    	\begin{split}
    		\mathcal{A}_J(\varepsilon,\beta_{L})\stackrel{{\rm DLC}_w}{\approx}e^{4\pi\sqrt{AJ}},
    	\end{split}
    \end{equation}
    up to some factor of slow growth.
\end{itemize}
The first five points can be justified rigorously, while the last point is argued with some additional hypotheses that we consider natural. We leave the technical details to appendix \ref{app:decoupledCFT}.

According to the first four points, this example falls within the scope of the cases discussed in this subsection, making it a suitable test-bed for examining the consistency of conjecture \ref{conjecture:infinitespin:shift}. 

The first, fourth, and fifth points indicate that the CFT exhibits a shifted twist accumulation point precisely located at $h=A-a^2$, aligning with the first part of conjecture \ref{conjecture:infinitespin:shift}. The final point demonstrates the correct exponential growth of $\mathcal{A}_J(\varepsilon,\beta_{L})$ in the DLC$_w$ limit, supporting the validity of the second part of conjecture \ref{conjecture:infinitespin:shift}.
\newline
\newline\textbf{Example 2: $W_N$ CFT}
\newline The second example is the unitary $W_N$ CFT with central charge $c>N-1$ and a twist gap $\tau_{\rm gap}^{W_N}>0$ in the spectrum of $W_N$-primaries \cite{Afkhami-Jeddi:2017idc}.

The $W_N$ algebra is an extension of the Virasoro algebra, which implies that the $W_N$ CFT possesses a more fine-tuned spectrum and dynamics to satisfy the constraints imposed by the $W_N$ algebra. However, for the purpose of our discussion, let us momentarily set aside the $W_N$ algebra and concentrate solely on the Virasoro algebra. In this context, the theory exhibits the following features from Virasoro algebra perspective:
\begin{itemize}
	\item The theory has a unique vacuum and infinitely many conserved chiral currents.
	\item The growth of the current degeneracy $D(j)$ (see \eqref{Zansatz:withcurrents}) is given by
	\begin{equation}
		\begin{split}
			D(j)\sim e^{4\pi\sqrt{\frac{N-1}{24}j}}\quad(j\rightarrow\infty)
		\end{split}
	\end{equation}
	up to a factor of slow growth in $j$.
	\item The theory has a twist gap for the Virasoro primaries with nonzero twists, given by $\tau_{\rm gap}=\min\left\{6,\tau_{\rm gap}^{W_N}\right\}$.
	\item The theory has a twist accumulation point of Virasoro primaries, given by
	\begin{equation}
		\begin{split}
			h=A-\frac{N-2}{24},\quad\bar{h}=\infty.
		\end{split}
	\end{equation} 
	\item The quantity $\mathcal{N}_J(\varepsilon)$ and $\mathcal{A}_J(\varepsilon,\beta_{L})$, defined in \eqref{redef:AJ}, are expected to have the following growth:
	\begin{equation}
		\begin{split}
			\mathcal{A}_J(\varepsilon,\beta_{L})\stackrel{{\rm DLC}_w}{\sim}e^{4\pi\sqrt{AJ}},\quad\mathcal{N}_J(\varepsilon\equiv\kappa J^{-1/2}\log J)\stackrel{J\rightarrow\infty}{\sim}e^{4\pi\sqrt{AJ}},
		\end{split}
	\end{equation}
	up to some factor of slow growth.
\end{itemize}
We provide the technical details of the above claims in appendix \ref{app:WNCFT}. The first three points can be rigorously proven. Consequently, this example falls under the case described in Conjecture \ref{conjecture:infinitespin:shift}. 

We believe that the fourth point, consistent with the first part of conjecture \ref{conjecture:infinitespin:shift}, can be demonstrated using a similar technique as in \cite{Pal:2022vqc}. 

The setup in the $W_N$ case exhibits significant similarities to the standard CFT case, leading us to propose conjecture \ref{conjecture:modulartauberianFJ:WN} in appendix \ref{app:WNconjecture}, which extends theorem \ref{theorem:modulartauberianFJ} to $W_N$ CFTs. 

Finally, the last point is consistent with the second part of conjecture \ref{conjecture:infinitespin:shift}. Its validity can be argued by considering the incorporation of some additional natural hypotheses, such as assuming the correctness of conjecture \ref{conjecture:modulartauberianFJ:WN} for $W_N$ CFTs.

\subsubsection{Case $b=\frac{1}{2}$ and $a=\sqrt{A}$}\label{section:twistshift2}
For the second possibility ($a=\sqrt{A}$), it is clear that there exists a twist accumulation point at $h=0$ due to parity symmetry and the existence of infinitely many left currents, which implies the existence of infinitely many right currents. This case can be realized by considering multiple copies of $c<1$ unitary minimal models (as demonstrated in the example of 3 copies of Ising CFTs in section \ref{appendix:example}).

Moreover, when $a=\sqrt{A}$, the counting of Virasoro primaries around the accumulation point $(h=0,\bar{h}=\infty)$ becomes straightforward. Since we assumed a nonzero twist gap $\tau_{\rm gap}$ for Virasoro primaries with twists different from zero, choosing $\varepsilon<\tau_{\rm gap}/2$ allows us to count the total number of spin-$J$ Virasoro primaries within the window $h\in[0,\varepsilon)$. In this case, these Virasoro primaries can only be right currents (i.e.\,$h=0$). Therefore, the counting simplifies to:
\begin{equation}\label{result:alphaequalA}
	\begin{split}
		\mathcal{N}_J(\varepsilon)=\mathcal{A}_J(\beta_{L},\varepsilon)={\rm \#\ of\ spin}-J\ {\rm right\ currents}=D(J)\quad (\varepsilon<\tau_{\rm gap}/2).
	\end{split}
\end{equation}
\textbf{Example 3: three copies of Ising CFTs}
\newline The case discussed in this subsection can be realized by taking several copies of $c<1$ unitary minimal models. Here we would like to present the simplest of them, which is the three copies of Ising CFTs:
$${\rm CFT}={\rm Ising}^{(1)}\otimes{\rm Ising}^{(2)}\otimes{\rm Ising}^{(3)}.$$
This CFT has the following features:
\begin{itemize}
	\item The theory has a central charge $c=\frac{3}{2}$, i.e.\,A=$\frac{1}{48}$.
	\item It has a unique vacuum and infinitely many conserved currents.
	\item It has a twist gap $\tau_{\rm gap}=\frac{1}{8}$ for Virasoro primaries with nonzero twists, which is equal to the scaling dimension of the Ising spin field $\sigma$.
	\item For large values of $j$, the growth of the current degeneracy is approximately given by $D(j)\sim e^{4\pi\sqrt{\frac{1}{48}j}}$, up to a slow-growing factor. So in this example, we have $b=\frac{1}{2}$ and $a=\sqrt{\frac{1}{48}}$.
\end{itemize}
All of the above points can be rigorously justified by using the explicit character formulas of the Ising CFT, which are well-known. The technical details are provided in appendix \ref{app:Ising}. 

Based on these points, this example falls within the case discussed in this subsection. It is important to note that this CFT, along with the previous two examples, exhibits infinitely many twist accumulation points for Virasoro primaries. This is due to the fact that these theories are fine-tuned to have a larger symmetry algebra. However, our analysis here, which only concerns the Virasoro algebra, only allows us to conclude that the accumulation point corresponding to $h=A-a^2$ and $\bar{h}=\infty$ (or with $h$ and $\bar{h}$ interchanged) is universal.

\section{Holographic CFTs and large $c$ limit}\label{section:largec}
In this section, our focus shifts to holographic CFTs in the large central charge limit $c\rightarrow\infty$. Specifically, we consider a family of 2D irrational CFTs labeled by the central charge $c$, denoted as $\left\{\mathcal{A}_c\right\}$. Our goal here is similar to that in section \ref{section:modularbootstrap}: counting the spectrum of Virasoro primaries of these CFTs near the twist accumulation point. We consider
\begin{equation}\label{NJ:largec}
	\begin{split}
		\mathcal{N}_{J}(\varepsilon_1,\varepsilon_2,A):=\sum\limits_{A-\varepsilon_1<h<A+\varepsilon_2}n_{h,h+J},\quad A\equiv\frac{c-1}{24},
	\end{split}
\end{equation}
and take the limit $J,A\rightarrow\infty$ and $\varepsilon_1,\varepsilon_2\rightarrow0$ with appropriate constraints between $\varepsilon_1$, $\varepsilon_2$, $J$ and $A$. In contrast to section \ref{section:modularbootstrap}, here we allow $\varepsilon_1$ and $\varepsilon_2$ to differ. The reason is that in the large central charge limit, the allowed lower bounds on $\varepsilon_1$ and $\varepsilon_2$ will behave differently in $A$. 

We adopt the same basic assumptions about CFTs as in section \ref{section:modularbootstrap}. Additionally, we assume the existence of two theory-independent quantities, namely $\alpha$ and $C(\beta_{L},\beta_{R})$, satisfying the following conditions:
\begin{enumerate}
	\item For any CFT in the family $\left\{\mathcal{A}_c\right\}$, the twist gap is bounded from below by
	\begin{equation}\label{assum1}
		T\geqslant \alpha A,
	\end{equation}
where $\alpha$ is strictly positive.
	\item For any CFT in the family $\left\{\mathcal{A}_c\right\}$, given that $\beta_L,\beta_R>2\pi$, the ratio of the full partition function and its vacuum part is bounded from above by
	\begin{equation}
		\begin{split}\label{assum2}
			\frac{\tilde{Z}(\beta_L,\beta_R)}{\tilde{Z}_{\rm vac}(\beta_L,\beta_R)}\leqslant C(\beta_L,\beta_R).
		\end{split}
	\end{equation}
    where $C(\beta_{L},\beta_{R})$ is finite as long as $\beta_{L},\beta_{R}>2\pi$. 
\end{enumerate}
It is worth noting that the first assumption is not strictly necessary, as we expect our results to hold even if we include a sparse spectrum of operators below the twist gap ($\tau_{\rm gap}\equiv 2T$) in the partition function. The second assumption takes inspiration from the HKS sparseness condition \cite{Hartman:2014oaa} and its implications.\footnote{We anticipate that \eqref{assum2} can be derived by imposing a similar, yet stronger, sparseness condition compared to the HKS sparseness condition introduced in \cite{Hartman:2014oaa}. While the HKS condition focuses on the sparseness of the low-energy spectrum, \eqref{assum2} specifically requires sparseness in the low-twist spectrum.}

Our analysis in this section closely follows that in section \ref{section:modularbootstrap}. However, there are additional subtleties that arise due to the following reason. In section \ref{section:modularbootstrap}, when we decomposed the partition function into several parts in different channels:
	\begin{equation*}
		\begin{split}
			\tilde{Z}=\tilde{Z}_{\rm vac}+\tilde{Z}_{T\leqslant h\leqslant A-\varepsilon_1}+\tilde{Z}_{A-\varepsilon_1<h< A+\varepsilon_2}+\tilde{Z}_{h\geqslant A+\varepsilon_2}=\tilde{Z}^{\rm dual}_{\rm vac}+\tilde{Z}^{\rm dual}_{\rm nonvac},
		\end{split}
	\end{equation*} 
and argued that several terms are subleading, we actually bounded them by the partition function itself at some fixed inverse temperature, e.g.\ $\tilde{Z}(\beta_0,\beta_0)$. However, in the case of a family of theories, these terms may no longer be subleading as they could all grow exponentially fast with $c$ in the large central charge limit. This is the reason why we introduce the two additional conditions mentioned above: we want the dominant terms in the fixed CFT case to remain dominant in the holographic case. Exploring the generality of our additional assumptions, particularly the second one, in holographic CFTs would be an intriguing avenue for future research.

\subsection{Holographic double lightcone limit, main results}\label{section:holoCFTresults}
To estimate $\mathcal{N}_J(\varepsilon_1,\varepsilon_2,A)$ defined in \eqref{NJ:largec}, we introduce the quantity $\mathcal{A}_J(\beta_L,\varepsilon_1,\varepsilon_2,A)$ defined as follows:
\begin{equation}\label{def:AJC}
	\begin{split}
		\mathcal{A}_J(\beta_L,\varepsilon_1,\varepsilon_2,A):=\sum\limits_{h\in(A-\varepsilon_1,A+\varepsilon_2)}n_{h,h+J}e^{-(h-A)\beta_L}.
	\end{split}
\end{equation}
This definition is similar to the previously defined $\mathcal{A}_J(\beta_L,\varepsilon)$ (as defined in eq.\,\eqref{def:AJ}), but now it depends on $A$ since the theory is no longer fixed. By definition, $\mathcal{N}_J(\varepsilon_1,\varepsilon_2)$ and $\mathcal{A}_J(\beta_L,\varepsilon)$ satisfy the following inequality:
\begin{equation}\label{NJbounds:largec}
	\begin{split}
		e^{-\varepsilon_1\beta_{L}}\mathcal{A}_J(\beta_L,\varepsilon_1,\varepsilon_2,A)\leqslant\mathcal{N}_J(\varepsilon_1,\varepsilon_2,A)\leqslant e^{\varepsilon_2\beta_{L}}\mathcal{A}_J(\beta_L,\varepsilon_1,\varepsilon_2,A).
	\end{split}
\end{equation}
We consider $\mathcal{A}_J(\beta_L,\varepsilon)$ in the holographic double lightcone limit ($\mathrm{HDLC}$), which is defined by the following limit procedure:
\begin{equation}\label{def:HDLC}
	\begin{split}
		\mathrm{HDLC}_{w}\ \mathrm{limit}:\quad&\beta_L/A,\ J,\ A\to\infty,\quad 2\pi \alpha(1-w^2)\sqrt{\frac{J}{A}}-\beta_L\rightarrow\infty, \\
		&\beta_L^{-1}\log\left(\frac{J}{A}\right)\rightarrow0\,,
	\end{split}
\end{equation}
where $w\in\left(\frac{1}{2},1\right)$ is fixed. 

We note that in the $\mathrm{HDLC}_w$ limit, the first, second, and fourth conditions imply that $J/A^3\rightarrow\infty$. When $A$ and $T$ are fixed, the $\mathrm{HDLC}_w$ limit reduces to the $\mathrm{DLC}_w$ limit defined in eq\,\eqref{def:DLCwJ}. Thus, we consider the $\mathrm{HDLC}_w$ limit as a natural generalization of the $\mathrm{DLC}_w$ limit to the case of the large central charge limit.

With the above setup, we have the following result:
\begin{theorem}\label{theorem:largec}
	Let $\left\{A_c\right\}$ be a family of CFTs satisfying the above mentioned conditions (see section \ref{section:modularsetup} and the beginning of section \ref{section:largec}). Consider any $w \in \left(\frac{1}{2},1\right)$ fixed and $\varepsilon_i$ within the range
	\begin{equation}\label{epsilon:choice:DLCJlargec}
		\begin{split}
			\varepsilon_{i,\rm min}(\beta_{L},J,A)\leqslant&\varepsilon_{i}\leqslant1-\frac{1}{2w}, \\
		\end{split}
	\end{equation}
    where $\varepsilon_{i,\rm min}$ are defined by
    \begin{equation}\label{def:epsilonminC}
    	\begin{split}
    		\varepsilon_{1,\rm min}(\beta_{L},J,A)&:=\dfrac{3\log \left(A\beta_{L}\right)}{2\pi\alpha w^2}\sqrt{\dfrac{A}{J}}\,.\\
    		\varepsilon_{2,\rm min}(\beta_{L},J,A)&:=\left(\beta_{L}-\frac{4\pi}{3}\right)^{-1}\left[3\pi A+\frac{3}{2}\log\left(\frac{\beta_{L}\sqrt{AJ}}{2\pi}\right)\right]\,.
    	\end{split}
    \end{equation}
	Then the quantity $\mathcal{A}_J(\beta_L,\varepsilon_1,\varepsilon_2,A)$ defined in \eqref{def:AJC} satisfies the following asymptotic two-sided bounds in the $\mathrm{HDLC}_w$ limit \eqref{def:HDLC}:
	\begin{equation}
		\begin{split}\label{eq:resultFiniteEpsilonB:largec}
			\frac{1}{w}\left(\frac{1}{1-\frac{\tan\left(\pi w(1-\varepsilon)\right)}{\pi w(1-\varepsilon)}}\right)\lesssim \frac{\mathcal{A}_J(\beta_L,\varepsilon_1,\varepsilon_2,A)}{4\pi^{5/2}\beta_L^{-3/2}J^{-1/2}e^{4\pi\sqrt{AJ}}} \lesssim \frac{1}{w}\left(\frac{2}{1+\frac{\sin(2\pi w\varepsilon)}{2\pi w\varepsilon}}\right),
		\end{split}
	\end{equation}
	where $\varepsilon\equiv\max\{\varepsilon_1,\varepsilon_2\}$. The above bounds are uniform in $\varepsilon_1$ and $\varepsilon_2$.
\end{theorem}
Theorem \ref{theorem:largec} serves as the large-$c$ counterpart to theorem \ref{theorem:modulartauberianFJ}. It establishes the universal behavior of the spectrum near the twist accumulation point in the regime where
\begin{equation}
	\begin{split}
		J\gg c^3.
	\end{split}
\end{equation}
Due to the similarity in the overall proof structure between theorem \ref{theorem:largec} and theorem \ref{theorem:modulartauberianFJ}, we leave the detailed proof of theorem \ref{theorem:largec} to appendix \ref{section:estimatelargec}. In this section, we focus on providing key observations and remarks regarding the main technical distinctions specific to the large central charge limit case.
\begin{remark}
	1) The basic idea in deriving the above bounds is to estimate the partition function by itself but evaluated at some fixed inverse temperature, e.g.\ $\tilde{Z}(\beta_0,\beta_0)$.  However,  in $A\to\infty$ limit,  $\tilde{Z}(\beta_0,\beta_0)$ grows like $e^{2A\beta_0}$.  This leads to subtle differences between theorem.~\!\ref{theorem:largec} and theorem.~\!\ref{theorem:modulartauberianFJ}.  For example,  while the eq.~\!\eqref{def:epsilonminC} is similar to the eq.~\!\eqref{epsilon:choice:DLCJ}, there is an extra $3\pi A$ factor in the expression for $\varepsilon_{2,\rm min}(\beta_{L},J,A)$,  which comes about and is important because $A$ is very large.

	2) In the above discussions we assumed that the non-vacuum spectrum of Virasoro primaries starts from an $O(c)$ twist gap, i.e. $h,\bar{h}\geqslant \alpha A$ with some fixed $\alpha>0$. In fact, our conclusion does not change if we have finitely many Virasoro primaries with $h,\bar h$ being $O(1)$ numbers.  Using the same analysis as in appendices \ref{app:estimateIR2} and \ref{app:estimateIR34}, one can show that the contributions  $I^{\rm dual}_{\pm,(h,\bar{h})}$,  from each of these extra operators are suppressed:
\begin{equation}
	\begin{split}
		\abs{\frac{I^{\rm dual}_{\pm,(h,\bar{h})}}{I^{\rm dual}_{\pm,\rm vac}}}\stackrel{{\rm HDLC}_{w}}{\longrightarrow}0\,,
	\end{split}
\end{equation}
and hence can be neglected in our analysis.
\end{remark}
Let us now estimate $\mathcal{N}_J(\varepsilon_1,\varepsilon_2)$ using theorem \ref{theorem:largec}. For this purpose, we choose the following values of $\beta_{L}$, $\varepsilon_1$ and $\varepsilon_2$:
\begin{equation}\label{NJchoice:largec}
	\begin{split}
		\beta_{L}=\kappa\sqrt{\frac{J}{A}},\quad\varepsilon_1=\frac{1}{\pi\alpha w^2}\sqrt{\frac{A}{J}}\log\left(AJ\right),\quad\varepsilon_2=\frac{3}{\kappa}\sqrt{\frac{A}{J}}\left(2\pi A+\log J\right),
	\end{split}
\end{equation}
where $\kappa\in(0,2\pi\alpha(1-w^2))$. It can be verified that in the limit $J/A^3\rightarrow\infty$ (which is necessary for the HDLC$_w$ limit as mentioned after \eqref{def:HDLC}), both the conditions of the HDLC$_w$ limit \eqref{def:HDLC} and the $\varepsilon_i$ bounds \eqref{def:epsilonminC} are satisfied. Additionally, it is worth noting that for fixed $A$, both $\varepsilon_1$ and $\varepsilon_2$ decay as $O(J^{-1/2}\log J)$ in the limit $J\rightarrow\infty$, which is consistent with the behavior observed in the case of fixed CFT.

By applying eq.\,\eqref{NJbounds:largec}, theorem \ref{theorem:largec}, and eq.\,\eqref{NJchoice:largec}, we obtain two-sided bounds for the quantity $\mathcal{N}_J(\varepsilon_1,\varepsilon_2)$. In order to simplify the statement of the result, we sacrifice optimality by choosing $w=\frac{3}{4}$. Thus, we arrive at the following estimate for $\mathcal{N}_J$:
\begin{corollary}\label{cor:operatorcount:largec}
	For any fixed $\kappa\in\left(0,\frac{7\pi\alpha}{8}\right)$, we have
	\begin{equation}
		\begin{split}
			&\mathcal{N}_J\left(\varepsilon_1,\varepsilon_2\right)=e^{4\pi\sqrt{AJ}+f_{\kappa}(A,J)}\quad(J/A^3\rightarrow\infty)\\
			&\left(\varepsilon_1\equiv\frac{1}{\pi\alpha}\sqrt{\frac{A}{J}}\log\left(AJ\right), \varepsilon_2\equiv\frac{3}{\kappa}\sqrt{\frac{A}{J}}(2\pi A+\log J)\right), \\
		\end{split}
	\end{equation}
    where the error term $f_\kappa(A,J)$ is bounded by
    \begin{equation}
    	\begin{split}
    		\abs{f_\kappa(A,J)}\leqslant6\pi A+5\log(AJ)+C(\kappa),
    	\end{split}
    \end{equation}
    with $C(\kappa)$ being a finite constant.
\end{corollary}

\subsection{Near-extremal rotating BTZ black holes}
In this subsection, we will discuss the implications of the results we have derived in the previous subsection within the context of holography. It is worth noting that the Cardy-like formulas are commonly used to compute the entropy of black holes in AdS$_3$ \cite{Strominger:1996sh}. However, our current investigation focuses on the rotating BTZ black holes.

The near-extremal rotating BTZ black hole has an approximate AdS$_2\times S^1$ throat, as discussed in section 4.1 of \cite{Ghosh:2019rcj}. Since the Schwarzian action describes gravity in nearly AdS$_2$ spacetime \cite{Maldacena:2016upp}, it is reasonable to expect that the nearly AdS$_2\times S^1$ throat of the near-extremal rotating BTZ black hole is described by the Schwarzian action. From a holographic perspective, this suggests the existence of a \textit{Schwarzian} sector in 2D holographic CFT.

In the work \cite{Ghosh:2019rcj}, the authors made significant progress in identifying this Schwarzian sector. They specifically highlighted that the near-extremal limit in 3D gravity is dual to the double lightcone limit in 2D CFT. Their analysis led to the proposal of a universal sector described by Schwarzian theory in irrational 2D CFTs. The methodology closely follows the intuitive aspects of the lightcone bootstrap in the large central charge regime, effectively capturing the qualitative features of this limit. Furthermore, the authors extended their analysis to correlators, providing additional evidence for the proposed universality.

Here we would like to focus on the torus partition function. In section \ref{section:largec:schwarzian recap}, we will briefly review the argument in \cite{Ghosh:2019rcj}. Then in section \ref{section:largec:fineprints}, we will compare our results to the ones in \cite{Ghosh:2019rcj}, and clarify what we can justify and what we cannot.

\subsubsection{Review: rotating BTZ black hole thermodynamics in the near-extremal limit}\label{section:largec:schwarzian recap}
In this section, we provide a concise review of the key points from \cite{Ghosh:2019rcj} regarding the analysis of the torus partition function in the lightcone limit and its connection to the near-extremal limit of the rotating BTZ black hole in AdS$_3$.

The metric of the rotating BTZ black hole (without electric charge) is given by \cite{Banados:1992wn}
\begin{equation}
ds^2= -f(r)dt^2+\frac{dr^2}{f(r)}+r^2\left(d\phi-\frac{r_+r_-}{\ell_3r^2}dt\right)^2\,, \quad f(r):=\frac{(r^2-r_+^2)(r^2-r_-^2)}{\ell_3^2r^2}\,.
\end{equation}
Here, $r_\pm$ denote the radii of the outer and inner horizons, respectively, satisfying $0<r_-\leqslant r_+$, and $\ell_3$ denotes the radius of AdS$_3$ (which is related to the cosmological constant $\Lambda$ through $\Lambda=-\ell_3^{-2}$).

To facilitate our discussion, we will use dimensionless parameters for the physical quantities (such as temperature, mass, etc.) of the BTZ black hole. The corresponding dimensionful parameters can be obtained by multiplying dimensionless quantities by $(\ell_3)^{a}$ with the appropriate power indices $a$.

The mass $M$ and spin $J$ of the black hole are given by
\begin{equation}\label{BTZ:phys}
M=\frac{r_+^2+r_-^2}{8G_N\ell_3}\,,\quad J= \frac{r_+r_-}{4G_N\ell_3}\,,
\end{equation}
where $G_N$ is Newton's constant. The requirement for $r_{\pm}$ to be real implies the bound $J\leqslant M$.

The thermodynamic quantities of the BTZ black hole, including the Hawking temperature $T_{\rm H}$, the angular momentum chemical potential $\Omega$, and the black hole entropy $S$, were derived using various semiclassical methods \cite{Banados:1992wn,Hyun:1994na,Ichinose:1994rg,Carlip:1994gc,Brown:1994gs,Carlip:1995qv,Carlip:1994hq,Englert:1994hu,Medved:2001ca,Emparan:1998qp}. The expressions for these quantities are as follows:
\begin{equation}\label{BTZ:thermo}
T_{\rm H}=\frac{r_+^2-r_-^2}{2\pi\ell_3 r_+}\,,\quad\Omega=\frac{r_-}{r_+}\,,\quad S=\frac{\pi r_+}{2G_N}\,.
\end{equation}
The black hole thermodynamics can be explicitly verified: $dM=T_{\rm H} dS+\Omega dJ$. 

In the classical regime where $G_N \ll \ell_3$ and in the near-extremal regime where $r_+ \approx r_-$, the entropy of the near-extremal black hole is related to its angular momentum using \eqref{BTZ:phys}, \eqref{BTZ:thermo}, and the Brown-Henneaux relation $c = \frac{3\ell_3}{2G_N}$ \cite{Brown:1986nw}. We find that the entropy is given by
\begin{equation}\label{BTZ:entropy}
	S \approx 2\pi\sqrt{\frac{c}{6}J}\approx4\pi\sqrt{AJ}\quad(c\gg1).
\end{equation}
This result is consistent with corollary \ref{cor:operatorcount:largec} which provides the operator counting formula in the large $c$ limit. It is remarkable that from the CFT side, we obtain the correct entropy formula for near-extremal black holes, providing a gravitational interpretation of our results. This matching between the CFT and gravitational descriptions not only reinforces the validity of the thermodynamic description of black hole physics but also provides support for the holographic principle.

The inverse temperatures for left and right movers, denoted as $\beta_L$ and $\beta_R$ respectively, are related to $T_{\rm H}$ and $\Omega$ as follows:
\begin{equation}\label{BTZ:defbeta}
	\begin{split}
		\beta_{L}=(1+\Omega)\beta,\quad\beta_{R}=(1-\Omega)\beta\quad(\beta=T_{\rm H}^{-1}).
	\end{split}
\end{equation}
Expressing $\beta_L$ and $\beta_R$ in terms of $r_\pm$, we have:
\begin{equation}
	\begin{split}
		\beta_{L}=\frac{2\pi\ell_3}{r_+-r_-},\quad\beta_{R}=\frac{2\pi\ell_3}{r_++r_-}.
	\end{split}
\end{equation}
In \cite{Preskill:1991tb}, it was emphasized that the self-consistency of semiclassical methods requires the back reaction of black hole radiation to be negligible. Specifically, the fluctuation in Hawking temperature, denoted as $\Delta T_{\rm H}$, should be much smaller than $T_{\rm H}$ itself:
\begin{equation}
	\begin{split}
		\frac{\braket{(\Delta T_{\rm H})^2}}{T_{H}^2}\ll1.
	\end{split}
\end{equation} 
This condition can be expressed using a standard thermodynamic argument \cite{landau2013statistical} as:
\begin{equation}\label{BTZ:condition}
	\begin{split}
		\frac{\braket{(\Delta T_{\rm H})^2}}{T_{H}^2}=\frac{1}{T_{\rm H}}\left(\frac{\partial T_{\rm H}}{\partial S}\right)_J\ll1.
	\end{split}
\end{equation}
By substituting \eqref{BTZ:thermo} into \eqref{BTZ:condition}, we find:
\begin{equation}\label{BTZ:consistency}
	\begin{split}
		T_{\rm H}\gg\frac{G_N}{\ell_3}.
	\end{split}
\end{equation}
Using \eqref{BTZ:defbeta} and the Brown-Henneaux relation $c = \frac{3\ell_3}{2G_N}$ in the semiclassical regime $G_N \ll \ell_3$, the above constraint can be written as:
\begin{equation}\label{BTZ:gaptemp}
	\begin{split}
		\beta_{L},\beta_{R}\ll c\quad(c\gg1).
	\end{split}
\end{equation}
For this reason, in ref.\,\cite{Ghosh:2019rcj}, $c^{-1}$ is referred to as the ``gap temperature" of the BTZ black hole. It was believed that the thermodynamic description of the black hole breaks down when $T_{\rm H} = O(c^{-1})$.

The above analysis raises a puzzle regarding the validity of black hole thermodynamics in the near-extremal regime ($r_+\approx r_-$ or $\Omega\approx1$), where $T_{\rm H}$ could be of $O(c^{-1})$ or even smaller. A resolution to this puzzle is recently proposed in \cite{Ghosh:2019rcj}.  See also \cite{Turiaci:2023wrh} for a recent review.

Ref.\,\cite{Ghosh:2019rcj} investigated the near-extremal regime of black holes characterized by the conditions:
\begin{equation}\label{GMT:condition1}
	\begin{split}
		\beta_{L}=O(c),\quad\beta_{R}=O(c^{-1})\quad(c\gg1).
	\end{split}
\end{equation}
Based on the previous argument, it appears that black hole thermodynamics breaks down in this regime due to the violation of the ``gap temperature condition" \eqref{BTZ:gaptemp}. However, \cite{Ghosh:2019rcj} proposed that the thermodynamics remains valid, but it is no longer described by semiclassical methods. Instead, a quantum mode governed by Schwarzian theory \cite{Maldacena:2016upp,Mertens:2017mtv} becomes dominant in the near-extremal regime.

Furthermore, \cite{Ghosh:2019rcj} argued for the appearance of the Schwarzian sector universally in a broad class of $c\gg1$ irrational CFTs with a twist gap and a significantly large central charge $c\gg1$. Here we aim to revisit their argument while presenting it from a slightly different perspective.  Instead of focusing on the full partition function $Z(\beta_{L},\beta_{R})$ as examined in \cite{Ghosh:2019rcj}, our analysis focuses on the reduced partition function $\tilde{Z}(\beta_{L},\beta_{R})$, which exclusively accounts for the Virasoro primaries. The computation below follows the same logic as in \cite{Ghosh:2019rcj}.\footnote{For a review on the same argument in terms of the full partition function $Z(\beta_{L},\beta_{R})$, see appendix B of \cite{Aggarwal:2022xfd}.} This choice enables us to conveniently compare our results with those of \cite{Ghosh:2019rcj}.

The intuitive argument goes as follows. In the grand canonical ensemble, focusing on $\tilde{Z}(\beta_{L},\beta_{R})$, the limit as $\beta_R\to 0$ favors the dominance of the vacuum state in the dual channel, while the limit as $\beta_L\to \infty$ favors non-vacuum states. The presence of a twist gap in the spectrum of Virasoro primaries ensures that each non-vacuum term is suppressed, leading to the vacuum state's dominance in the dual channel. The vacuum contribution in the dual channel can then be identified with the contribution from the rotating BTZ black hole. Hence, we have the approximation:
\begin{equation}\label{GMT:BTZdominance}
	\begin{split}
		\tilde{Z}(\beta_{L},\beta_{R})\stackrel{\eqref{GMT:condition1}}{\approx}\sqrt{\frac{4\pi^2}{\beta_{L}\beta_{R}}}e^{\frac{4\pi^2 A}{\beta_{L}}+\frac{4\pi^2A}{\beta_{R}}}\left(1-e^{-\frac{4\pi^2}{\beta_{L}}}\right)\left(1-e^{-\frac{4\pi^2}{\beta_{R}}}\right)\equiv \tilde{Z}_{\rm BTZ}(\beta_{L},\beta_{R}),
	\end{split}
\end{equation}
where $A\equiv\frac{c-1}{24}$. In the large central charge limit, the left-moving part of $\tilde{Z}_{\rm BTZ}$ can be identified to the Schwarzian partition function. To see this, we introduce the Schwarzian variable $\tilde{\beta}$ by rescaling $\beta(\equiv(\beta_{L}+\beta_{R})/2)$
\begin{equation}\label{def:Schwbeta}
	\begin{split}
		\tilde{\beta}:=\frac{\beta}{2A}\left(\equiv\frac{12\beta}{(c-1)}\right).
	\end{split}
\end{equation}
In the regime \eqref{GMT:condition1} with large central charge, we have $\beta\approx\beta_{L}/2$, then the grand canonical partition function can be further approximated as
\begin{equation}
	\begin{split}
		\tilde{Z}(\beta_{L},\beta_{R})\stackrel{\eqref{GMT:condition1}}{\approx}\left(\frac{\pi}{\tilde{\beta}}\right)^{3/2}e^{\pi^2/\tilde{\beta}}\times\left(\frac{\pi^2}{A}\right)^{3/2}\beta_{R}^{-1/2}e^{\frac{4\pi^2A}{\beta_{R}}},
	\end{split}
\end{equation}
and
\begin{equation}\label{def:Schwpartition}
	\begin{split}
		Z_{\text{Schw}}(\tilde{\beta})=\left(\frac{\pi}{\tilde{\beta}}\right)^{3/2}e^{\pi^2/\tilde{\beta}}
	\end{split}
\end{equation}
is the Schwarzian partition function with a circle length $\tilde{\beta}$. So we see that the grand canonical partition function is dominated by the Schwarzian modes in the regime \eqref{GMT:condition1}. The grand canonical entropy can be determined using the standard thermodynamic formula:
\begin{equation}\label{GMT:Sgrand}
	\begin{split}
		S_{\rm grand}(\beta_{L},\beta_{R})\equiv&\left(1-\beta_{L}\frac{\partial}{\partial\beta_{L}}-\beta_{R}\frac{\partial}{\partial\beta_{R}}\right)\log\tilde{Z}(\beta_{L},\beta_{R}) \\
		=&\frac{8\pi^2 A}{\beta_{R}}+O(\log\beta_{L})+O(\log\beta_{R})+O(\log A)
	\end{split}
\end{equation}
Here, the entropy from $Z_{\rm Schw}$ is included in the error term. 

The Schwarzian sector can also be seen in the canonical ensemble of primary states with $\bar{h}=h+J$ (spin-$J$):
\begin{equation}\label{GMT:defZJ}
	\begin{split}
		\tilde{Z}_J(\beta)\equiv\int_{0}^{2\pi}\frac{d\theta}{2\pi}\,e^{i\theta J}\,\tilde{Z}(\beta-i\theta,\beta+i\theta).
	\end{split}
\end{equation}
In \cite{Ghosh:2019rcj}, it was argued that in an equivalent near-extremal regime:\footnote{We modify the condition stated in ref.\,\cite{Ghosh:2019rcj}, eq.\,(2.14), to $J=O(c^3)$. We believe this to be the correct equivalent condition to $\beta_{R}=O(c^{-1})$. The reason will become clear shortly: we will see that $\beta_{R}=2\pi\sqrt{\frac{A}{J}}$ under the saddle point approximation. This identification, along with the condition $\beta_{R}=O(c^{-1})$, implies $J=O(c^3)$. This has also been noted in the appendix B of \cite{Aggarwal:2022xfd} to justify the validity of saddle.}
\begin{equation}\label{GMT:condition2}
	\begin{split}
		\beta=O(c),\quad J=O(c^3)\quad(c\gg1),
	\end{split}
\end{equation}
$\tilde{Z}_J(\beta)$ can be evaluated by replacing $\tilde{Z}$ with $\tilde{Z}_{\rm BTZ}$ and computing the contribution around the complex saddle point $\theta= i\beta-2\pi i\sqrt{\frac{A}{J}}+O(J^{-1})$. This yields:
\begin{equation}\label{GMT:Zcanonicalresult}
	\begin{split}
		\tilde{Z}_J(\beta)\stackrel{\eqref{GMT:condition2}}{\approx}\sqrt{2}\pi^{5/2}\beta^{-3/2}J^{-1/2}e^{4\pi\sqrt{AJ}-\beta J+\frac{2\pi^2A}{\beta}}.
	\end{split}
\end{equation}
Then the canonical partition function $\tilde{Z}_J(\beta)$ can be expressed as
\begin{equation}
	\begin{split}
		\tilde{Z}_J(\beta)=Z_{\rm Schw}(\tilde{\beta})\times e^{4\pi\sqrt{AJ}-\beta J+O(\log A)+O(\log\beta)},
	\end{split}
\end{equation}
The canonical entropy is obtained using the standard thermodynamic formula:
\begin{equation}\label{GMT:Scanonical}
	\begin{split}
		S_{\rm canon}(J,\beta)\equiv\left(1-\beta\frac{\partial}{\partial\beta}\right)\tilde{Z}_J(\beta)=4\pi\sqrt{AJ}+O(\log A)+O(\log\beta).
	\end{split}
\end{equation}
Here, the entropy from $Z_{\rm Schw}$ is included in the $O(\log A)+O(\log\beta)$ terms. To match the grand canonical entropy \eqref{GMT:Sgrand} with the canonical entropy \eqref{GMT:Scanonical}, we recall $\beta_{R}=\beta+i\theta=2\pi\sqrt{\frac{A}{J}}+O(J^{-1})$ in the saddle point approximation. This identification ensures the agreement of their leading terms: $S=4\pi\sqrt{AJ}+(errors)$.

Let us consider the microcanonical ensemble now. In the limit \eqref{GMT:condition1}, where $\beta_{L}\approx2\beta=O(c)$, we expect the dominant contribution to the grand canonical partition function $\tilde{Z}(\beta_{L},\beta_{R})$ to arise from states with $h-A=O(c^{-1})$. This expectation is based on the duality between $\beta_{L}$ and $h-A$ in the definition of $\tilde{Z}(\beta_L,\beta_R)$ (as seen in \eqref{def:Ztilde}). Additionally, in the canonical ensemble, we identify $\beta_{R}$ with the saddle point at large spin, which leads to $\beta_{R}=2\pi\sqrt{\frac{A}{J}}$. Using this relation together with \eqref{GMT:condition1}, we find that $J\equiv\bar{h}-h=O(c^3)$. In the microcanonical ensemble, we therefore anticipate that the relevant spectrum for the limit \eqref{GMT:condition1} consists of states characterized by
\begin{equation}\label{GMT:condition3}
	\begin{split}
		J(\equiv\bar{h}-h)=O(c^3),\quad h-A=O(c^{-1}).
	\end{split}
\end{equation}
To determine the microcanonical entropy, we need to determine the number of Virasoro primaries within the range \eqref{GMT:condition3}. Based on the BTZ dominance \eqref{GMT:BTZdominance} in the limit \eqref{GMT:condition1}, a plausible approach is to perform an inverse Laplace transform of $\tilde{Z}_{\rm BTZ}$ to obtain the coarse-grained spectral density of Virasoro primaries within the specified range. The inverse Laplace transform of $\tilde{Z}_{\rm BTZ}$ yields the following expressions for the left- and right-moving vacuum characters:
\begin{equation}
	\begin{split}
		\frac{1}{\sqrt{\beta}}e^{\frac{4\pi^2A}{\beta}}(1-e^{-\frac{4\pi^2}{\beta}})=\int_{A}^{\infty}d \bar h\,\rho_0(A;\bar h)\,e^{-(\bar h-A)\beta}, \\
	\end{split}
\end{equation}
where the modular crossing kernel $\rho_0$ is given by
\begin{equation}
	\begin{split}
		\rho_0(A;\bar h)=\dfrac{\cosh(4\pi\sqrt{A(\bar h-A)})}{\sqrt{(\bar h-A)\pi}}-\dfrac{\cosh(4\pi\sqrt{(A-1)(\bar h-A)})}{\sqrt{(\bar h-A)\pi}}. \\
	\end{split}
\end{equation}
Therefore, $2\pi\rho_{0}(A;h)\rho_{0}(A;\bar{h})$ serves as the naive coarse-grained spectral density for $\tilde{Z}$. Let us use it to estimate the total number of Virasoro primaries within the range:
\begin{equation}\label{GMT:microrange}
	\begin{split}
		\abs{h-A}\leqslant\lambda c^{-1},\quad \abs{h-A-J}\leqslant\frac{1}{2},\quad\lambda=O(1).
	\end{split}
\end{equation}
By integrating $2\pi\rho_{0}(A;h)\rho_{0}(A;\bar{h})$ over the specified range of $h$ and $\bar{h}$, we obtain
\begin{equation}
	\begin{split}
		&2\pi\int_A^{A+\lambda c^{-1}} dh\,\rho_{0}(A,h)\,\times\,\int_{A+J-\frac{1}{2}}^{A+J+\frac{1}{2}}d\bar{h}\,\rho_{0}(A;\bar{h}) \\
		\approx&\frac{\pi^{5/2}}{A^{3/2}}\int_0^{\lambda/6}d(k^2)\sinh(2\pi k)\,\times\,\frac{1}{2\sqrt{\pi J}}\,e^{4\pi\sqrt{AJ}}\quad(k=2\sqrt{A(h-A)}).
	\end{split}
\end{equation}
In this expression, the integral over the left-moving part is identified as the integral over the Schwarzian density of states (which is $\sinh(2\pi k)$). By taking the logarithm of the result, we obtain the microcanonical entropy of the states within the range \eqref{GMT:microrange}:
\begin{equation}\label{GMT:Smicro}
	\begin{split}
		S_{\rm micro}(A,\lambda,J)=4\pi\sqrt{AJ}+O(\log A)+O(\log J),
	\end{split}
\end{equation}
where the contribution from the left-moving part, which includes the Schwarzian sector and is of $O(\log A)$, is absorbed into the error term.

The agreement between the entropy computations in different ensembles for near-extremal BTZ black holes, as seen from eqs.\,\eqref{GMT:Sgrand}, \eqref{GMT:Scanonical}, and \eqref{GMT:Smicro}, is significant. Furthermore, the entropy formula $S=4\pi\sqrt{AJ}$ can be reproduced using eqs.\,\eqref{BTZ:phys}, \eqref{BTZ:thermo}, and $c=\frac{3\ell_3}{2G_N}$ in the near-extremal regime. This provides strong evidence supporting the validity of the thermodynamic description of AdS$_3$ pure gravity in the near-extremal regime, where the Hawking temperature is of the same order as the ``gap temperature".

However, it is important to note that the above arguments have certain caveats, and it is necessary to consider these limitations when concluding the existence of a universal Schwarzian sector in a general class of irrational CFTs, even without a gravitational dual. We will discuss these caveats in the next subsection.
	
\subsubsection{Fine-prints and resolution a la Tauberian}\label{section:largec:fineprints}
In this subsection, we would like to compare our results from section \ref{section:holoCFTresults} to the ones in \cite{Ghosh:2019rcj}, and put some of the intuitive arguments above on rigorous footing and clarify what can be proven rigorously. 

Let us first consider the partition function in the grand canonical ensemble: $\tilde{Z}(\beta_{L},\beta_{R})$. By assuming \eqref{assum1} and \eqref{assum2}, we can establish the BTZ dominance, i.e.\,the vacuum dominance in the dual channel, within the regime specified by \eqref{GMT:condition1}. Similar to the estimate performed in \cite{Pal:2022vqc}, we find that
\begin{equation}
	\begin{split}
		\frac{\tilde{Z}_{\rm BTZ}(\beta_{L},\beta_{R})}{\tilde{Z}(\beta_{L},\beta_{R})}=1+O\left[\beta_{L}^{3/2}e^{-A\left(\frac{4\pi^2\alpha}{\beta_{R}}-\beta_{L}-3\pi\right)}\right],
	\end{split}
\end{equation}
where $\tilde{Z}_{\rm BTZ}$ was defined in \eqref{GMT:BTZdominance}. From this result, the BTZ dominance can be reached under condition \eqref{GMT:condition1} and an extra technical assumption
\begin{equation}
	\begin{split}
		\frac{4\pi^2\alpha}{\beta_{R}}-\beta_{L}\geqslant\kappa c,
	\end{split}
\end{equation}
where $\kappa>0$ is a fixed constant. The above extra assumption is weaker than the one imposed in the HDLC$_w$ condition \eqref{def:HDLCequiv} (equivalent to \eqref{def:HDLC} when we identify $\beta_R=2\pi\sqrt{A/J}$). In this limit, $\beta_{L}$ is allowed to be of $O(c)$, which gives an $O(1)$ Schwarzian variable $\tilde{\beta}\approx\frac{\beta_{L}}{4A}$ using \eqref{def:Schwbeta}. Consequently, the grand canonical partition function $\tilde{Z}(\beta_{L},\beta_{R})$ exhibits the following asymptotic behavior:
\begin{equation}
	\begin{split}
		\tilde{Z}(\beta_{L},\beta_{R})\approx\tilde{Z}_{\rm BTZ}(\beta_{L},\beta_{R})\approx Z_{\rm Schw}(\tilde{\beta})\times\left(\frac{\pi^2}{A}\right)^{3/2}\beta_{R}^{-1/2}e^{\frac{4\pi^2A}{\beta_{R}}}.
	\end{split}
\end{equation}
Therefore, in our framework, we have justified that the Schwarzian partition function appears in the left-moving sector of the grand canonical partition function. However, it should be noted that the presence of the Schwarzian sector in the theory is not guaranteed, as the direct channel spectrum capable of reproducing the above asymptotic behavior is not necessarily unique.

Let us next consider the canonical ensemble. We would like to start by establishing a relationship between the quantity $\mathcal{A}_J(\beta_{L},\varepsilon_1,\varepsilon_2,A)$ (defined in \eqref{def:AJC}) which we used in theorem \ref{theorem:largec} and the canonical partition function $\tilde{Z}_{J}(\beta)$ (defined in \eqref{GMT:defZJ}). By explicitly evaluating the integral in \eqref{GMT:defZJ}, we find
\begin{equation}
	\begin{split}
		\tilde{Z}_J(\beta)=e^{-\beta J}\sum\limits_{h}n_{h,h+J}e^{-2\beta(h-A)}.
	\end{split}
\end{equation}
Comparing this expression with \eqref{def:AJC}, we obtain the exact relation
\begin{equation}\label{relation:ZtildeA}
	\begin{split}
		\tilde{Z}_J(\beta)=\mathcal{A}_J(2\beta,\infty,\infty,A)e^{-\beta J}.
	\end{split}
\end{equation}
Then it is not surprising that setting $\beta\approx\beta_{L}/2\gg A$ in \eqref{GMT:Zcanonicalresult} yields
\begin{equation}
	\begin{split}
		\tilde{Z}_J(\beta)\approx 4\pi^{5/2}\beta_{L}^{-3/2}J^{-1/2}e^{4\pi\sqrt{AJ}}\times e^{-\beta J},
	\end{split}
\end{equation}
which, by \eqref{relation:ZtildeA}, implies
\begin{equation}\label{GMT:AJresult}
	\begin{split}
		\mathcal{A}_J(\beta_{L},\infty,\infty,A)\approx4\pi^{5/2}\beta_{L}^{-3/2}J^{-1/2}e^{4\pi\sqrt{AJ}}.
	\end{split}
\end{equation}
This result agrees well with theorem \ref{theorem:largec} (see \eqref{eq:resultFiniteEpsilonB:largec}, the denominator below $\mathcal{A}_J$). The main distinction is that \eqref{GMT:AJresult} takes into account contributions from all spin-$J$ states (i.e.\,with $\varepsilon_1=\varepsilon_2=\infty$), while theorem \ref{theorem:largec} only counts the contributions from spin-$J$ states with twists near $\frac{c-1}{12}$ (i.e.\,with finite $\varepsilon_1$ and $\varepsilon_2$). This difference is actually one of the main points of theorem \ref{theorem:largec}: in the HDLC$_w$ limit, both $\mathcal{A}_J(\beta_{L},\varepsilon_1,\varepsilon_2,A)$ and $\mathcal{A}_J(\beta_{L},\infty,\infty,A)$ yield the same leading behavior.

Now, let us clarify the conditions required for our analysis in the canonical ensemble. For theorem \ref{theorem:largec} to hold true, we need the HDLC$_w$ conditions \eqref{def:HDLC}, which imply:
\begin{equation}\label{holoCFT:condition1}
	\begin{split}
		\beta_{L}\gg c,\quad J\gg c^{3}.
	\end{split}
\end{equation}
Thus, theorem \ref{theorem:largec} pertains to a different regime than the one considered in \cite{Ghosh:2019rcj}, where they require the condition \eqref{GMT:condition2}. In \cite{Ghosh:2019rcj}, having $\beta_{L}=O(c)$ is crucial as it enables the identification of the Schwarzian sector through the rescaling \eqref{def:Schwbeta}. The condition \eqref{holoCFT:condition1} indicates that our results are valid when the Hawking temperature is much lower than the ``gap temperature":
\begin{equation}
	\begin{split}
		T_{\rm H}\ll c^{-1}.
	\end{split}
\end{equation}
In addition, according to theorem \ref{theorem:largec}, the HDLC$_w$ limit conditions \eqref{def:HDLC} impose a lower bound on the temperature. Specifically, we require
\begin{equation}
	\begin{split}
		T_{\rm H}\geqslant{\rm const}\times\alpha^{-1}\sqrt{\frac{c}{J}}
	\end{split}
\end{equation}
in order to ensure the validity of our analysis. 

Figure \ref{fig:threeregimes} provides a schematic illustration of the corresponding regimes for our arguments, as well as for the arguments presented in \cite{Ghosh:2019rcj} and \cite{Preskill:1991tb}.
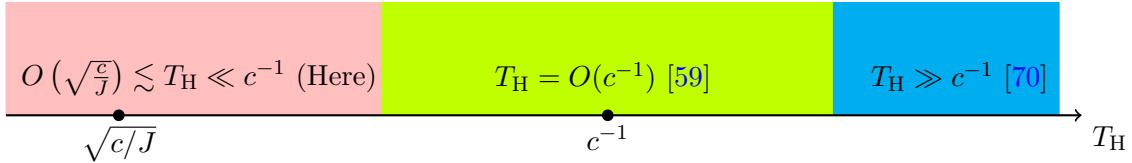
\begin{figure}
	\centering
	\begin{tikzpicture}
		\draw [fill=pink,pink] (-3,0) rectangle (2,1.5);  
		\draw [fill=lime,lime] (2,0) rectangle (8,1.5);  
		\draw [fill=cyan,cyan] (8,0) rectangle (11,1.5);  
		\draw[thick,->] (-3,0) -- (11.3,0) node[anchor=north west] {$T_{\rm H}$};
		\draw (5,0) node[anchor=north] {$c^{-1}$};
		\draw (5,0) circle[radius=2pt];
		\fill (5,0)  circle[radius=2pt];
		\draw (-1.5,0) node[anchor=north] {$\sqrt{c/J}$};
		\draw (-1.5,0) circle[radius=2pt];
		\fill (-1.5,0)  circle[radius=2pt];
		\node[text width=5cm] at (-0.3,0.5) {$O\left(\sqrt{\frac{c}{J}}\right)\lesssim T_{\rm H}\ll c^{-1}$ (Here)};
		\node[text width=3cm] at (5,0.5) {$T_{\rm H}=O(c^{-1})$ \cite{Ghosh:2019rcj}};
		\node[text width=3cm] at (10,0.5) {$T_{\rm H}\gg c^{-1}$ \cite{Preskill:1991tb}};
	\end{tikzpicture}
    \caption{\label{fig:threeregimes}The three regimes of $T_{\rm H}$ compared to the ``gap temperature" $c^{-1}$ (assuming $c\gg1$ and $\beta_{R}=O(c^{-1})$ or smaller). The results in our paper are valid in the pink regime. The arguments of \cite{Ghosh:2019rcj} (Schwarzian sector) correspond to the green regime. The arguments of \cite{Preskill:1991tb} (general regime of validity of black hole thermodynamics) correspond to the blue regime.}
\end{figure}
We would like to highlight that the arguments concerning the universal Schwarzian sector presented in the previous subsection \cite{Ghosh:2019rcj} are meaningful only if it can be shown that the dominant contribution to $\tilde{Z}_J(\beta)$ arises from the BTZ partition function, subject to the condition \eqref{GMT:condition2}. However, to our knowledge, this has not been rigorously established yet. A clean treatment of this issue would involve proving the following equation:
\begin{equation}\label{GMT:toestablish}
	\begin{split}
		\lim\limits_{c\rightarrow\infty}\frac{\tilde{Z}_{J,\rm BTZ}(\beta)}{\tilde{Z}_{J}(\beta)}\rightarrow1
	\end{split}
\end{equation}
under the condition \eqref{GMT:condition2}, where $\tilde{Z}_{J,\rm BTZ}(\beta)$ represents the contribution from the BTZ partition function. Eq.\,\eqref{GMT:toestablish} is similar to \eqref{modular:sketch1}, where the compact support of $\hat{\phi}_\pm$ played a crucial role in the proof. In the case of \eqref{GMT:toestablish}, the analog of $\hat{\phi}_{\pm}$ is $e^{i\theta J}$, which is supported over the entire real axis of $\theta$. Due to this technical complication, we are unable to rigorously justify \eqref{GMT:toestablish}. It is possible that \eqref{GMT:toestablish} is not universally true for all classes of irrational CFTs, but may hold with certain additional assumptions that arise from the gravitational perspective. We leave this question for future study.

Lastly, let us consider the microcanonical ensemble. In this paper, we have shown that in the HDLC$_w$ limit, the dominant contribution to $\mathcal{A}_J(\beta_{L},\varepsilon_1=\infty,\varepsilon_2=\infty,A)$ arises from the spectrum satisfying the conditions
\begin{equation}\label{holoCFT:condition2}
	\begin{split}
		J\gg c^3,\quad\Delta-J-\frac{c-1}{12}\in(-\varepsilon_1,\varepsilon_2),
	\end{split}
\end{equation}
where
\begin{equation}\label{holoCFT:condition3}
	\begin{split}
		\varepsilon_1=O\left(\sqrt{\frac{c}{J}}\log(J)\right),\quad\varepsilon_2=O\left(\sqrt{\frac{c^3}{J}}\right)+O\left(\sqrt{\frac{c}{J}}\log(J)\right).
	\end{split}
\end{equation}
We have established that within this range, the leading term of the microcanonical entropy is given by $S_{\rm micro}=4\pi\sqrt{AJ}$, which coincides with \eqref{GMT:Smicro} as well as the standard black hole thermodynamic prediction \eqref{BTZ:entropy}. It is important to note that our argument is purely based on CFT considerations and applies to a general class of irrational CFTs, without necessarily relying on a gravitational interpretation.

While the microcanonical entropy formulas are the same, the regime of validity of our result differs from that of \cite{Ghosh:2019rcj}, as mentioned earlier. In our case, the first condition \eqref{holoCFT:condition2}, which also appeared in the canonical ensemble (see \eqref{holoCFT:condition1}), is more restrictive than the corresponding condition in \eqref{GMT:condition2}.

In \cite{Ghosh:2019rcj}, it is crucial for the spectrum with $0\leqslant\Delta-J-\frac{c-1}{12}\leqslant O(c^{-1})$ to dominate $\tilde{Z}(\beta_{L},\beta_{R})$ in the limit \eqref{GMT:condition1} (i.e.\,the second condition of \eqref{GMT:condition3}) in order to identify $\Delta-J-\frac{c-1}{12}$ with the positive Schwarzian energy of $O(1)$:
\begin{equation}
	\begin{split}
		k^2\propto c\left(\Delta-J-\frac{c-1}{12}\right)=O(1).
	\end{split}
\end{equation}
In contrast, in the second condition of our case \eqref{holoCFT:condition2}, we do not exclude the spectrum with twists lower than $\frac{c-1}{12}$ (i.e.\,$\Delta-J-\frac{c-1}{12}\in(-\varepsilon_1,0)$). Additionally, the width of the window depends on $J$ and is not necessarily of $O(c^{-1})$.

In our case, it is unclear which modes within the range $\Delta-J-\frac{c-1}{12}\in(-\varepsilon_1,\varepsilon_2)$ are more important for the partition function $\tilde{Z}(\beta_{L},\beta_{R})$ in the HDLC$_w$ limit. It would be interesting to investigate the general conditions under which we can access the ``Schwarzian regime" \eqref{GMT:condition3} and rigorously perform the microstate counting. We leave this for future study.

This finishes our discussion on holographic CFTs. We end this section with the remark that it is conceivable to generalize the rigorous discussion to the supersymmetric case along the lines of \cite{Heydeman:2020hhw,Iliesiu:2022kny,Iliesiu:2022onk}.

\section{Conclusion and brief discussion}\label{section:conclusion}
In this paper, we present a refined twist accumulation result for two-dimensional unitary conformal field theories with central charge $c>1$ and a twist gap in the spectrum of Virasoro primaries. Using the lightcone bootstrap argument and Tauberian theory, we rigorously estimate the number of Virasoro primary operators with twist near $\frac{c-1}{12}$ and large spin, leading to the derivation of a Cardy-like formula \eqref{def:Micentropy} which counts the states around the twist accumulation point with twist spacing going to $0$ in the large spin limit.

We also explore potential generalizations of our result for CFTs with conserved currents. Depending on the growth of the number of currents $D(j)$ with respect to the spin $j$, the generalization to \eqref{def:Micentropy} can vary significantly. We propose conjectures in this regard and anticipate that converting these conjectures into theorems should be a feasible endeavor.

For future research, it would be interesting to consider irrational CFTs symmetric under a larger chiral algebra and with a twist gap in the spectrum of primaries of the larger chiral algebra. The twist accumulation point in the spectrum of such primaries is expected to shift, and it is conceivable to establish analogous rigorous findings for irrational CFTs with the larger chiral algebra.

Additionally, we study a family of CFTs with a twist gap growing linearly in the central charge and a uniform boundedness condition on the torus partition function. We establish a similar Cardy-like formula for microcanonical entropy in the limit of large central charge. From a holographic perspective, our result can be interpreted as the entropy formula for near-extremal rotating BTZ black holes in the regime where the Hawking temperature is much lower than the ``gap temperature". It would be interesting to investigate a general CFT condition, inspired by the gravity side, under which we can rigorously perform the microstate counting in the ``Schwarzian regime". In that regime, the Hawking temperature is comparable to the ``gap temperature".

Another avenue to consider is the investigation of CFTs with a global symmetry $\mathrm{G}$ and the study of the symmetry-resolved version of asymptotic CFT data. For instance, it is possible to derive the density of states restricted to an irreducible representation of $\mathrm{G}$ in the in large $\Delta$ limit  \cite{Pal:2020wwd} (see also \cite{Lin:2021udi,Lin:2023uvm}). This analysis has been extended to higher-dimensional CFTs and holographic CFTs in subsequent works such as \cite{Harlow:2021trr} followed by \cite{Kang:2022orq}. Universal results for CFTs with non-invertible symmetry are discussed in  \cite{Lin:2022dhv} and in \cite{Choi:2023xjw}.  By building upon the techniques elucidated in this paper, one can aspire to derive the universality of the CFT spectrum when restricted to an irreducible representation of $\mathrm{G}$ in the regime of fixed twist and large spin.  Furthermore, Tauberian theory can be potentially useful for extracting detailed structure of asymptotic CFT data,  unveiled in a beautiful recent paper \cite{Benjamin:2023qsc}. 

One might hope to use the techniques in this paper in the context of generic irrational CFTs with Virasoro symmetry only to shed light on 1) the claimed twist gap \cite{Benjamin:2019stq} of $\frac{c-1}{16}$ and/or 2) shifting of BTZ threshold $\frac{c-1}{24}$ by a spin dependent quantity \cite{Maxfield:2019hdt}.

Finally, it is crucial to highlight the significance of the identity block approximation in elucidating the origin of universal results in CFTs under appropriate limits. The Tauberian formalism serves as a valuable tool for rigorously understanding such approximations and their regime of validity.  Notably, recent investigations have shed light on subtle effects related to identity block dominance in \cite{Chandra:2023dgq} and \cite{Chandra:2023rhx}.  The Identity block approximation appears in the context of Virasoro mean field theory \cite{Kusuki:2018wpa,Collier:2018exn} as well. We believe that techniques developed in this paper will be useful to investigate  such effects and more.

\section*{Acknowledgments}
We thank Nathan Benjamin, Gabriele Di Ubaldo, Tom Hartman, Yikun Jiang, João Penedones, Eric Perlmutter,  Biswajit Sahoo and Joaquin Turiaci for useful and enlightening discussions.  We also thank Alexandre Belin, João Penedones, Slava Rychkov and Joaquin Turiaci for useful comments on the draft.  SP acknowledges the support by the U.S. Department of Energy, Office of Science, Office of High Energy Physics, under Award Number DE-SC0011632 and by the Walter Burke Institute for Theoretical Physics. JQ is supported by the Swiss National Science Foundation through the National Centre of Competence in Research SwissMAP and by the Simons Collaboration on Confinement and QCD Strings.

\appendix
\section{Estimating the dual-vacuum term}\label{app:Ivac:asym}	
In this section we compute the asymptotic behavior of the integral
\begin{equation}
	\begin{split}
		I^{\rm dual}_{\pm,\rm vac}\equiv&\int_{-\infty}^{+\infty} d t \sqrt{\frac{4\pi^2}{\beta_L(\beta_R+i t)}}\tilde{Z}_{\rm vac}\left(\frac{4\pi^2}{\beta_L},\frac{4\pi^2}{\beta_R+i t}\right)\hat{\phi}_\pm(t)e^{i(\bar{H}-A)t} \\
	\end{split}
\end{equation}
in the modular double lightcone limit, where
\begin{equation}
	\begin{split}
		\tilde{Z}_{\rm vac}\left(\beta,\bar{\beta}\right)&\equiv e^{A\left(\beta+\bar{\beta}\right)}\left(1-e^{-\beta}\right)\left(1-e^{-\bar{\beta}}\right). \\
	\end{split}
\end{equation}
By definition, $I^{\rm dual}_{\pm,\rm vac}$ is factorized into two parts
\begin{equation}\label{Ivac:factorize}
	\begin{split}
		I^{\rm dual}_{\pm,\rm vac}(A,\bar{H};\beta_L,\beta_R)=&I^{\rm dual,L}_{\pm,\rm vac}(A;\beta_L)I^{\rm dual,R}_{\pm,\rm vac}(A,\bar{H};\beta_R), \\
		I^{\rm dual,L}_{\pm,\rm vac}(A;\beta_L)=&\sqrt{\frac{4\pi^2}{\beta_L}}e^{\frac{4\pi^2A}{\beta_L}}\left(1-e^{-\frac{4\pi^2}{\beta_L}}\right), \\
		I^{\rm dual,R}_{\pm,\rm vac}(A,\bar{H};\beta_R)=&\int_{-\infty}^{+\infty} d t \frac{1}{\sqrt{\beta_R+i t}}e^{\frac{4\pi^2A}{\beta_R+i t}}\left(1-e^{-\frac{4\pi^2}{\beta_R+i t}}\right)\hat{\phi}_\pm(t)e^{i(\bar{H}-A)t}. \\
	\end{split}
\end{equation}
When $\beta_L\rightarrow\infty$, the asymptotic behavior of $I^{\rm dual,L}_{\pm,\rm vac}$ is given by
\begin{equation}
	\begin{split}
		I^{\rm dual,L}_{\pm,\rm vac}(A;\beta_L)\sim \left(\frac{4\pi^2}{\beta_L}\right)^{3/2}e^{\frac{4\pi^2A}{\beta_L}}.
	\end{split}
\end{equation}
For $I^{\rm dual,R}_{\pm,\rm vac}$ we introduce the identity
\begin{equation}\label{id:Ivac}
	\begin{split}
		\frac{1}{\sqrt{\beta}}e^{\frac{4\pi^2A}{\beta}}(1-e^{-\frac{4\pi^2}{\beta}})=\int_{A}^{\infty}d \bar h\,\rho_0(A;\bar h)\,e^{-(\bar h-A)\beta}, \\
	\end{split}
\end{equation}
where the kernel $\rho_0$ is given by
\begin{equation}\label{modularkernel}
	\begin{split}
		\rho_0(A;\bar h)=\dfrac{\cosh(4\pi\sqrt{A(\bar h-A)})}{\sqrt{(\bar h-A)\pi}}-\dfrac{\cosh(4\pi\sqrt{(A-1)(\bar h-A)})}{\sqrt{(\bar h-A)\pi}}. \\
	\end{split}
\end{equation}
Then by eqs.\,(\ref{Ivac:factorize}) and (\ref{id:Ivac}) we get
\begin{equation}\label{Idualvac:int}
	\begin{split}
		I^{\rm dual,R}_{\pm,\rm vac}(A,\bar{H};\beta_R)=\int_{A-\bar{H}}^{\infty}\mathrm{d}x\ \rho_0(A;x+\bar{H}) e^{-\beta_R(x+\bar{H}-A)} \phi_{\pm}(x).
	\end{split}
\end{equation}
We choose $\phi_\pm$ which satisfy the following properties
\begin{equation}\label{phipm:condition}
	\begin{split}
		(a)&\ \abs{\phi_{\pm}(x)}\leqslant \frac{C_{\pm}}{1+x^2},\quad C_\pm<\infty \\
		(b)&\ \hat{\phi}_{\pm}(0)\neq0 \\
	\end{split}
\end{equation}
Below we will study $I^{\rm dual,R}_{\pm,\rm vac}(A,\bar{H};\beta_R)$ with $\beta_R=2\pi\sqrt{\frac{A}{\bar{H}-A}}$ in two special limits
\begin{equation}\label{limit:HbarA}
	\begin{split}
		\mathrm{I}.&\ A\geqslant A_0>0,\quad \bar{H},\frac{\bar{H}}{A}\rightarrow\infty; \\
		\mathrm{II}.&\ A,\bar{H}\rightarrow\infty,\quad\frac{\bar{H}}{A}=s>1\ {\rm fixed}.
	\end{split}
\end{equation}
\begin{remark}
	For the purposes of this paper, our focus is primarily on considering the case of limit I. However, interestingly, in our analysis, it requires minimal effort to include the case of limit II as well. Consequently, we provide an argument that encompasses both cases simultaneously, with the hope that the inclusion of limit II will be beneficial for future studies.
\end{remark}
We split the kernel $\rho_0$ in (\ref{modularkernel}) into four parts
\begin{equation}
	\begin{split}
		\rho_0(A;\bar h)=&\rho_{0}^{(1)}(A;\bar h)-\rho_{0}^{(2)}(A;\bar h)+\rho_{0}^{(3)}(A;\bar h)-\rho_{0}^{(4)}(A;\bar h), \\
		\rho_{0}^{(1)}(A;\bar h)=&\dfrac{e^{4\pi\sqrt{A(\bar h-A)}}}{2\sqrt{(\bar h-A)\pi}}, \\
		\rho_{0}^{(2)}(A;\bar h)=&\dfrac{e^{4\pi\sqrt{(A-1)(\bar h-A)}}}{2\sqrt{(\bar h-A)\pi}},\\
		\rho_{0}^{(3)}(A;\bar h)=&\dfrac{e^{-4\pi\sqrt{A(\bar h-A)}}}{2\sqrt{(\bar h-A)\pi}}, \\
		\rho_{0}^{(4)}(A;\bar h)=&\dfrac{e^{-4\pi\sqrt{(A-1)(\bar h-A)}}}{2\sqrt{(\bar h-A)\pi}}.\\
	\end{split}
\end{equation}
Correspondingly, $I^{\rm dual,R}_{\pm,\rm vac}$ is split into four parts
\begin{equation}
	\begin{split}
		I^{\rm dual,R}_{\pm,\rm vac}=I^{R(1)}_{\pm,\rm vac}+I^{R(2)}_{\pm,\rm vac}+I^{R(3)}_{\pm,\rm vac}+I^{R(4)}_{\pm,\rm vac}.
	\end{split}
\end{equation}
We would like to show that for $\hat{\phi}_\pm(0)\neq0$:\footnote{By definition, we have $\hat{\phi}_{+}(0)>0$, whereas it is not always the case that $\hat{\phi}_{-}(0)\neq0$ (see the defining properties of $\phi_{\pm}$ in \eqref{phipm}) and \eqref{phipm:suppcondition}. In this paper, we will explicitly choose a specific form of $\phi_\pm$ that ensures $\hat{\phi}_{\pm}(0)\neq0$.}
\begin{itemize}
	\item In the limit I, the dominant contribution to $I^{\rm dual,R}_{\pm,\rm vac}$ comes from $I^{R(1)}_{\pm,\rm vac}$. Consequently, $I^{\rm dual}_{\pm,\rm vac}$ has the following asymptotic behavior:
	\begin{equation}
		\begin{split}
			I^{\rm dual}_{\pm,\rm vac}\stackrel{\rm I}{\sim}\left(\dfrac{4\pi^2}{\beta_L}\right)^{3/2}\sqrt{\frac{\pi}{\bar{H}}}e^{2\pi\sqrt{A\bar{H}}}\hat{\phi}_\pm(0).
		\end{split}
	\end{equation}
	\item In the limit II, the dominant contribution to $I^{\rm dual,R}_{\pm,\rm vac}$ comes from $I^{R(1)}_{\pm,\rm vac}$ and $I^{R(2)}_{\pm,\rm vac}$. Consequently, $I^{\rm dual}_{\pm,\rm vac}$ has the following asymptotic behavior:
	\begin{equation}
		\begin{split}
			I^{\rm dual}_{\pm,\rm vac}\stackrel{\rm II}{\sim}\left(1+e^{-2\pi\sqrt{s-1}}\right)\left(\dfrac{4\pi^2}{\beta_L}\right)^{3/2}\sqrt{\frac{\pi}{\bar{H}}}e^{2\pi\sqrt{A\bar{H}}}\hat{\phi}_\pm(0).
		\end{split}
	\end{equation}
    where the first and second terms correspond to the contributions from $I^{R(1)}_{\pm,\rm vac}$ and $I^{R(2)}_{\pm,\rm vac}$, respectively.
\end{itemize}

\subsection{Estimating $I^{R(1)}_{\pm,\rm vac}$ (the dominant term)}
By definition, $I^{R(1)}_{\pm,\rm vac}$ is given by
\begin{equation}
	\begin{split}
		I^{R(1)}_{\pm,\rm vac}:=&\int_{A-\bar{H}}^{\infty}\mathrm{d}x\  \rho_{0}^{(1)}\left(A;\bar x+\bar{H}\right) e^{-\beta_R\left(x+\bar{H}-A\right)} \phi_{\pm}(x), \\
		\rho_{0}^{(1)}(A;\bar h)=&\dfrac{e^{4\pi\sqrt{A(\bar h-A)}}}{2\sqrt{(\bar h-A)\pi}}, \\
	\end{split}
\end{equation}
We rewrite $I^{R(1)}_{\pm,\rm vac}$ as
\begin{equation}\label{I1:rewrite}
	\begin{split}
		I^{R(1)}_{\pm,\rm vac}(A,\bar{H},\beta_R)=&\dfrac{e^{\frac{4\pi^2 A}{\beta_R}}}{2\sqrt{\pi(\bar{H}-A)}}\int_{A-\bar{H}}^{\infty}\mathrm{d}x\ \sqrt{\frac{\bar{H}-A}{x+\bar{H}-A}}e^{-\beta_R \left(\sqrt{x+\bar{H}-A}-\frac{2\pi\sqrt{A}}{\beta_R}\right)^2}  \phi_{\pm}(x).
	\end{split}
\end{equation}
We split the domain of integration into two parts
\begin{equation}\label{I*int:split}
	\begin{split}
		\int_{A-\bar{H}}^{\infty}=\int_{A-\bar{H}}^{-\bar{H}^\tau}+\int_{-\bar{H}^\tau}^{\infty},
	\end{split}
\end{equation}
where $\tau\in(\frac{1}{2},1)$ is some fixed constant.
The first integral is bounded by
\begin{equation}
	\begin{split}
		&\left|\int_{A-\bar{H}}^{-\bar{H}^\tau}dx\ \sqrt{\frac{\bar{H}-A}{x+\bar{H}-A}}e^{-\beta_R \left(\sqrt{x+\bar{H}-A}-\frac{2\pi\sqrt{A}}{\beta_R}\right)^2}  \phi_{\pm}(x)\right|\\ 
		\leqslant
		&\left(\max\limits_{A-\bar{H}\leqslant x\leqslant-\bar{H}^\tau}\left\{|\phi_\pm(x)|\right\}\right)\left(\int_{A-\bar{H}}^{-\bar{H}^\tau}dx\ \sqrt{\frac{\bar{H}}{x+\bar{H}-A}}\right) \\
		=&\left(\max\limits_{A-\bar{H}\leqslant x\leqslant-\bar{H}^\tau}\left\{|\phi_\pm(x)|\right\}\right)\left(2\sqrt{\bar{H}(\bar{H}-\bar{H}^\tau-A)}\right)\\
	\end{split}
\end{equation}
Here we bounded the exponential factor by 1 and bounded $\phi_{\pm}$ by its maximal value. Since $\phi_{\pm}$ has the upper bound (\ref{phipm:condition}), we get
\begin{equation}
	\begin{split}
		&\left|\int_{A-\bar{H}}^{-\bar{H}^\tau}dx\ \sqrt{\frac{\bar{H}-A}{x+\bar{H}-A}}e^{-\beta_R \left(\sqrt{x+\bar{H}-A}-\frac{2\pi\sqrt{A}}{\beta_R}\right)^2}  \phi_{\pm}(x)\right|\leqslant\frac{2C_{\pm}\bar{H}}{1+\bar{H}^{2\tau}}\rightarrow0
	\end{split}
\end{equation}
in the limit (\ref{limit:HbarA}).

The second integral is controlled as follows. We rewrite the integral as
\begin{equation}
	\begin{split}
		&\int_{-\infty}^{\infty}\mathrm{d}x\ Y_{\bar{H}}(x), \\
		&Y_{\bar{H}}(x):=\sqrt{\frac{\bar{H}-A}{x+\bar{H}-A}}e^{-\beta_R \left(\sqrt{x+\bar{H}-A}-\frac{2\pi\sqrt{A}}{\beta_R}\right)^2}  \phi_{\pm}(x)\theta(x+\bar{H}^\tau).
	\end{split}
\end{equation}
We have
\begin{equation}
	\begin{split}
		\sqrt{\frac{\bar{H}-A}{x+\bar{H}-A}}\theta(x+\bar{H}^\tau)\leqslant\sqrt{\frac{\bar{H}}{\bar{H}-\bar{H}^\tau-A}}\leqslant\sqrt{2}
	\end{split}
\end{equation}
for $\bar{H}\geqslant\max\left\{4A,4^{1/(1-\tau)}\right\}$. Let us explain the above bound. Since $\bar{H}\geqslant\max\left\{4A,4^{1/(1-\tau)}\right\}$, we have $\bar H\geqslant 4A$ and $\bar{H} \geqslant 4 \bar H^\tau$, leading to $\bar H-A-\bar H^\tau \geqslant \bar H/2$.  Plugging this in,  we obtain $\sqrt{\frac{\bar{H}}{\bar{H}-\bar{H}^\tau-A}}\leqslant\sqrt{2}$.

Now the exponential factor in $Y_{\bar{H}}(x)$ is obviously bounded by 1. So we get
\begin{equation}
	\begin{split}
		|Y_{\bar{H}}(x)|\leqslant\sqrt{2}|\phi_\pm(x)|
	\end{split}
\end{equation}
for sufficiently large $\bar{H}$ and the r.h.s. is an integrable function on $\mathbb{R}$. Now we consider the point-wise limit of $Y_{\bar{H}}(x)$ as $\bar{H}$ goes to $\infty$. First, using $\tau>0$ we have
\begin{equation}
	\begin{split}
		\sqrt{\frac{\bar{H}-A}{x+\bar{H}-A}}\theta(x+\bar{H}^\tau)\stackrel{(\ref{limit:HbarA})}{\longrightarrow}1
	\end{split}
\end{equation}
for fixed $x$. So far we do not have good control on the exponential factor $e^{-\beta_R \left(\sqrt{x+\bar{H}-A}-\frac{2\pi\sqrt{A}}{\beta_R}\right)^2}$ (other than knowing that it is bounded from above by 1). We would like to make this factor reach its maximal value (i.e.~1) in the limit. For this purpose we choose $\beta_R=2\pi\sqrt{\frac{A}{\bar{H}-A}}$, which gives
\begin{equation}
	\begin{split}
		\beta_R \left(\sqrt{x+\bar{H}-A}-\frac{2\pi\sqrt{A}}{\beta_R}\right)^2=&2\pi\sqrt{\frac{A}{\bar{H}-A}}\frac{x^2}{\left(\sqrt{x+\bar{H}-A}+\sqrt{\bar{H}-A}\right)^2} \\
		\leqslant&\frac{2\pi}{\sqrt{s-1}}\frac{x^2}{\left(\sqrt{x+\bar{H}-A}+\sqrt{\bar{H}-A}\right)^2} \\
		\rightarrow&0 \\
	\end{split}
\end{equation}
for fixed $x$ in the limit (\ref{limit:HbarA}). Thus we get
\begin{equation}
	\begin{split}
		Y_{\bar{H}}(x)\stackrel{(\ref{limit:HbarA})}{\longrightarrow}\phi_\pm(x)
	\end{split}
\end{equation}
for fixed $x$. Then using the dominated convergence theorem we conclude that
\begin{equation}
	\begin{split}
		\int_{-\infty}^{\infty}\mathrm{d}x\ Y_{\bar{H}}(x)\stackrel{(\ref{limit:HbarA})}{\longrightarrow}\int_{-\infty}^{\infty}d x\ \phi_{\pm}(x)=2\pi\hat{\phi}_\pm(0). \\
	\end{split}
\end{equation}
Putting everything together we get
\begin{equation}\label{eq:I1}
	\begin{split}
		I^{R(1)}_{\pm,\rm vac}\left(A,\bar{H},\beta_R\equiv2\pi\sqrt{\frac{A}{\bar{H}-A}}\right)\stackrel{(\ref{limit:HbarA})}{\sim} \sqrt{\frac{\pi}{\bar{H}-A}}e^{2\pi\sqrt{A({\bar{H}-A)}}}\hat{\phi}_\pm(0).
	\end{split}
\end{equation}
\begin{remark}
	This result holds for both limit I and II in (\ref{limit:HbarA}).
\end{remark}

\subsection{Estimating $I^{R(2)}_{\pm,\rm vac}$}\label{app:estimateIR2}
By definition, $I^{R(2)}_{\pm,\rm vac}$ is given by
\begin{equation}
	\begin{split}
		I^{R(2)}_{\pm,\rm vac}:=&\int_{A-\bar{H}}^{\infty}\mathrm{d}x\  \rho_{0}^{(2)}\left(A;\bar x+\bar{H}\right) e^{-\beta_R\left(x+\bar{H}-A\right)} \phi_{\pm}(x), \\
		\rho_{0}^{(2)}(A;\bar h)=&\dfrac{e^{4\pi\sqrt{(A-1)(\bar h-A)}}}{2\sqrt{(\bar h-A)\pi}}, \\
	\end{split}
\end{equation}
Similarly to the analysis for $I^{R(1)}_{\pm,\rm vac}$, we rewrite $I^{R(2)}_{\pm,\rm vac}$ as
\begin{equation}\label{I2:rewrite}
	\begin{split}
		I^{R(2)}_{\pm,\rm vac}(A,\bar{H},\beta_R)=&\dfrac{e^{\frac{4\pi^2 (A-1)}{\beta_R}}}{2\sqrt{\pi(\bar{H}-A)}}\int_{A-\bar{H}}^{\infty}\mathrm{d}x\ \sqrt{\frac{\bar{H}-A}{x+\bar{H}-A}}e^{-\beta_R \left(\sqrt{x+\bar{H}-A}-\frac{2\pi\sqrt{A-1}}{\beta_R}\right)^2}  \phi_{\pm}(x).
	\end{split}
\end{equation}
We see that its structure is very similar to (\ref{I1:rewrite}), except that $A$ is replaced by $A-1$ in two places. Then we do the same analysis as $I^{R(1)}_{\pm,\rm vac}$. We split the domain of integration into two parts as we have done in (\ref{I*int:split}). In the limit (\ref{limit:HbarA}), the first part goes to zero for the same reason. For the second part, the only different analysis is in the exponential factor. For $\beta=2\pi\sqrt{\frac{A}{\bar{H}-A}}$ we have
\begin{equation}
	\begin{split}
		\beta_R \stackrel{(\ref{limit:HbarA})}{\longrightarrow}
		\begin{cases}
			0
			& ({\rm limit\ I}) \\
			\frac{2\pi}{\sqrt{s-1}} & ({\rm limit\ II}) \\
		\end{cases}.
	\end{split}
\end{equation}
Then for limit I we have
\begin{equation}\label{betaR:limit}
	\begin{split}
		\abs{\frac{I^{R(2)}_{\pm,\rm vac}(A,\bar{H},\beta_R)}{I^{R(1)}_{\pm,\rm vac}(A,\bar{H},\beta_R)}}_{\beta_R=2\pi\sqrt{\frac{A}{\bar{H}-A}}}\leqslant const(\phi_{\pm})e^{-\frac{4\pi^2}{\beta_R}}\rightarrow0
	\end{split}
\end{equation}
because of the exponential suppression.

For limit II, the exponential factor in the integral tends to 1 point-wise:
\begin{equation}
	\begin{split}
		\beta_R \left(\sqrt{x+\bar{H}-A}-\frac{2\pi\sqrt{A-1}}{\beta_R}\right)^2=&\frac{2\pi}{\sqrt{s-1}}\frac{(x+s-1)^2}{\left(\sqrt{x+\bar{H}-A}+\sqrt{\frac{(A-1)(\bar{H}-A)}{A}}\right)^2} \\
		\rightarrow&0.
	\end{split}
\end{equation}
Here we used (\ref{betaR:limit}). Now using the same dominated-convergence-theorem argument as the analysis for $I^{R(1)}_{\pm,\rm vac}$ we conclude that in the limit II, has the similar asymptotic behavior to (\ref{eq:I1}):
\begin{equation}\label{eq:I2}
	\begin{split}
		I^{R(2)}_{\pm,\rm vac}\left(A,\bar{H},\beta_R\equiv2\pi\sqrt{\frac{A}{\bar{H}-A}}\right)\stackrel{\rm II}{\sim} \sqrt{\frac{\pi}{(s-1)A}}e^{2\pi\sqrt{s-1}(A-1)}\hat{\phi}_\pm(0).
	\end{split}
\end{equation}
In this case $I^{R(2)}_{\pm,\rm vac}$ is comparable to $I^{R(1)}_{\pm,\rm vac}$:
\begin{equation}
	\begin{split}
		\abs{\frac{I^{R(2)}_{\pm,\rm vac}(A,\bar{H},\beta_R)}{I^{R(1)}_{\pm,\rm vac}(A,\bar{H},\beta_R)}}_{\beta_R=2\pi\sqrt{\frac{A}{\bar{H}-A}}}\stackrel{\rm II}{\longrightarrow}e^{-2\pi\sqrt{s-1}}=O(1).
	\end{split}
\end{equation}

\subsection{Estimating $I^{R(3)}_{\pm,\rm vac}$ and $I^{R(4)}_{\pm,\rm vac}$}\label{app:estimateIR34}
By definition, $I^{R(3)}_{\pm,\rm vac}$ is given by
\begin{equation}
	\begin{split}
		I^{R(3)}_{\pm,\rm vac}:=&\int_{A-\bar{H}}^{\infty}\mathrm{d}x\  \rho_{0}^{(3)}\left(A;\bar x+\bar{H}\right) e^{-\beta_R\left(x+\bar{H}-A\right)} \phi_{\pm}(x), \\
		\rho_{0}^{(3)}(A;\bar h)=&\dfrac{e^{-4\pi\sqrt{A(\bar h-A)}}}{2\sqrt{(\bar h-A)\pi}}, \\
	\end{split}
\end{equation}
Similarly to the analysis for $I^{R(1)}_{\pm,\rm vac}$, we rewrite $I^{R(3)}_{\pm,\rm vac}$ as
\begin{equation}\label{I3:rewrite}
	\begin{split}
		I^{R(3)}_{\pm,\rm vac}(A,\bar{H},\beta_R)=&\dfrac{e^{\frac{4\pi^2 A}{\beta_R}}}{2\sqrt{\pi(\bar{H}-A)}}\int_{A-\bar{H}}^{\infty}\mathrm{d}x\ \sqrt{\frac{\bar{H}-A}{x+\bar{H}-A}}e^{-\beta_R \left(\sqrt{x+\bar{H}-A}+\frac{2\pi\sqrt{A}}{\beta_R}\right)^2}  \phi_{\pm}(x).
	\end{split}
\end{equation}
We see that $I^{R(3)}_{\pm,\rm vac}$ only differs from $I^{R(1)}_{\pm,\rm vac}$ by a change of the ``$\pm$" sign in the exponential factor. Then the analysis is similar to $I^{R(1)}_{\pm,\rm vac}$. We split the integral into two parts as we have done in (\ref{I*int:split}). In the limit (\ref{limit:HbarA}), the first part goes to zero for the same reason, and the second part is controlled by the dominated convergence theorem. But for $I^{R(3)}_{\pm,\rm vac}$ we have
\begin{equation}
	\begin{split}
		\beta_R(\sqrt{x+\bar{h}-A}+\frac{2\pi\sqrt{A}}{\beta_R})=&2\pi\sqrt{\frac{A}{\bar{H}-A}}\left(\sqrt{x+\bar{H}-A}+\sqrt{\bar{H}-A}\right)^2 \\
		\geqslant&2\pi\sqrt{A(\bar{H}-A)} \\
		\stackrel{(\ref{limit:HbarA})}{\longrightarrow}&\infty. \\
	\end{split}
\end{equation}
Here in the first line we used $\beta_R=2\pi\sqrt{\frac{A}{\bar{H}-A}}$, in the second line we dropped the first term in the bracket, and in the last line used (\ref{limit:HbarA}). This estimate implies that the second part of the integral also goes to zero. Therefore we get
\begin{equation}
	\begin{split}
		\frac{I^{R(3)}_{\pm,\rm vac}(A,\bar{H},\beta_R)}{I^{R(1)}_{\pm,\rm vac}(A,\bar{H},\beta_R)}\bigg{|}_{\beta_R=2\pi\sqrt{\frac{A}{\bar{H}-A}}}\stackrel{(\ref{limit:HbarA})}{\longrightarrow}0\quad(\hat{\phi}_\pm(0)\neq0).
	\end{split}
\end{equation}
By the same analysis, we can also get
\begin{equation}
	\begin{split}
		\frac{I^{R(4)}_{\pm,\rm vac}(A,\bar{H},\beta_R)}{I^{R(1)}_{\pm,\rm vac}(A,\bar{H},\beta_R)}\bigg{|}_{\beta_R=2\pi\sqrt{\frac{A}{\bar{H}-A}}}\stackrel{(\ref{limit:HbarA})}{\longrightarrow}0\quad(\hat{\phi}_\pm(0)\neq0).
	\end{split}
\end{equation}
So we conclude that $I^{R(3)}_{\pm,\rm vac}$ and $I^{R(4)}_{\pm,\rm vac}$ are subleading in the limit (\ref{limit:HbarA}).


\section{Some uniform bounds on $\phi_{\pm,\delta}$}\label{appendix:uniformbounds}
In this appendix, we address the subtleties that arise when we take the limit $\delta\rightarrow0$ or $\delta\rightarrow1$ for the selected functions $\phi_{\pm,\delta}$ (as discussed in section \ref{section:epsilonwindow}, specifically points 1 and 2). Throughout our analysis, we consistently choose the following expressions for $\phi_{\pm,\delta}$, which are \eqref{def:phipm} rescaled under \eqref{phipm:scaling}:
\begin{equation}\label{phipm:rescale}
	\begin{split}
		\phi_{+,\delta}(x)&=\frac{16 \delta ^2 \left(x \cos \left(\frac{\delta  \Lambda }{2}\right) \sin \left(\frac{\Lambda  x}{2}\right)-\delta  \sin \left(\frac{\delta  \Lambda }{2}\right) \cos \left(\frac{\Lambda  x}{2}\right)\right)^2}{\left(x^2-\delta ^2\right)^2 (\delta  \Lambda +\sin (\delta  \Lambda ))^2},\\
		\phi_{-,\delta}(x)&=\frac{4 \delta ^2 \left(x \cos \left(\frac{\Lambda  x}{2}\right)-\delta  \cot \left(\frac{\delta  \Lambda }{2}\right) \sin \left(\frac{\Lambda  x}{2}\right)\right)^2}{x^2 (\delta^2 -x^2) \left(\delta  \Lambda  \cot \left(\frac{\delta  \Lambda }{2}\right)-2\right)^2}. \\
	\end{split}
\end{equation}
Recall that the range of allowed values for $\delta$ is given by $\varepsilon<\delta<1-\varepsilon$. As we approach the limit $\varepsilon\rightarrow0$, it eventually leads us to consider $\delta\rightarrow0$ for $\phi_{+,\delta}$ and $\delta\rightarrow1$ for $\phi_{-,\delta}$. In these limits, the values of the corresponding functions $\phi_{\pm,\delta}$ are obtained point-wise as follows:
\begin{equation}
	\begin{split}
		\phi_{+,0}(x)=\frac{4\sin^2\left(\frac{\Lambda x}{2}\right)}{\Lambda^2x^2},\quad\phi_{-,1}(x)=\frac{4\left[x\cos\left(\frac{\Lambda x}{2}\right)-\cot\left(\frac{\Lambda}{2}\right)\sin\left(\frac{\Lambda x}{2}\right)\right]^2}{x^2(1-x^2)\left[\Lambda\cot\left(\frac{\Lambda}{2}\right)-2\right]^2}.
	\end{split}
\end{equation}
These limits are well-defined functions and are $L^1$-integrable over the real axis. This observation provides strong evidence that $\phi_{\pm,\delta}$ satisfy certain uniform bounds that are essential for our analysis.

We aim to establish the following properties of $\phi_{\pm,\delta}$:
\begin{lemma}\label{lemma:phiplusbounds}
For a fixed $\Lambda\in(0,2\pi]$, the function $\phi_{+,\delta}$ satisfies the following bounds for $\delta\in\left[0,\frac{1}{2}\right]$:
		\begin{equation}\label{phiplusbounds}
			\begin{split}
				\phi_{+,\delta}(x)\leqslant&\begin{cases}
					4 & |x|\leqslant 1, \\
					\frac{64}{\Lambda^2x^2} & |x|\geqslant1, \\
				\end{cases} \\
				 \abs{\frac{\phi_{+,\delta}(x)}{\hat{\phi}_{+,\delta}(0)}}\leqslant &\max\left\{\frac{64}{\Lambda},\ 4\Lambda\right\}, \\
				\abs{\frac{\hat{\phi}_{+,\delta}(x)}{\hat{\phi}_{+,\delta}(0)}}\leqslant &1.
			\end{split}
		\end{equation}
\end{lemma}
\begin{lemma}\label{lemma:phiminusbounds}
Let $w\in\left(\frac{1}{2},1\right)$ and $\Lambda\in\left[\frac{(2w+1)\pi}{2},2\pi w\right]$ be fixed. For $\delta\in\left[\frac{(2w+1)\pi}{2\Lambda},1\right]$, the function $\phi_{-,\delta}$ satisfies the following bounds:
\begin{equation}\label{phiminusbounds}
	\begin{split}
		\phi_{-,\delta}(x)\leqslant&\begin{cases}
			\pi^2 & |x|\leqslant 2, \\
			\frac{4\pi^2}{\Lambda^2(\abs{x}-1)^2} & |x|\geqslant2. \\
		\end{cases} \\
		\abs{\frac{\phi_{-,\delta}(x)}{\hat{\phi}_{-,\delta}(0)}}\leqslant& \Lambda\left(1+\frac{4\cot\left(\frac{(2w-1)\pi}{4}\right)}{(2w+1)\pi}\right)\max\left\{\pi^2,\ \frac{4\pi^2}{\Lambda^2}\right\},  \\
		\abs{\frac{\hat{\phi}_{-,\delta}(x)}{\hat{\phi}_{-,\delta}(0)}}\leqslant&2\pi\Lambda\left(1+\frac{4\cot\left(\frac{(2w-1)\pi}{4}\right)}{(2w+1)\pi}\right)\left(1+\frac{2}{\Lambda^2}\right). \\
	\end{split}
\end{equation}
\end{lemma}

\subsection{Proof of lemma \ref{lemma:phiplusbounds}}
When $\abs{x}\geqslant1$, we have
\begin{equation}
	\begin{split}
		\abs{\phi_{+,\delta}(x)}\equiv&16\left[\frac{\frac{x}{\Lambda}\cos\left(\frac{\delta\Lambda}{2}\right)\sin\left(\frac{\Lambda x}{x}\right)-\frac{\delta}{\Lambda}\sin\left(\frac{\delta \Lambda}{2}\right)\cos\left(\frac{\Lambda x}{2}\right)}{(x^2-\delta^2)\left(1+\frac{\sin(\delta\Lambda)}{\delta\Lambda}\right)}\right]^2 \\
		\leqslant&16\left[\frac{\frac{\abs{x}+\delta}{\Lambda}}{x^2-\delta^2}\right]^2 \\
		=&\frac{16}{\Lambda^2}\left[\frac{1}{\abs{x}-\delta}\right]^2 \\
		\leqslant&\frac{64}{\Lambda^2x^2}. \\
	\end{split}
\end{equation}
In the first line, we rewrite $\phi_{+,\delta}(x)$. In the second line, we bound $\sin$ and $\cos$ by 1 and use the fact that $\frac{\sin(\delta\Lambda)}{\delta\Lambda}\geqslant0$ for $\delta\in\left[0,\frac{1}{2}\right]$ and $\Lambda\in(0,2\pi]$. The third line is a rewrite of the second line. In the last line, we use the fact that $|x|-\delta\geqslant\frac{|x|}{2}$ for $|x|\geqslant1$ and $\delta\in\left[0,\frac{1}{2}\right]$.

When $0\leqslant x\leqslant1$ we have
\begin{equation}
	\begin{split}
		\abs{\phi_{+,\delta}}(x)\equiv&4\left[\frac{\frac{2\delta\sin\left(\frac{\Lambda(x-\delta)}{2}\right)}{\Lambda(x-\delta)}+\frac{2\cos\left(\frac{\delta\Lambda}{2}\right)\sin\left(\frac{\Lambda x}{2}\right)}{\Lambda}}{(\delta+x)\left(1+\frac{\sin(\delta\Lambda)}{\delta\Lambda}\right)}\right]^2 \\
		\leqslant&4\left[\frac{\delta+x}{(\delta+x)\left(1+\frac{\sin(\delta\Lambda)}{\delta\Lambda}\right)}\right] \\
		=&4. \\
	\end{split}
\end{equation}
In the first line, we rewrite $\phi_{+,\delta}(x)$. In the second line, we use the inequality $|\frac{\sin x}{x}|\leqslant1$ for any $x\in\mathbb{R}$ and $\sin x\leqslant x$ for $x\geqslant0$. In the last line, we use the fact that $\frac{\sin(\delta\Lambda)}{\delta\Lambda}\geqslant0$ for $\delta\in\left[0,\frac{1}{2}\right]$ and $\Lambda\in(0,2\pi]$. The same bound on $\phi_{+,\delta}(x)$ holds for $-1\leqslant x\leqslant0$ because it is an even function of $x$. This completes the proof of the first inequality in \eqref{phiplusbounds}.

Next, we derive a uniform upper bound on the ratio $\abs{\frac{\phi_{+,\delta}(x)}{\hat{\phi}_{+,\delta}(0)}}$. Using \eqref{hatphipm:optimalR}, we have
\begin{equation}
	\begin{split}
		\hat{\phi}_{+,\delta}(0)\geqslant\frac{1}{\Lambda}
	\end{split}
\end{equation}
for $\Lambda\in(0,2\pi]$ and $\delta\in\left[0,\frac{1}{2}\right]$. Together with the first inequality in \eqref{phiplusbounds}, we obtain
\begin{equation}
	\begin{split}
		\abs{\frac{\phi_{+,\delta}(x)}{\hat{\phi}_{+,\delta}(0)}}\leqslant\max\left\{\frac{64}{\Lambda},\ 4\Lambda\right\}\quad\left(0\leqslant\delta\leqslant\frac{1}{2},\ 0<\Lambda\leqslant2\pi\right).
	\end{split}
\end{equation}
The bound on $|\frac{\hat{\phi}{+,\delta}(x)}{\hat{\phi}{+,\delta}(0)}|$ follows trivially from the fact that $\phi{+,\delta}(x)\geqslant0$ for all $x$. Hence, we have
\begin{equation}
	\begin{split}
		\abs{\frac{\hat{\phi}_{+,\delta}(x)}{\hat{\phi}_{+,\delta}(0)}}\leqslant1.
	\end{split}
\end{equation}
This completes the proof of Lemma \ref{lemma:phiplusbounds}.

\subsection{Proof of lemma \ref{lemma:phiminusbounds}}
We note that $\phi_{-,\delta}(x)$ is an even function of $x$, so it is sufficient to prove the first inequality in \eqref{phiminusbounds} for $x\geqslant0$. We rewrite $\phi_{-,\delta}(x)$ as follows:
\begin{equation}\label{phiminus:rewrite}
	\begin{split}
		\phi_{-,\delta}(x)=\frac{\delta-x}{\delta+x}\left[\frac{2\delta}{\delta\Lambda\cos\left(\frac{\delta\Lambda}{2}\right)-2\sin\left(\frac{\delta\Lambda}{2}\right)}\right]^2\left[\frac{\sin\left(\frac{\Lambda(\delta-x)}{2}\right)}{\delta-x}-\frac{\sin\left(\frac{\Lambda x}{2}\right)\cos\left(\frac{\delta\Lambda}{2}\right)}{x}\right]^2.
	\end{split}
\end{equation}
When $x\geqslant0$, the first factor in \eqref{phiminus:rewrite} is bounded by 1. For the second factor, we use the condition on the range of $\delta$ and $\Lambda$, which implies that $\frac{\delta\Lambda}{2}\in\left[\frac{\pi}{2},\pi\right]$. Thus, the second factor is bounded as follows:
\begin{equation}
	\begin{split}
		\left[\frac{2\delta}{\delta\Lambda\cos\left(\frac{\delta\Lambda}{2}\right)-2\sin\left(\frac{\delta\Lambda}{2}\right)}\right]^2\leqslant\frac{\pi^2}{\Lambda^2},
	\end{split}
\end{equation}
where we used the fact that the function $f(x)\equiv\frac{1}{\left[\cos x-\frac{\sin x}{x}\right]^2}$ is maximized at $x=\frac{\pi}{2}$ with the maximum value of $\frac{\pi^2}{4}$ for $\frac{\pi}{2}\leqslant x\leqslant\pi$. This can be seen from the plot of $f(x)$ shown in figure \ref{figure:fx}.
\begin{figure}
	\centering
	\begin{tikzpicture}
		\draw[->] (pi/2-0.5, 0.5) -- (pi+0.5, 0.5) node[right] {$x$};
		\draw[->] (pi/2-0.5, 0.5) -- (pi/2-0.5, 2.8) node[above] {$f(x)$};
		\draw[scale=1, domain=pi/2:pi, smooth, variable=\x, blue] plot ({\x}, {1/((cos(\x r)-sin(\x r)/(\x))*(cos(\x r)-sin(\x r)/(\x)))});
	\end{tikzpicture}
    \caption{Plot of $f(x)=\frac{1}{\left[\cos x-\frac{\sin x}{x}\right]^2}$ in the regime $\frac{\pi}{2}\leqslant x\leqslant\pi$.}
    \label{figure:fx}
\end{figure}
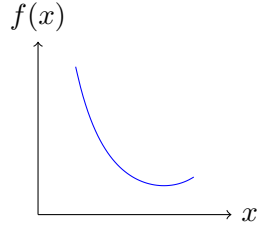
The last factor in \eqref{phiminus:rewrite} is bounded as follows:
\begin{equation}
	\begin{split}
		\left[\frac{\sin\left(\frac{\Lambda(\delta-x)}{2}\right)}{\delta-x}-\frac{\sin\left(\frac{\Lambda x}{2}\right)\cos\left(\frac{\delta\Lambda}{2}\right)}{x}\right]^2\leqslant\left[\frac{\Lambda}{2}+\frac{\Lambda}{2}\right]^2=\Lambda^2.
	\end{split}
\end{equation} 
This bound is obtained by using the inequalities $\abs{\frac{\sin x}{x}}\leqslant1$, $\abs{\sin x}\leqslant\abs{x}$, and $\abs{\cos x}\leqslant1$.

Combining the inequalities derived above, we obtain the following bounds for $\phi_{-,\delta}(x)$:
\begin{equation}
	\begin{split}
		\abs{\phi_{-,\delta}(x)}\leqslant\pi^2\quad(\forall\ x\in\mathbb{R}).
	\end{split}
\end{equation}
In particular, this bound holds for $\abs{x}\leqslant2$.

When $x\geqslant2$, we use the condition $\delta\leqslant1$ and further refine the bound. The last factor in \eqref{phiminus:rewrite} is bounded by
\begin{equation}
	\begin{split}
		\left[\frac{\sin\left(\frac{\Lambda(\delta-x)}{2}\right)}{\delta-x}-\frac{\sin\left(\frac{\Lambda x}{2}\right)\cos\left(\frac{\delta\Lambda}{2}\right)}{x}\right]^2\leqslant\left[\frac{1}{x-1}+\frac{1}{x}\right]^2\leqslant\frac{4}{(x-1)^2}.
	\end{split}
\end{equation}
Thus, we obtain the improved bound
\begin{equation}
	\begin{split}
		\abs{\phi_{-,\delta}(x)}\leqslant\frac{4\pi^2}{\Lambda^2(\abs{x}-1)^2}\quad(\abs{x}\geqslant2).
	\end{split}
\end{equation}
This completes the proof of the first inequality in \eqref{phiminusbounds}.

For the second and third inequalities in \eqref{phiminusbounds}, we make use of the explicit expression of $\hat{\phi}_{-,\delta}$ given in \eqref{hatphipm:optimalR}:
\begin{equation}\label{hatphiminus:lowerbound}
	\begin{split}
		\hat{\phi}_{-,\delta}(0)=\frac{1}{\Lambda}\frac{1}{1-\frac{2\tan\left(\frac{\Lambda\delta}{2}\right)}{\Lambda\delta}}\geqslant\frac{1}{\Lambda}\frac{1}{1+\frac{4\cot\left(\frac{(2w-1)\pi}{4}\right)}{(2w+1)\pi}}.
	\end{split}
\end{equation}
Here, we used the condition $\frac{(2w+1)\pi}{4}\leqslant\frac{\delta\Lambda}{2}\leqslant\pi$, where $\frac{1}{2}< w<1$, and the fact that $\frac{1}{1-\frac{\tan x}{x}}$ is monotonically increasing for $\frac{\pi}{2}\leqslant x\leqslant\pi$. Consequently, $\hat{\phi}_{-,\delta}(0)$ possesses a strictly positive lower bound that remains uniform in $\delta$ within the specified range.

The second inequality of \eqref{phiminusbounds} follows from \eqref{hatphiminus:lowerbound} and the first inequality of \eqref{phiminusbounds}. Similarly, the third inequality of \eqref{phiminusbounds} is obtained by combining \eqref{hatphiminus:lowerbound} and the first inequality of \eqref{hatphiminus:lowerbound} with the additional estimate:
\begin{equation}
	\begin{split}
		\abs{\hat{\phi}_{-,\delta}(x)}\leqslant\int \frac{dy}{2\pi} \abs{\phi_{-,\delta}(y)}\leqslant 2\pi+\frac{4\pi}{\Lambda^2}.
	\end{split}
\end{equation}
Therefore, we have established the validity of the second and third inequalities in \eqref{phiminusbounds}, completing the proof of lemma \ref{lemma:phiminusbounds}.

\section{Examples of CFTs with a shifted twist accumulation point}\label{appendix:example}
The purpose of this appendix is to provide technical details for the examples discussed in sections \ref{section:twistshift1} and \ref{section:twistshift2}.

\subsection{Decoupled irrational CFTs}\label{app:decoupledCFT}
The first example considers a CFT composed of several decoupled copies of CFTs:
$${\rm CFT}={\rm CFT}^{(1)}\otimes{\rm CFT}^{(2)}\otimes\ldots\otimes{\rm CFT}^{(N)}.$$
By construction, the symmetry algebra $\mathcal{G}$ of the combined CFT includes, at least, the direct sum of Virasoro algebras:
$$\mathcal{G}\supseteq Vir_L^{(1)}\oplus Vir_R^{(1)}\oplus\ldots\oplus Vir_L^{(N)}\oplus Vir_R^{(N)}.$$
However, for the purpose of our analysis, we focus solely on the "naive Virasoro algebra" generated by
\begin{equation}\label{def:naiveVir}
	\begin{split}
		L_n=\sum\limits_{i=1}^{N}L_n^{(i)},\quad\bar{L}_n=\sum\limits_{i=1}^{N}\bar{L}_n^{(i)}.
	\end{split}
\end{equation}
We assume that each individual CFT$^{(i)}$ has a central charge $c^{(i)}>1$, a unique twist-0 primary state (the vacuum), and a nonzero twist gap $\tau_{\rm gap}^{(i)}>0$. 

\subsubsection{Central charge, vacuum, currents and the twist gap}
The central charge and the twist gap of the combined CFT are then given by
\begin{equation}
	\begin{split}
		c=\sum\limits_{i=1}^{N}c^{(i)},\quad\tau_{\rm gap}=\min\{4,\tau_{\rm gap}^{(1)},\ldots,\tau_{\rm gap}^{(N)}\},
	\end{split}
\end{equation}
where the ``4" in the twist gap corresponds to the states
\begin{equation}\label{example1:spin2current}
	\begin{split}
		\left(\sum\limits_{i=1}^Nu_i L^{(i)}_{-2}\right)\left(\sum\limits_{i=1}^Nv_i \bar{L}^{(i)}_{-2}\right)\ket{\rm vac},
	\end{split}
\end{equation}
subject to the constraints
\begin{equation}\label{example1:spin2constraint}
	\begin{split}
		\sum\limits_{i=1}^Nu_i\,c^{(i)}=\sum\limits_{i=1}^Nv_i\,c^{(i)}=0.
	\end{split}
\end{equation}
These states are twist-4 scalar primaries of the "naive Virasoro algebra" \eqref{def:naiveVir}, and there are $(N-1)^2$ of them (by constraints \eqref{example1:spin2constraint}).

From the perspective of the ``naive Virasoro algebra," the theory contains infinitely many twist-0 primary states, which correspond to chiral currents. This can be argued in the following way. 

Let us focus on the left currents since the argument for the right currents is analogous. We consider the representation of the direct sum of Virasoro algebras:
$$\mathcal{U}\left[Vir_L^{(1)}\oplus Vir_R^{(1)}\oplus\ldots\oplus Vir_L^{(N)}\oplus Vir_R^{(N)}\right]\ket{h_1,\bar{h}_1,h_2,\ldots,\bar{h}_N},$$
where $\mathcal{U}[\mathcal{G}]$ denotes the universal enveloping algebra of the Lie algebra $\mathcal{G}$, and $\left\{h_i,\bar{h}_i\right\}$ is the set of the highest weights of the representation. The primaries of the ``naive Virasoro algebra" are obtained by taking linear combinations of the elements with the same conformal weights in the above representation. By considering all possible $\ket{h_1,\bar{h}_1,h_2,\ldots,\bar{h}_N}$, we obtain all the primaries of the ``naive Virasoro algebra." It is important to note that these primaries can only have conformal weights satisfying:
\begin{equation}\label{example1:confweightineq}
	\begin{split}
		h\geqslant h_1+\ldots+h_N,\quad \bar{h}\geqslant\bar{h}_1+\ldots+\bar{h}_N.
	\end{split}
\end{equation}
Since we assumed a twist gap $\tau_{\rm gap}^{(i)}$ for each CFT$^{(i)}$, the left currents, which have $\bar{h}=0$, can only be descendants of the vacuum state. Moreover, they can only be obtained by acting with $L^{(i)}_{-n}$'s (not $\bar{L}^{(i)}_{-n}$'s!) on the vacuum state. Consequently, the counting of the left currents reduces to considering the product of the left vacuum characters:
\begin{equation}\label{example1:vaccharacter}
	\begin{split}
		\prod\limits_{i=1}^{N}\chi_{\rm vac}^{(i)}(q)=&\prod\limits_{i=1}^{N}\frac{q^{-\frac{c^{(i)}-1}{24}}(1-q)}{\eta(q)}\equiv\frac{q^{-\frac{\sum\limits_{i=1}^{N}c^{(i)}-1}{24}}}{\eta(q)}f(q), \\
		f(q)=&\frac{q^{\frac{N-1}{24}}}{\eta(q)^{N-1}}(1-q)^N. \\
	\end{split}
\end{equation}
The function $f(q)$ is related to $\mathcal{F}_{\rm current}(q)$ (see \eqref{Fcurrent:asymp:withshift}) through the equation
\begin{equation}\label{fFrelation}
	\begin{split}
		f(q)=1-q+\mathcal{F}_{\rm current}(q)\quad\left(q=e^{-\beta}\right).
	\end{split}
\end{equation}
The coefficients of $q^j$ in $\mathcal{F}{\rm current}(q)$, which correspond to the total number of spin-$j$ left currents, are the same as those in $f(q)$ for $j\geqslant2$. Taking the limit $q\rightarrow1$, we have
\begin{equation}
	\begin{split}
		f(1)=\mathcal{F}_{\rm current}(1)={\rm total\ number\ of\ left\ currents}.
	\end{split}
\end{equation}
By \eqref{example1:vaccharacter}, $f(q)$ diverges as $q\rightarrow1$. This demonstrates that there are infinitely many left currents.

The uniqueness of the vacuum state and the candidate twist gap ``4" can also be observed from the character formula. By considering \eqref{example1:confweightineq}, we find that the vacuum state of the ``naive Virasoro algebra" and any potential twist gap other than $\tau_{\rm gap}^{(i)}$ can only arise from the descendants of the vacuum state. Let us decompose the product of vacuum characters into characters of the ``naive Virasoro algebra":
\begin{equation}\label{example1:vacdecomp}
	\begin{split}
		\prod\limits_{i=1}^{N}\chi_{\rm vac}^{(i)}(q)\chi_{\rm vac}^{(i)}(\bar{q})=\frac{(q\bar{q})^{-\frac{\sum\limits_{i=1}^{N}c^{(i)}-1}{24}}}{\eta(q)\eta(\bar{q})}f(q)f(\bar{q}),
	\end{split}
\end{equation}
where $f(q)$ is the same function as in \eqref{example1:vaccharacter}. In \eqref{example1:vacdecomp}, the prefactor accounts for the contribution from descendants of the "naive Virasoro algebra," and the coefficients of $q^h\bar{q}^{\bar{h}}$ in the power series expansion of $f(q)f(\bar{q})$ correspond to the number of Virasoro primaries with conformal weights $(h,\bar{h})$ (with the exception that we need to use $1-q$ instead of $1$ for $h=0$, and similarly for $\bar{h}$). The expression for $f(q)f(\bar{q})$ is given by
\begin{equation}
	\begin{split}
		f(q)f(\bar{q})=&(1-q)(1-\bar{q})+(1-q)\mathcal{F}_{\rm current}(\bar{q})+\mathcal{F}_{\rm current}(q)(1-\bar{q})+(N-1)^2q^2\bar{q}^2 \\
		&+O(q^3\bar{q}^2)+O(q^2\bar{q}^3).
	\end{split}
\end{equation}
In this expression, the first term corresponds to the vacuum state, the second and third terms correspond to the chiral currents, the fourth term corresponds to the twist-4 scalar state described earlier, and the error terms correspond to operators with twists $\geqslant4$. This decomposition justifies the uniqueness of the vacuum state and the possibility of a candidate twist gap of ``4".

Now let us analyze the growth of the current degeneracy $D(j)$ by examining the asymptotic behavior of $f(q)$ as $q\rightarrow 1$ (or equivalently, $\beta\rightarrow 0$) using the relation \eqref{fFrelation}. We use of the property of the modular form $\eta(q)$ under S modular transformation, given by
\begin{equation}
	\begin{split}
		\eta(q)=\sqrt{\frac{2\pi}{\beta}}\eta(q')\quad\left(q=e^{-\beta},\ q'=e^{-\frac{4\pi^2}{\beta}}\right).
	\end{split}
\end{equation}
Applying this to $f(q)$, we find that
\begin{equation}
	\begin{split}
		f(q)=\left(\frac{\beta}{2\pi}\right)^{\frac{N-1}{2}}\frac{q^{\frac{N-1}{24}}}{\eta(q')^{N-1}}(1-q)^N\stackrel{q\rightarrow1}{\sim}\left(\frac{\beta}{2\pi}\right)^{\frac{N-1}{2}}\beta^{N}\,e^{\frac{4\pi^2}{\beta}\frac{N-1}{24}}\quad\left(q'=e^{-\frac{4\pi^2}{\beta}}\right). \\
	\end{split}
\end{equation}
Then by \eqref{Fcurrent:asymp:withshift}, we get $b=\frac{1}{2}$ and $a=\sqrt{\frac{N-1}{24}}$ in the notation of \eqref{D:ansatz}. 

Since we assumed that $c^{(i)}>1$, our case falls under the condition $b=\frac{1}{2}$ and $0<a<\sqrt{A}$, which corresponds to conjecture \ref{conjecture:infinitespin:shift}. Therefore, let's examine the predictions made by this conjecture.

\subsubsection{Testing conjecture \ref{conjecture:infinitespin:shift}}
The first part of conjecture \ref{conjecture:infinitespin:shift} says that the big CFT should contain a twist accumulation point given by
\begin{equation}\label{shift:example1}
	\begin{split}
		h=A-a^2=\frac{\sum\limits_{i=1}^{N}c^{(i)}-1}{24}-\frac{N-1}{24}=\sum\limits_{i=1}^{N}A^{(i)},\quad A^{(i)}=\frac{c^{(i)}-1}{24}.
	\end{split}
\end{equation}
This is expected because each individual CFT$^{(i)}$ contains a twist accumulation point at 
$$h=A^{(i)},\quad\bar{h}=\infty.$$
Let $\left\{\mathcal{O}_k^{(i)}\right\}$ represent the family of Virasoro primary operators in CFT$^{(i)}$ that approach the above twist accumulation point. In the big CFT, we have the primaries given by
\begin{equation}\label{example1:operatorprod}
	\begin{split}
		\mathcal{O}_k(x)=\prod\limits_{i=1}^{N}\mathcal{O}_k^{(i)}(x).
	\end{split}
\end{equation}
This family of primaries will approach the twist accumulation point with $h$ given by \eqref{shift:example1}. Consequently, the first part of conjecture \ref{conjecture:infinitespin:shift} is confirmed.

The second part of conjecture \ref{conjecture:infinitespin:shift} is challenging to verify due to the limited information available on each individual CFT$^{(i)}$. One naive intuition is that, in the double lightcone limit, the dominant contribution to $\mathcal{A}_J(\beta_{L},\varepsilon)$ arises from operators of the form \eqref{example1:operatorprod}. However, this intuition overlooks a significant portion of Virasoro primaries that can be generated by applying $\bar{L}^{(i)}_{-k}$ operators to \eqref{example1:operatorprod}.

To count the number of Virasoro primaries with the same $h$ ($=\sum\limits_{i=1}^{N} h_k^{(i)}$) generated from each operator of the form \eqref{example1:operatorprod}, we examine the product of the characters of the right movers:
\begin{equation}
	\begin{split}
		\prod\limits_{i=1}^{N}\chi_{\bar{h}_k^{(i)}}(\bar{q})=&\prod\limits_{i=1}^{N}\frac{\bar{q}^{\sum\limits_{i}(\bar{h}^{(i)}_k-A^{(i)})}}{\eta(\bar{q})}=\frac{\bar{q}^{-A+\sum\limits_{i}\bar{h}^{(i)}_k}}{\eta(\bar{q})}f(\bar{q}), \\
	\end{split}
\end{equation}
where $f(\bar{q})$ is the same function as in \eqref{example1:vaccharacter}. The coefficient of the $\bar{q}^n$ term in $f(\bar{q})$ represents the number of generated Virasoro primaries with
$$h_k=\sum\limits_{i=1}^{N} h_k^{(i)},\quad\bar{h}_k=\sum\limits_{i=1}^{N} \bar{h}_k^{(i)}+n.$$
From our earlier analysis of the vacuum character, we already know that for large values of $n$, the coefficient of the $\bar{q}^n$ term exhibits growth behavior on the order of $e^{4\pi\sqrt{\frac{N-1}{24}n}}$, up to a slow-growing factor. This provides insight into the growth of the coefficient and confirms the conjecture's prediction.

Taking into account all these operators, we can now make the following hypothesis:
\begin{itemize}
	\item In the double lightcone limit, the leading term of $\log\mathcal{A}_J(\beta_{L},\varepsilon)$ arises from operators of the form \eqref{example1:operatorprod}, including the Virasoro primaries generated from them.
\end{itemize}
Based on this assumption, we approximate $\mathcal{A}_J(\beta_{L},\varepsilon)$ as follows:
\begin{equation}
	\begin{split}
		\mathcal{A}_J(\beta_{L},\varepsilon)\stackrel{{\rm DLC}_w}{\approx}\sum\limits_{n+\sum\limits_{i=1}^{N}J_i=J}\left(\prod\limits_
		{i=1}^{N}\mathcal{A}_{J_i}(\beta_{L},\varepsilon)\right)f_n,
	\end{split}
\end{equation} 
Here, the factor $f_n$ represents the coefficient of $\bar{q}^n$ in $f(\bar{q})$ and counts the number of generated Virasoro primaries, as mentioned earlier. It is important to note that this approximation captures only the leading exponential growth in $J$, as we have neglected the slow-growing factor. 

Next, we apply the second approximation as follows. We expect that for large $J$, the leading exponential growth of $A_J(\beta_{L},\varepsilon)$ in $J$ is dominated by the regime where $n$ and all the $J_i$'s are large. Under this assumption, we use $f_n\sim e^{4\pi\sqrt{\frac{N-1}{24}n}}$ and theorem \ref{theorem:modulartauberianFJ}, which allows us to replace $f_n$ and all $\mathcal{A}_{J_i}$'s with their asymptotic growth, up to slow-growing factors. Consequently, we have:
\begin{equation}
	\begin{split}
		\mathcal{A}_J(\beta_{L},\varepsilon)\approx&\sum\limits_{n+\sum\limits_{i=1}^{N}J_i=J}\left(\prod\limits_
		{i=1}^{N}e^{4\pi\sqrt{A^{(i)} J_i}}\right)e^{4\pi\sqrt{\frac{N-1}{24}n}} \\
		=&\sum\limits_{n+\sum\limits_{i=1}^{N}J_i=J}e^{4\pi\sum\limits_{i}\sqrt{A^{(i)} J_i}+4\pi\sqrt{\frac{N-1}{24}n}}. \\
	\end{split}
\end{equation}
In the limit of $J\rightarrow\infty$, the leading exponential growth of the expression above is determined by maximizing $\sqrt{\frac{N-1}{24}n}+\sum\limits_{i}\sqrt{A^{(i)} J_i}$ subject to the constraint $n+\sum\limits_{i}J_i=J$. This is achieved by the condition:
\begin{equation}\label{example1:realizecondition}
	\begin{split}
		\frac{N-1}{24n}=\frac{A^{(1)}}{J_1}=\frac{A^{(2)}}{J_2}=\ldots=\frac{A^{(N)}}{J_N}\left(=\frac{A}{J}\right).
	\end{split}
\end{equation}
It is worth noting that this condition implies that both the big CFT and all the small CFTs employ the same right inverse temperature in the analysis:
\begin{equation}
	\begin{split}
		\beta_{R}=2\pi\sqrt{\frac{A}{J}},\quad\beta_{R}^{(i)}=2\pi\sqrt{\frac{A^{(i)}}{J_i}}=2\pi\sqrt{\frac{A}{J}}.
	\end{split}
\end{equation}
This provides a consistency check on our approximation.

Therefore, we obtain:
\begin{equation}
	\begin{split}
		\mathcal{A}_J(\beta_{L},\varepsilon)\stackrel{{\rm DLC}_w}{\approx}e^{4\pi\sqrt{\left(\sum\limits_{i}A^{(i)}+\frac{N-1}{24}\right)J}}=e^{4\pi\sqrt{AJ}}.
	\end{split}
\end{equation}
In this expression, we have neglected the slow-growing factor, which is currently beyond our computational capabilities. However, this result is consistent with the prediction in the second part of conjecture \ref{conjecture:infinitespin:shift}, as evidenced by the presence of the $e^{4\pi\sqrt{AJ}}$ term in \eqref{replace:shift}.

\subsection{$W_N$ CFT}\label{app:WNCFT}
The second example is the $W_N$ CFT with central charge $c>N-1$ and a twist gap $\tau_{\rm gap}^{W_N}$ in the spectrum of $W_N$-primaries. This example was studied in \cite{Afkhami-Jeddi:2017idc}, where the authors predicted the existence of a twist accumulation point of $W_N$ primaries, given by
\begin{equation}\label{WN:accumpt}
	\begin{split}
		h=\frac{c-N+1}{24},\quad\bar{h}=\infty.
	\end{split}
\end{equation}
Notably, since every $W_N$ primary state is also a Virasoro primary state, this twist accumulation point holds true for Virasoro primaries as well. It is interesting to observe that when $N=2$, the $W_N$ algebra reduces to the Virasoro algebra, and the expression \eqref{WN:accumpt} coincides with the twist accumulation point in the Virasoro case, without any need for shifting.

In the case of $W_N$ CFT, we aim to verify the following aspects (from perspective of Virasoro algebra):
\begin{itemize}
	\item The theory contains an infinite number of chiral currents. The growth of the current degeneracy $D(j)$ (as indicated in \eqref{Zansatz:withcurrents}) is approximately given by
	\begin{equation}
		\begin{split}
			D(j)\sim e^{4\pi\sqrt{\frac{N-2}{24}j}}\quad(j\rightarrow\infty),
		\end{split}
	\end{equation}
	up to a slow-growing factor in $j$.
	\item It has a unique vacuum and a nonzero twist gap.
	\item It has a twist accumulation point for Virasoro primaries, characterized by \eqref{WN:accumpt}. The counting of spin-$J$ Virasoro primaries near the twist accumulation point follows the pattern $e^{4\pi\sqrt{AJ}}$, up to a potential slow growth factor (see equation \eqref{WN:expcheck} for precise formulation).
\end{itemize}
The first two points establish that the theory falls within the framework of conjecture \ref{conjecture:infinitespin:shift} with $a=\sqrt{\frac{N-2}{24}}$. The existence of the twist accumulation point \eqref{WN:accumpt} directly corresponds to the statement made in the first part of conjecture \ref{conjecture:infinitespin:shift}. While verifying the complete consistency with the second part of conjecture \ref{conjecture:infinitespin:shift} is a complex task, in this context, we solely focus on confirming the exponential growth factor $e^{4\pi\sqrt{AJ}}$ associated with the last point, disregarding the slow-growing factor.

Our main tool for analysis is the $W_N$ character. In the case of a chiral $W_N$ primary with conformal weight $h$, its corresponding $W_N$ representation character is given by \cite{Afkhami-Jeddi:2017idc}:
\begin{equation}\label{def:WNcharacter}
	\begin{split}
		\chi^{W_N}_{h}(q)=\begin{cases}
			\frac{q^{-A+\frac{N-2}{24}}}{\eta(q)^{N-1}}\prod\limits_{n=1}^{N-1}(1-q^n)^{N-n} &\text{if } h = 0,\\
			\frac{q^{h-A+\frac{N-2}{24}}}{\eta(q)^{N-1}}&\text{if } h > 0. \\
		\end{cases}
	\end{split}
\end{equation}
In particular, when $N=2$, these characters reduce to the expected Virasoro characters.

\subsubsection{Virasoro vacuum, twist gap and chiral currents}
Assuming the existence of a positive $W_N$ twist gap $\tau^{W_N}_{\rm gap}$ and a unique twist-0 $W_N$ primary (which corresponds to the $W_N$ vacuum), the torus partition function is given by:
\begin{equation}
	\begin{split}
		Z(q,\bar{q})=\chi^{W_N}_{\rm vac}(q)\chi^{W_N}_{\rm vac}(\bar{q})+\sum\limits_{h,\bar{h}\geqslant\tau^{W_N}_{\rm gap}/2}\chi^{W_N}_{h}(q)\chi^{W_N}_{\bar{h}}(\bar{q}).
	\end{split}
\end{equation}
The Virasoro chiral currents comes from the $W_N$ vacuum character:
\begin{equation}
	\begin{split}
		\chi^{W_N}_{0}(q)=\chi_0(q)+\sum\limits_{j=1}^{\infty}D(j)\chi_j(q).
	\end{split}
\end{equation}
Similar to appendix \ref{app:decoupledCFT}, we introduce the function $f(q)$ to count the number of Virasoro chiral currents:
\begin{equation}
	\begin{split}
		\chi^{W_N}_{0}(q)\equiv&\frac{q^{-A}}{\eta(q)}f(q), \\
		f(q)=&1-q+\sum\limits_{j=1}^{\infty}D(j)q^j. \\
	\end{split}
\end{equation}
Using the $h=0$ case of \eqref{def:WNcharacter}, we find:
\begin{equation}
	\begin{split}
		f(q)=\frac{q^{\frac{N-2}{24}}}{\eta(q)^{N-2}}\prod\limits_{n=1}^{N-1}(1-q^n)^{N-n}.
	\end{split}
\end{equation}
To analyze the asymptotic behavior of $D(j)$, we use the fact that $\eta(q)$ is a modular form, which leads to:
\begin{equation}\label{WN:fq}
	\begin{split}
		f(q)=\left(\frac{2\pi}{\beta}\right)^{-\frac{N-2}{2}}\frac{q^{\frac{N-2}{24}}}{\eta(q')^{N-2}}\prod\limits_{n=1}^{N-1}(1-q^n)^{N-n}\quad \left(q'=e^{-\frac{4\pi^2}{\beta}}\right).
	\end{split}
\end{equation}
Taking the limit $q\rightarrow1$ (or equivalently, $\beta\rightarrow0$), we obtain the asymptotic behavior of $f(q)$:
\begin{equation}\label{fqasymp1}
	\begin{split}
		f(q)\stackrel{q\rightarrow1}{\sim}\left(\frac{2\pi}{\beta}\right)^{-\frac{N-2}{2}}\left(\prod\limits_{n=1}^{N-1}\left(n\beta\right)^{N-n}\right)e^{\frac{N-2}{24}\frac{4\pi^2}{\beta}}.
	\end{split}
\end{equation}
Then by \eqref{Fcurrent:asymp:withshift}, we determine that $a=\sqrt{\frac{N-2}{24}}$, which aligns with our expectations. This completes the verification of the first point.

By using the $W_N$ character formula \eqref{def:WNcharacter}, we observe that a $W_N$ descendant cannot possess a smaller twist than its corresponding $W_N$ primary. Hence, in a $W_N$ CFT, the twist gap of the Virasoro primaries is at most $\tau_{\rm gap}^{W_N}$. Furthermore, if $\tau_{\rm gap}^{W_N}$ is sufficiently large, there can be Virasoro primaries with smaller nonzero twists originating from the $W_N$ vacuum sector. To determine the lowest nontrivial twist from the $W_N$ vacuum sector, let us compute several leading terms in $f(q)$. Using \eqref{WN:fq}, we find
\begin{equation}
	\begin{split}
		f(q)=1-q+q^3+O(q^4).
	\end{split}
\end{equation}
Taking into account both the left and right movers, we have
\begin{equation}
	\begin{split}
		f(q)f(\bar{q})=&(1-q)(1-\bar{q})+(1-q)\mathcal{F}_{\rm current}(\bar{q})+\mathcal{F}_{\rm current}(q)(1-\bar{q})+q^3\bar{q}^3 \\
		&+O(q^3\bar{q}^4)+O(q^4\bar{q}^3). \\
	\end{split}
\end{equation}
Therefore, we observe that the $W_N$ vacuum sector includes a unique Virasoro vacuum, chiral currents, a twist-6 Virasoro scalar primary, and other Virasoro primaries with twists $\geqslant6$. Taking into account the other candidate twist gap $\tau^{W_N}_{\rm gap}$, we conclude that the twist gap in the spectrum of Virasoro primaries is given by
\begin{equation}
	\begin{split}
		\tau_{\rm gap}=\min\left\{6,\tau_{\rm gap}^{W_N}\right\}.
	\end{split}
\end{equation}
This completes the verification of the second point.

\subsubsection{Generalizing theorem \ref{theorem:modulartauberianFJ} to $W_N$ CFT}\label{app:WNconjecture}
Before going to the last point, we would like to propose a $W_N$ analogue of theorem \ref{theorem:modulartauberianFJ}. The setup is as follows. We define the reduced partition function in the $W_N$ CFT as
\begin{equation}
	\begin{split}
		\tilde{Z}^{W_N}(q,\bar{q}):=\left[\eta(q)\eta(\bar{q})\right]^{N-1}Z(q,\bar{q}).
	\end{split}
\end{equation}
Recall that $\tilde{Z}$ counts the Virasoro primaries, here similarly, $\tilde{Z}^{W_N}$ counts the $W_N$ primaries:
\begin{equation}
	\begin{split}
		\tilde{Z}^{W_N}(q,\bar{q})=&(q\bar{q})^{-A+\frac{N-2}{24}}\left[\prod\limits_{n=1}^{N-1}\left[(1-q^n)(1-\bar{q}^n)\right]^{N-n}+\sum\limits_{h, \bar{h}\geqslant T} q^{h}\bar{q}^{\bar{h}}\right], \\
	\end{split}
\end{equation}
where $A=\frac{c-1}{24}$ and $T=2\tau_{\rm gap}^{W_N}$.

By using the modular invariance of $Z(q,\bar{q})$ and the modular transformation property of $\eta(q)$, we obtain an analogue of \eqref{modulartransformation} for the reduced partition function $\tilde{Z}^{W_N}$:
\begin{equation}\label{WN:modulartransformation}
	\begin{split}
		\tilde{Z}^{W_N} (\beta_L, \beta_R) = \left(\frac{4 \pi^2}{\beta_L \beta_R}\right)^{\frac{N-1}{2}} \tilde{Z}^{W_N}
		\left( \frac{4 \pi^2}{\beta_L}, \frac{4 \pi^2}{\beta_R} \right)\quad\left(q=e^{-\beta_{L}},\bar{q}=e^{-\beta_{R}}\right).
	\end{split}
\end{equation}
Comparing this setup with the usual CFT (as described in section \ref{section:modularsetup}), we observe the following modifications specific to $W_N$ CFTs:
\begin{itemize}
	\item The power index of the overall $q\bar{q}$ factor changes from $A$ to $A-\frac{N-2}{24}$.
	\item Here the vacuum term is given by $\prod\limits_{n=1}^{N-1}\left[(1-q^n)(1-\bar{q}^n)\right]^{N-n}$ instead of $(1-q)(1-\bar{q})$.
	\item The crossing factor in the modular invariance equation of the reduced partition function is modified to $\left(\frac{4 \pi^2}{\beta_L \beta_R}\right)^{\frac{N-1}{2}}$ instead of $\left(\frac{4 \pi^2}{\beta_L \beta_R}\right)^{\frac{1}{2}}$.
\end{itemize}
The first modification is crucial as it shifts the position of the twist accumulation point. The second and third modifications are minor, affecting only the slow-growing factors in the analysis.

We define the $W_N$ analogues of $\mathcal{N}_J$, $\mathcal{A}_J$, and the double lightcone limit as follows:
\begin{equation}\label{redef:WN}
	\begin{split}
		\mathcal{N}_J^{W_N}(\varepsilon) :=& \sum\limits_{h\in\left(A-\frac{N-2}{24}-\varepsilon, A-\frac{N-2}{24}+\varepsilon\right)} n_{h,h+J}, \\
		\mathcal{A}_J^{W_N}(\varepsilon,\beta_{L}):=&\sum\limits_{h\in\left(A-\frac{N-2}{24}-\varepsilon, A-\frac{N-2}{24}+\varepsilon\right)} n_{h,h+J}\,e^{-\left(h-A+\frac{N-1}{24}\right)\beta_{L}}, \\
			W_N-\mathrm{DLC}_w\ \mathrm{limit}:&\quad\beta_L,\ J\rightarrow\infty,\quad \frac{2\pi T(1-w^2)}{A-\frac{N-2}{24}}\sqrt{\frac{J}{A-\frac{N-2}{24}}}-\beta_L\rightarrow\infty\,,\\
			&\quad \beta_L^{-1}\log{J}\to 0\,. \\
	\end{split}
\end{equation}
It is worth noting that the double lightcone limit here is weaker than the one defined in (\ref{redef:DLCwJ}), with the identification $a=\sqrt{\frac{N-2}{24}}$. Therefore, when we apply (\ref{redef:DLCwJ}) to check conjecture \ref{conjecture:infinitespin:shift} in the $W_N$ CFT, the conditions of the $W_N$ DLC$_w$ limit defined here will always be satisfied.

Now we are ready to propose the following $W_N$ analogue of theorem \ref{theorem:modulartauberianFJ}:
\begin{conjecture}\label{conjecture:modulartauberianFJ:WN}
	For any $w\in\left(\frac{1}{2},1\right)$, the quantity $\mathcal{A}_J^{W_N}(\varepsilon,\beta_{L})$ satisfies asymptotic two-sided bounds in the $W_N$-DLC$_w$ limit, given by:
	\begin{equation}
		\begin{split}\label{eq:resultFiniteEpsilon:WN}
			\frac{1}{w}\frac{1}{1-\frac{\tan\left(\pi w(1-\varepsilon)\right)}{\pi w(1-\varepsilon)}}\lesssim \frac{\mathcal{A}_J^{W_N}(\beta_L,\varepsilon )}{S_N(\beta_{L},J)e^{4\pi\sqrt{\left(A-\frac{N-2}{24}\right)J}}} \lesssim \frac{1}{w}\frac{2}{1+\frac{\sin(2\pi w\varepsilon)}{2\pi w\varepsilon}},
		\end{split}
	\end{equation}
    where $S_N(\beta_{L},J)$ is a slow-growing factor of the form 
    \begin{equation}
    	\begin{split}
    		S_N(\beta_{L},J)=C_N\beta_{L}^{\mu_N}J^{\nu_N}
    	\end{split}
    \end{equation}
    with finite coefficients $C_N$, $\mu_N$, and $\nu_N$ determined by $N$ of the $W_N$ algebra.
    
    Furthermore, the parameter $\varepsilon$ belongs to the interval
	\begin{equation}\label{epsilon:choice:DLCJ:WN}
		\begin{split}
			\varepsilon_{\rm min}(\beta_{L},J)\leqslant&\varepsilon\leqslant1-\frac{1}{2w}, \\
		\end{split}
	\end{equation}
	where $\varepsilon_{\rm min}(\beta_{L},J)$ is defined as:
	\begin{equation}\label{def:epsilonmin:WN}
		\begin{split}
			\varepsilon_{\rm min}(\beta_{L},J):=&\max\left\{P_N(A,T,w)\frac{\log{J}}{\sqrt{J}},\ 
			Q_N(A,T,w)\frac{\log{J}}{\beta_{L}}+R_N(A,T,w)\frac{\log\beta_{L}}{\beta_{L}}\right\},
		\end{split}
	\end{equation}
with finite coefficients $P_N$, $Q_N$, and $R_N$.
\end{conjecture}
We believe that conjecture \ref{conjecture:modulartauberianFJ:WN} can be proven using the same argument as presented in section \ref{section:modularbootstrap}.

A direct consequence of conjecture \ref{conjecture:modulartauberianFJ:WN} is the estimate of $\mathcal{N}_J^{W_N}(\varepsilon)$, which represents the number of $W_N$ primaries near the twist accumulation point. This estimate serves as an analogue of corollary \ref{cor:operatorcount}. According to the conjecture, we have:
\begin{equation}\label{operatorcount:WN}
	\begin{split}
		\mathcal{N}_J^{W_N}(\varepsilon\equiv\kappa J^{-1/2}\log J)=e^{4\pi\sqrt{\left(A-\frac{N-2}{2}\right)J}+O(\log J)},
	\end{split}
\end{equation}
where $\kappa$ is a fixed positive constant. We note that $\kappa$ has a positive lower bound, similar to what we observed in corollary \ref{cor:operatorcount}. This estimate demonstrates the exponential growth of $\mathcal{N}_J^{W_N}(\varepsilon)$ with respect to $J$, with the leading term determined by $A$ and the $W_N$ algebra. The subleading term contributes a factor of powerlaw growth.

\subsubsection{Testing conjecture \ref{conjecture:infinitespin:shift}}
We have no doubt that the first part of conjecture \ref{conjecture:infinitespin:shift} can be checked rigorously in the case of $W_N$ CFT. Now, let us discuss the exponential growth of $\mathcal{A}_J$ in the second part of the conjecture. Specifically, we aim to demonstrate that both $\mathcal{A}_J$ and $\mathcal{N}_J$ exhibit the expected exponential growth as stated in conjecture \ref{conjecture:infinitespin:shift} and eq.\,\eqref{operatorcounting:shift}:
\begin{equation}\label{WN:expcheck}
	\begin{split}
		\mathcal{A}_J(\beta_{L},\varepsilon)\stackrel{{\rm DLC}_w}{\sim}e^{4\pi\sqrt{AJ}},\quad \mathcal{N}_J(\varepsilon\equiv\kappa J^{-1/2}\log J)\sim e^{4\pi\sqrt{AJ}},
	\end{split}
\end{equation}
with the understanding that there may be additional factors of slow growth.

Analogous to the first example, where we were unable to provide a rigorous verification, we would like to propose the following hypotheses for the $W_N$ CFT:
\begin{itemize}
	\item We assume that conjecture \ref{conjecture:modulartauberianFJ:WN} is correct.
	\item In the DLC$_w$ limit, we expect the dominant contribution to $\log\mathcal{A}_J(\beta_{L},\varepsilon)$ to come from two sources: (a) $W_N$ primaries near the twist accumulation point, and (b) Virasoro primaries originating from the $W_N$ descendants of these $W_N$ primaries.
\end{itemize}
According to the second hypothesis, in the DLC$_w$ limit, we can approximate $\mathcal{A}_J(\beta_{L},\varepsilon)$ as a sum over contributions from $W_N$ primaries near the twist accumulation point and their $W_N$ descendants which are Virasoro primaries. Specifically, we have the approximation:
\begin{equation}\label{WN:AJapprox}
	\begin{split}
		\mathcal{A}_J(\beta_{L},\varepsilon)\approx\sum\limits_{n=0}^{J} \mathcal{A}_{J-n}^{W_N}(\beta_{L},\varepsilon)\times B_n.
	\end{split}
\end{equation}
Here, $J-n$ represents the spin of the $W_N$ primary near the twist accumulation point, and $B_n$ corresponds to the number of independent $W_N$ descendants. These descendants are Virasoro primaries with the same twist as the $W_N$ primary and an additional spin of $n$ (i.e.\,the total spin of the Virasoro primary is $J$).

The coefficients $B_n$ can be determined using the second case of the $W_N$ character formula \eqref{def:WNcharacter}. Similar to our analysis for the $W_N$ vacuum character, we have the expression:
\begin{equation}
	\begin{split}
		\left(\frac{q^{\frac{1}{24}}}{\eta(q)}\right)^{\frac{N-2}{24}}=\sum\limits_{n=0}^{\infty}B_nq^n.
	\end{split}
\end{equation}
By the same analysis as the $W_N$ vacuum character, we get the asymptotic growth of $B_n$ when $n$ is very large:
\begin{equation}\label{WN:Dnasymp}
	\begin{split}
		B_n=b_ne^{4\pi\sqrt{\frac{N-2}{24}n}},
	\end{split}
\end{equation}
where $b_n$ is a factor of slow growth in $n$.

By considering the first hypothesis, namely conjecture \ref{conjecture:modulartauberianFJ:WN}, we can derive the exponential factor that governs the growth of $\mathcal{A}_{J-n}^{W_N}$:
\begin{equation}\label{WN:AJnasymp}
	\begin{split}
		\mathcal{A}_{J-n}^{W_N}(\beta_L,\varepsilon)\approx e^{4\pi\sqrt{\left(A-\frac{N-2}{24}\right)(J-n)}}.
	\end{split}
\end{equation}
According to conjecture \ref{conjecture:modulartauberianFJ:WN}, the aforementioned asymptotic behavior remains valid when the growth rates of $J-n$ and $\beta_{L}$, instead of $J$ and $\beta_{L}$, satisfy the condition of the $W_N$-DLC$_w$ limit. We will address this subtlety in further detail later.

Substituting \eqref{WN:Dnasymp} and \eqref{WN:AJnasymp} into \eqref{WN:AJapprox}, we obtain
\begin{equation}\label{WN:AJapprox2}
	\begin{split}
		\mathcal{A}_J(\beta_{L},\varepsilon)\stackrel{{\rm DLC}_w}{\approx}\sum\limits_{n=0}^{J} e^{4\pi\sqrt{\left(A-\frac{N-2}{24}\right)(J-n)}}e^{4\pi\sqrt{\frac{N-2}{2}n}}.
	\end{split}
\end{equation}
Here, we neglect the slow-growing factor. In the limit $J\rightarrow\infty$, the leading exponential growth of the sum is determined by maximizing
$$\sqrt{\left(A-\frac{N-2}{24}\right)(J-n)}+\sqrt{\frac{N-2}{24}n},$$
which occurs when
\begin{equation}\label{WN:maxpoint}
	\begin{split}
		\frac{A-\frac{N-2}{24}}{J-n}=\frac{N-2}{24n}\quad\Longleftrightarrow\quad J-n=\frac{A-\frac{N-2}{24}}{A}J,\ n=\frac{N-2}{24A}J.
	\end{split}
\end{equation}
Now we come back to the subtlety mentioned after eq.\,\eqref{WN:AJnasymp}. In the framework of $W_N$ CFT, the right inverse temperature in the analysis of $\mathcal{A}_{J-n}^{W_N}(\beta_{L},\varepsilon)$ is identified with spin $J$ by
\begin{equation}
	\begin{split}
		\beta_{R}=2\pi\sqrt{\frac{A-\frac{N-2}{24}}{J-n}}.
	\end{split}
\end{equation}
Here the numerator in the square-root is $A-\frac{N-2}{24}$ instead of $A$ because in the setup of $W_N$ reduced partition function $\tilde{Z}^{W_N}$, everything depending on $A$ is modified to $A-\frac{N-2}{24}$, so is the identification between $\beta_{R}$ and $J$:
\begin{equation}
	\begin{split}
		{\rm CFT}:\ \beta_{R}=2\pi\sqrt{\frac{A}{J}}\quad{\rm vs}\quad W_N\ {\rm CFT}:\ \beta_{R}=2\pi\sqrt{\frac{A-\frac{N-2}{24}}{J}}.
	\end{split}
\end{equation}
But now, because of condition \eqref{WN:maxpoint}, we have
\begin{equation}
	\begin{split}
		\beta_{R}=2\pi\sqrt{\frac{A-\frac{N-2}{24}}{J-n}}=2\pi\sqrt{\frac{A}{J}}.
	\end{split}
\end{equation}
This result is of utmost importance in our examination of conjecture \ref{conjecture:infinitespin:shift} as it indicates that the right inverse temperature we are considering corresponds precisely to the one in the standard CFT! One can also check that by condition \eqref{WN:maxpoint}, the $W_N$-DLC$_w$ condition \eqref{redef:WN}, with $J$ replaced by $J-n$, is compatible with the DLC$_w$ condition \eqref{redef:DLCwJ} for conjecture \ref{conjecture:infinitespin:shift}. This resolves the previously mentioned subtlety and serves as a consistency check on our approximation.

By condition \eqref{WN:maxpoint}, we have
\begin{equation}
	\begin{split}
		\max\limits_{n\leqslant J}\left\{\sqrt{\left(A-\frac{N-2}{24}\right)(J-n)}+\sqrt{\frac{N-2}{24}n}\right\}\approx \sqrt{AJ}.
	\end{split}
\end{equation}
Here we use ``$\approx$" instead of ``$=$" because the maximal value is generally not attained at integer values of $J$ and $n$. However, this approximation does not affect the leading exponential growth. Consequently, we obtain
\begin{equation}
	\begin{split}
		\mathcal{A}_J(\beta_{L},\varepsilon)\stackrel{{\rm DLC}_w}{\approx}e^{4\pi\sqrt{AJ}}.
	\end{split}
\end{equation}
We expect that a more careful analysis will lead to the precise statement:
\begin{equation}
	\begin{split}
		\lim\limits_{{\rm DLC}_w}\frac{\log\mathcal{A}_J(\beta_{L},\varepsilon)}{4\pi\sqrt{AJ}}=1.
	\end{split}
\end{equation}
This completes the verification for the first equation of \eqref{WN:expcheck}. The same reasoning applies to the second equation of \eqref{WN:expcheck}.

Thus, we have checked the last point, which is consistent with the second part of conjecture \ref{conjecture:infinitespin:shift}.

\subsection{Three copies of Ising CFTs}\label{app:Ising}
The third example involves three copies of Ising CFTs:
$${\rm CFT}={\rm Ising}^{(1)}\otimes{\rm Ising}^{(2)}\otimes{\rm Ising}^{(3)}.$$
This CFT has the central charge $c=\frac{3}{2}$, which corresponds to $A=\frac{1}{48}$, and a nonzero twist gap of Virasoro primaries. Therefore, based on the aforementioned argument, there exists a accumulation point for twists with $h\leqslant\frac{1}{48}$. In a single Ising CFT, there are only three primary states: $I$ ($h_I=\bar{h}_I=0$), $\sigma$ ($h_\sigma=\bar{h}_\sigma=\frac{1}{16}$), and $\epsilon$ ($h_\epsilon=\bar{h}_\epsilon=\frac{1}{2}$). Consequently, in the big CFT, the only possible states with $h$ not exceeding $\frac{1}{48}$ are given by
\begin{equation}
	\begin{split}
		\prod\limits_{i=1,2,3}\prod\limits_{n\geqslant2}(\bar{L}^{(i)}_{-n})^{\alpha(i,n)}\ket{{\rm vac}}^{(1)}\otimes\ket{{\rm vac}}^{(2)}\otimes\ket{{\rm vac}}^{(3)},
	\end{split}
\end{equation}
where $i$ corresponds to the label in Ising$^{(i)}$, and $\alpha(i,n)$ are non-negative integers. By taking linear combinations of these states, we can obtain the right-moving currents of the Virasoro algebra, which have $h=0$.

Therefore, even without performing any computations, we can deduce that the twist accumulation point must be shifted by
\begin{equation}
	\begin{split}
		a^2=\frac{1}{48},
	\end{split}
\end{equation}
and the corresponding Virasoro primaries, which are the right currents, are located precisely at $h=0$. Now, let's perform a consistency check on the growth of $D(j)$, the number of spin-$j$ currents. Similar to the previous examples, the chiral currents are contained in the product of vacuum characters of the Ising CFTs:
\begin{equation}
	\begin{split}
		\chi^{\rm Ising}_{\rm vac}(q)^3=&\frac{q^{-\frac{1}{48}}}{\eta(q)}f(q),\\
	\end{split}
\end{equation}
where $f(q)$ is related to $\mathcal{F}_{\rm current}(q)$ through \eqref{fFrelation}, and the vacuum character of the Ising CFT is given by
\begin{equation}
	\begin{split}
		\chi^{\rm Ising}_{\rm vac}(q)=\frac{1}{2\sqrt{\eta(q)}}\left(\sqrt{\theta_3(q)}+\sqrt{\theta_4(q)}\right),
	\end{split}
\end{equation}
where $\theta_i(q)$ represents Jacobi's theta functions. Combining the above two equations, we find that
\begin{equation}
	\begin{split}
		f(q)=\frac{q^{\frac{1}{48}}}{8\sqrt{\eta(q)}}\left(\sqrt{\theta_3(q)}+\sqrt{\theta_4(q)}\right)^3.
	\end{split}
\end{equation}
Using the properties of $\eta(q)$ and $\theta_i(q)$ under S modular transformation, we determine the asymptotic behavior of $f(q)$ as $q\rightarrow 1$ (or equivalently, $\beta\rightarrow 0$):
\begin{equation}
	\begin{split}
		f(q)=\frac{1}{8}\sqrt{\frac{2\pi}{\beta}}\frac{q^{\frac{1}{48}}}{\sqrt{\eta(q')}}\left(\sqrt{\theta_3(q')}+\sqrt{\theta_2(q')}\right)^3\stackrel{\beta\rightarrow0}{\sim}\frac{1}{8}\sqrt{\frac{2\pi}{\beta}}e^{\frac{4\pi^2}{\beta}\frac{1}{48}}\quad\left(q'=e^{-\frac{4\pi^2}{\beta}}\right).
	\end{split}
\end{equation}
Then by \eqref{Fcurrent:asymp:withshift} and \eqref{fFrelation}, we obtain $a=\sqrt{\frac{1}{48}}$.

In this case, there is no need to look for the candidate twist gap from the product of Ising vacuum characters. This is because the Ising spin field $\sigma$, which has twist $\tau_\sigma=\frac{1}{8}$, already gives the twist gap:
\begin{equation}
	\begin{split}
		\tau_{\rm gap}=\frac{1}{8}.
	\end{split}
\end{equation}
Any Virasoro primary from the product of Ising vacuum characters can only have zero twist or twist $\geqslant1$.

\section{Proof of theorem \ref{theorem:largec}}\label{section:estimatelargec}
In this appendix, we present a proof of theorem \ref{theorem:largec}.  Similar to section \ref{section:modularbootstrap}, to derive bounds on $\mathcal{A}_J(\beta_L,\varepsilon_1,\varepsilon_2,A)$ it is convenient to introduce the quantity
\begin{equation}
	\begin{split}
		\mathcal{A}(\beta_{L},\bar{H},\varepsilon_1,\varepsilon_2,\delta,A):=&\int_{A-\varepsilon_1}^{A+\varepsilon_2}d h\int_{\bar{H}-\delta}^{\bar{H}+\delta} d\bar{h}\,\rho(h,\bar{h})e^{-(h-A)\beta_L}, \\
	\end{split}
\end{equation}
which is the analogue of $\mathcal{A}(\beta_L,\bar{H},\varepsilon,\delta)$ defined in \eqref{def:Acal} (recall that $\bar{H}\equiv A+J$). Recall that we always use the identification $\beta_{R}=2\pi\sqrt{\frac{A}{J}}$, the HDLC$_w$ limit \eqref{def:HDLC} is equivalent to
\begin{equation}\label{def:HDLCequiv}
	\begin{split}
		&\beta_L/A\to\infty\,,\quad\beta_R\to 0\,, \quad A\to\infty,\quad \frac{4\pi^2 \alpha(1-w^2)}{\beta_R}-\beta_L\rightarrow\infty, \\
		&\beta_L^{-1}\log\left(\beta_R\right)\rightarrow0\,. \\
	\end{split}
\end{equation}
Our approach here is similar to the one explained in section \ref{section:modular:idea}. We aim to show that, under the additional assumptions stated at the beginning of the section, the results similar to the ones in section \ref{section:modularbootstrap} holds. Specifically, we want to prove
\begin{equation}\label{modular:sketch1c}
	\begin{split}
		\lim\limits_{{\mathrm{HDLC}_{w}}}\frac{I^{\rm dual}_{\pm,\rm nonvac}}{I^{\rm dual}_{\pm,\rm vac}}=0.
	\end{split}
\end{equation}
in the dual channel and
\begin{equation}\label{modular:sketch2c}
	\begin{split}
		\lim\limits_{\mathrm{HDLC}_{w}}\frac{I_{\pm,\rm vac}}{I^{\rm dual}_{\pm,\rm vac}}=0\,,\quad 
		\lim\limits_{\mathrm{HDLC}_{w}}\frac{I_{\pm,T\leqslant h\leqslant A-\varepsilon_1}}{I^{\rm dual}_{\pm,\rm vac}}=0\,,\quad
		\lim\limits_{\mathrm{HDLC}_{w}}\frac{I_{\pm,h\geqslant A+\varepsilon_2}}{I^{\rm dual}_{\pm,\rm vac}}=0\,
	\end{split}
\end{equation}
in the direct channel. Here these $I_{\pm}$'s are defined in a similar way to \eqref{def:allIpm}.

\subsection{Dual channel: vacuum}
We reconsider the integral (\ref{int:dualvac}), the vacuum term in the dual channel of $I_{\pm}$, keeping in mind that $A$ is no longer fixed. We claim that a similar asymptotic behavior holds:
\begin{equation}\label{Ipmvac:asymp:largec}
	\begin{split}
		I^{\rm dual}_{\pm,\rm vac}\stackrel{{\rm HDLC}_{w}}{\sim} \frac{4\pi^{5/2}\beta_R}{\beta_L^{3/2}A^{1/2}}e^{A\left(\frac{4\pi^2}{\beta_L}+\frac{4\pi^2}{\beta_R}\right)}\hat{\phi}_\pm(0).  
	\end{split}
\end{equation}
Note that (\ref{Ipmvac:asymp:largec}) has an extra factor $e^{\frac{4\pi^2A}{\beta_L}}$ compared to (\ref{Ipmvac:asymp}). In fact, we also have this factor for fixed $A$, but we ignored it in (\ref{Ipmvac:asymp}) because it is asymptotically equal to 1 in the double lightcone limit when $A$ is fixed. This asymptotic behavior holds for any limit $\beta_{L}\rightarrow\infty,\ \beta_{R}\rightarrow0$ (keeping in mind that $\beta_{R}\equiv2\pi\sqrt{\frac{A}{J}}$) without any constraints between $A$, $\beta_{L}$ and $\beta_{R}$. In the HDLC$_w$ limit, we have the extra condition $A/\beta_{L}\rightarrow0$, so \eqref{Ipmvac:asymp:largec} reduces to \eqref{Ipmvac:asymp}, the one in a fixed CFT.

To see the asymptotic behavior (\ref{Ipmvac:asymp:largec}), we just need to repeat the estimates in section \ref{section:modulardual} more carefully, taking into account that the quantities dependent on $A$ are no longer constants. The technical details are given in appendix \ref{app:Ivac:asym}.

\subsection{Other terms: some technical preparations}
Our aim is to show that, as we approach the double lightcone limit, the nonvacuum contribution in the dual channel as well as the vacuum/low-twist/high-twist contributions in the direct channel are subleading compared to the vacuum term in the dual channel. This was illustrated in eqs.\,\eqref{modular:sketch1c} and \eqref{modular:sketch2c}.

Our analysis builds upon the arguments presented in sections \ref{section:modulardual} and \ref{section:modulardirect}, but to handle the limit $A\rightarrow\infty$ appropriately, we would like to introduce a technical lemma and a remark (see below) that facilitate the generalization of the previous arguments.

The following lemma will be convenient for us to eliminate the unimportant powerlw factors of $\beta_{L}$, $\beta_{R}$ and $A$ in the analysis:
\begin{lemma}\label{lemmaLargec}
	Given any fixed $u>0$ and $r,s,t\in\mathbb{R}$, we have
	\begin{equation}
		\begin{split}
			\beta_{L}^r\beta_{R}^sA^te^{-\frac{u}{\beta_{R}}}\stackrel{{\rm HDLC}_{w}}{\longrightarrow}0,\\ 
		\end{split}
	\end{equation}
	where the HDLC$_w$ limit was defined in (\ref{def:HDLC}).
\end{lemma}
\begin{proof}
	Since $\beta_{L}\rightarrow\infty$, $\beta_{R}\rightarrow0$ and $A\rightarrow\infty$ in the HDLC$_w$ limit, it suffices to prove the case of $r>0$, $s<0$ and $t>0$. 
	
	When we approach the HDLC$_w$ limit, eventually we will have
	\begin{equation}
		\begin{split}
			\beta_{L}\leqslant\frac{4\pi^2\alpha(1-w^2)}{\beta_{R}},\quad A\leqslant \frac{M}{\beta_{R}},
		\end{split}
	\end{equation}
	for some fixed positive $M$. In the above regime we have
	\begin{equation}
		\begin{split}
			\beta_{L}^r\beta_{R}^sA^te^{-\frac{u}{\beta_{R}}}\leqslant\left[4\pi^2\alpha(1-w^2)\right]^rM^t\beta_{R}^{-r+s-t}e^{-\frac{u}{\beta_{R}}}\rightarrow0.
		\end{split}
	\end{equation}
	In the last step we used $\beta_{R}\rightarrow0$ in the HDLC$_w$ limit.
\end{proof}

We would also like to make the following remark:
\begin{remark}\label{remark:Ztildebound:largec}
	Lemma \ref{lemma:Ztildebound} assumes only $\beta_0>0$. However, according to assumption 2 at the beginning of section \ref{section:largec}, lemma \ref{lemma:Ztildebound} still holds in the current context, given that $\beta_0>2\pi$. 
\end{remark}

\subsection{Dual channel: non-vacuum}\label{sec:largeCdualnvac}
Recall from section \ref{section:modular:directnonvac} that in the dual channel of $I_{\pm}$, we separated the non-vacuum contribution $I^{\rm dual}_{\pm,\rm nonvac}$ into $I^{\rm dual}_{\pm,T\leqslant\bar{h}\leqslant A}$ and $I^{\rm dual}_{\pm,\bar{h}\geqslant A}$ and found that in the double lightcone limit, they were aymptotically bounded by \eqref{eq:dualLessA} and \eqref{eq:dualGreatA}. Now, we aim to generalize \eqref{eq:dualLessA} and \eqref{eq:dualGreatA} to the current context, where we need to carefully examine the changes that arise due to the theory no longer being fixed.

Let us first consider  $I^{\rm dual}_{\pm,T\leqslant\bar{h}\leqslant A}$. The first step of the estimate can be found in \eqref{Ztilde:dualLessA:derivation}. According to remark \ref{remark:Ztildebound:largec}, the inequalities in \eqref{Ztilde:dualLessA:derivation} remain valid if we choose $\beta_0>2\pi$. Let us now look at \eqref{eq:dualLessA}, the next step of the estimate. The difference between our current situation and the one in \eqref{eq:dualLessA} is that $I^{\rm dual}_{\pm,\rm vac}$ now has an extra factor of $e^{\frac{4\pi^2A}{\beta_{L}}}$ due to $A$ no longer being fixed, as shown in eq. \eqref{Ipmvac:asymp:largec}. Therefore, estimate \eqref{eq:dualLessA} is modified to:
\begin{equation}\label{eq:dualLessA:largec}
	\begin{split}
		\abs{\frac{I^{\rm dual}_{\pm,T\leqslant\bar{h}<A}}{I^{\rm dual}_{\pm,{\rm vac}}}}\stackrel{{\rm HDLC}_w}{\lesssim}&\frac{\Lambda\beta_0^{1/2}\kappa(\beta_0)\max\limits_{x}\abs{\hat{\phi}_\pm(x)}}{2\pi^{5/2}\hat{\phi}_\pm(0)}\left(\frac{\beta_{L}}{\beta_{R}}\right)^{3/2}A^{1/2}e^{-\frac{4\pi^2 T}{\beta_{R}}-\frac{4\pi^2A}{\beta_{L}}+A\beta_{L}+A\beta_0-\frac{4\pi^2(A-T)}{\beta_0}}. \\
	\end{split}
\end{equation}
The factors on the r.h.s. of the above equation can be expressed in terms of $\beta_{L}$, $\beta_{R}$, and $A$ as follows:
\begin{equation}\label{dualLessA:largec:exp}
	\begin{split}
		\left[\left(\frac{\beta_{L}}{\beta_{R}}\right)^{3/2}A^{1/2}e^{-\frac{4\pi^2\alpha w^2 A}{\beta_{R}}}\right]\times e^{-A\left[\frac{4\pi^2\alpha(1-w^2)}{\beta_{R}}-\beta_{L}+\left(\frac{T}{A}-\alpha\right)\frac{4\pi^2}{\beta_{R}}+\frac{4\pi^2}{\beta_{L}}+\frac{4\pi^2(A-T)}{A\beta_0}-\beta_0\right]}.
	\end{split}
\end{equation}
In the HDLC$_w$ limit, the first factor vanishes, as shown by Lemma \ref{lemmaLargec}, while the second factor decays exponentially to 0 due to (\ref{assum1}), (\ref{def:HDLC}), and the fact that $T\leqslant A$. Consequently, we obtain:
\begin{equation}\label{Result:dualLessA:largec}
	\begin{split}
		\abs{\frac{I^{\rm dual}_{\pm,T\leqslant\bar{h}<A}}{I^{\rm dual}_{\pm,{\rm vac}}}}\stackrel{{\rm HDLC}_w}{\longrightarrow}0.
	\end{split}
\end{equation}
Let us now consider $I^{\rm dual}_{\pm,\bar{h}\geqslant A}$ and perform a similar analysis as we did for $I^{\rm dual}_{\pm,T\leqslant\bar{h}<A}$. The inequalities in (\ref{Ztilde:dualGreatA:derivation}) remain valid for $\beta_0>2\pi$. Additionally, we have an extra factor $e^{\frac{4\pi^2A}{\beta_{L}}}$ in the asymptotic behavior of $I^{\rm dual}_{\pm,\rm vac}$. Thus, the estimate in \eqref{eq:dualGreatA} is modified to:
\begin{equation}\label{eq:dualGreatA:largec}
	\begin{split}
		\abs{\frac{I^{\rm dual}_{\pm,\bar{h}\geqslant A}}{I^{\rm dual}_{\pm,{\rm vac}}}}\stackrel{{\rm HDLC}_w}{\lesssim}&\frac{\Lambda\kappa(\beta_0)\max\limits_{x}\abs{\hat{\phi}_\pm(x)}}{\sqrt{2}\pi^{5/2}\hat{\phi}_\pm(0)}\frac{\beta_{L}^{3/2}A^{1/2}}{\beta_{R}^2}e^{-A\left[\frac{4\pi^2-\Lambda^2}{\beta_{R}}-\beta_{L}-\beta_{R}+\frac{4\pi^2}{\beta_{L}}\right]}. \\
	\end{split}
\end{equation}
We can express the factors that depend on $\beta_{L}$, $\beta_{R}$ and $A$ in a similar manner as in \eqref{dualLessA:largec:exp}:
\begin{equation}
	\begin{split}
		\left[\frac{\beta_{L}^{3/2}A^{1/2}}{\beta_{R}^2}e^{-A\frac{4\pi^2w^2-\Lambda^2}{\beta_{R}}}\right]\times e^{-A\left[\frac{4\pi^2(1-w^2)}{\beta_{R}}-\beta_{L}-\beta_{R}+\frac{4\pi^2}{\beta_{L}}\right]}.
	\end{split}
\end{equation}
Both factors decay to 0 in the HDLC$_w$ limit. Therefore, we obtain
\begin{equation}\label{Result:dualGreatA:largec}
	\begin{split}
		\abs{\frac{I^{\rm dual}_{\pm,\bar{h}\geqslant A}}{I^{\rm dual}_{\pm,{\rm vac}}}}\stackrel{{\rm HDLC}_w}{\longrightarrow}0.
	\end{split}
\end{equation}
Eqs.\,\eqref{Result:dualLessA:largec} and \eqref{Result:dualGreatA:largec} show that in the dual channel, the non-vacuum contribution is suppressed by the vacuum contribution in the HDLC$_w$. This finishes the proof of \eqref{modular:sketch1c}.

\subsection{Direct channel: vacuum}
Let us consider $I_{\pm,\rm vac}$, the vacuum term in the direct channel of $I_{\pm}$, with the goal of showing the first equation of \eqref{modular:sketch2c}. In section \ref{modular:crossvac}, we analyzed $I_{\pm,\rm vac}$ for fixed CFT and derived the upper bound \eqref{Ipmdrect:upperbound}, which still holds. As in the previous subsection, we need to modify eq.\,\eqref{eq:vac} to account for the large $A$ dependence of $I^{\rm dual}_{\pm,\rm vac}$. Here we obtain the following modified estimate:
\begin{equation}\label{eq:vac:largec}
	\begin{split}
		\abs{\frac{I_{\pm,\rm vac}}{I^{\rm dual}_{\pm,\rm vac}}}\stackrel{\mathrm{HDLC}_w}{\lesssim}& \frac{\max\limits_{x}\abs{\phi_\pm(x)}}{2\pi^{5/2}\hat{\phi}_\pm(0)}\frac{\beta_L^{3/2}A^{1/2}}{\beta_R} e^{-A\left(\frac{4\pi^2}{\beta_R}+\frac{4\pi^2}{\beta_{L}}-\beta_L-\beta_R\right)}. \\
	\end{split}
\end{equation}
As in the previous subsection, we express the factors that depend on $\beta_{L}$, $\beta_{R}$ and $A$ as 
\begin{equation}
	\begin{split}
		\left[\frac{\beta_L^{3/2}A^{1/2}}{\beta_R} e^{-\frac{4\pi^2w^2A}{\beta_{R}}}\right]\times e^{-A\left(\frac{4\pi^2(1-w^2)}{\beta_R}+\frac{4\pi^2}{\beta_{L}}-\beta_L-\beta_R\right)}.
	\end{split}
\end{equation}
Both factors decay to 0 in the HDLC$_w$ limit by the same reason as in the previous subsection. So we get
\begin{equation}\label{Result:directvac:largec}
	\begin{split}
		\abs{\frac{I_{\pm,\rm vac}}{I^{\rm dual}_{\pm,{\rm vac}}}}\stackrel{{\rm HDLC}_w}{\longrightarrow}0.
	\end{split}
\end{equation}

\subsection{Direct channel: low twist and high twist}\label{sec:htc}
Let us consider the non-vacuum terms in the direct channel of $I_{\pm}$: the low-twist ($T\leqslant h\leqslant A-\varepsilon_1$) and high-twist ($h\geqslant A+\varepsilon_2$) terms, with the goal of showing the second and third equations of \eqref{modular:sketch2c}.

Before proceeding with the analysis, we want to emphasize that our estimates on $I^{\rm dual}_{\pm,\rm vac}$, $I^{\rm dual}_{\pm,\rm nonvac}$, and $I_{\pm,\rm vac}$ are based on a weaker version of the double lightcone limit. Specifically, the constraints we need for these estimates are as follows:\footnote{It is worth noting that under these weaker conditions, lemma \ref{lemmaLargec} still holds.}
\begin{equation}\label{HDLC:weaker}
	\begin{split}
		\beta_{L}\rightarrow\infty,\quad\beta_{R}\rightarrow0,\quad \frac{4\pi^2 \alpha(1-w^2)}{\beta_R}-\beta_L\rightarrow\infty,\quad  A_{\rm min}\leqslant A\leqslant \frac{M}{\beta_{R}},
	\end{split}
\end{equation}
where $A_{\rm min}$ and $M$ are some arbitrary fixed positive constant.\footnote{The last condition in \eqref{HDLC:weaker} is fulfilled in the HDLC$_w$ limit.} The additional constraints in the definition of the HDLC$_w$ limit (see eq.\,\eqref{def:HDLC}) are only required for the estimates on $I_{\pm,T\leqslant h\leqslant A-\varepsilon_1}$ and $I_{\pm,h\geqslant A+\varepsilon_2}$.

Here we only consider the case of fixed $\varepsilon_1$ and $\varepsilon_2$, and we will discuss the case of $\varepsilon_1,\varepsilon_2\rightarrow0$ later in section \ref{section:epsilonwindow:largec}.

For the high-twist term $I_{\pm,h\geqslant A+\varepsilon_2}$, we reconsider the estimates from \eqref{Zhigh:bound} to \eqref{eq:highprime}. The subtleties here are again $\beta_0>2\pi$ and the extra factor $e^{\frac{4\pi^2A}{\beta_{R}}}$ of $I^{\rm dual}_{\pm,\rm vac}$. We have a modified version of (\ref{eq:high}):
\begin{equation}\label{eq:high:largec}
	\begin{split}
		\abs{\frac{I_{\pm,h\geqslant A+\varepsilon_2}}{I^{\rm dual}_{\pm,\rm vac}}}\stackrel{{\rm HDLC}_w}{\lesssim}&\frac{\kappa(\beta_0)\beta_0^{1/2}\max\limits_{x}\abs{\phi_{\pm}(x)}}{4\pi^{5/2}\hat{\phi}_\pm(0)}\left(\frac{\beta_{L}}{\beta_{R}}\right)^{3/2}A^{1/2} e^{-\varepsilon_2\left(\beta_L-\frac{4\pi^2}{\beta_0}\right)+A\beta_0-\frac{4\pi^2A}{\beta_{L}}}.\\
	\end{split}
\end{equation}
From \eqref{eq:high:largec} we notice two issues:
\begin{itemize}
	\item There is a factor of $\beta_{R}^{-3/2}$, while the exponential factor does not depend on $\beta_{R}$. This issue has already appeared in the case of fixed CFT.
	\item There is an extra factor $e^{A\beta_0}$ which blows up in the limit $A\rightarrow\infty$.
\end{itemize}
The first issue was already resolved by introducing the condition $\log(\beta_{R})/\beta_{L}\rightarrow0$. For the second issue, the only way to resolve it is to let $\beta_{L}$ go to $\infty$ much faster than $A$. This is where we use the condition $\beta_{L}/A\rightarrow\infty$. We express the factors that depend on $\beta_{L}$, $\beta_{R}$ and $A$ as follows:
\begin{equation}
	\begin{split}
		\left[\left(\frac{\beta_{L}}{\beta_{R}}\right)^{3/2}A^{1/2} e^{-\frac{\varepsilon_2\beta_{L}}{2}}\right]\times e^{-\varepsilon_2\left(\frac{\beta_L}{2}-\frac{4\pi^2}{\beta_0}\right)+A\beta_0-\frac{4\pi^2A}{\beta_{L}}}
	\end{split}
\end{equation}
It is not hard to see that for fixed $\varepsilon_2$, both factors vanishes in the HDLC$_w$ limit. Therefore, we get
\begin{equation}\label{Result:directhigh:largec}
	\begin{split}
		\abs{\frac{I_{\pm,h\geqslant A+\varepsilon_2}}{I^{\rm dual}_{\pm,\rm vac}}}\stackrel{{\rm HDLC}_w}{\longrightarrow}&0\quad(\varepsilon_2\ {\rm fixed}).\\
	\end{split}
\end{equation}

For the low-twist term $I_{\pm,T\leqslant h\leqslant A-\varepsilon_1}$, we reconsider the estimates from \eqref{ineq:lowtwist} to \eqref{eq:lowprime}. Here we modify the definition of $\beta_{L}'$ to
\begin{equation}\label{def:betaLprime:largec}
	\begin{split}
		\beta_{L}'=\frac{4\pi^2\alpha(1-w^2/2)}{\beta_{R}}.
	\end{split}
\end{equation}
Then the estimates from \eqref{ineq:lowtwist} to \eqref{ineq:Zlowtwist} remain valid if we choose $\beta_0>2\pi$. So we still have
\begin{equation}
	\begin{split}
		\abs{I_{\pm,T\leqslant h\leqslant A-\varepsilon_1}}\leqslant&\max\limits_{x}\abs{\phi_{\pm}(x)}e^{-\varepsilon_1(\beta_{L}'-\beta_{L})}\left(\frac{4\pi^2}{\beta_{L}'}\right)^{3/2}\beta_{R}^{-1/2}e^{A\left(\frac{4\pi^2}{\beta_L'}+\frac{4\pi^2}{\beta_R}\right)} \\
		&\times\left[1+\kappa(\beta_0)\beta_0^{1/2}\left(\frac{\beta_L'}{4\pi^2}\right)^{3/2}e^{A\left(\beta_L'-\frac{4\pi^2}{\beta_{L}'}+\beta_0-\frac{4\pi^2}{\beta_0}\right)-T\left(\frac{4\pi^2}{\beta_R}-\frac{4\pi^2}{\beta_0}\right)}\right] \\
	\end{split}
\end{equation}
in the regime $\beta_{L}\leqslant\beta_{L}',\ \beta_{R}\leqslant\beta_0$. Here we also used \eqref{Ipmlowhigh:obviousbound}. The second term in the $[\ldots]$ of the second line vanishes in the HDLC$_w$ limit in the following way:
\begin{equation}
	\begin{split}
		&\kappa(\beta_0)\beta_0^{1/2}\left(\frac{\beta_L'}{4\pi^2}\right)^{3/2}e^{A\left(\beta_L'-\frac{4\pi^2}{\beta_{L}'}+\beta_0-\frac{4\pi^2}{\beta_0}\right)-T\left(\frac{4\pi^2}{\beta_R}-\frac{4\pi^2}{\beta_0}\right)} \\
		\leqslant&\left[\kappa(\beta_0)\beta_0^{1/2}\left(\frac{\alpha(1-w^2/2)}{\beta_{R}}\right)^{3/2}e^{-\frac{\pi^2\alpha w^2A}{\beta_{R}}}\right]e^{-A\left[\frac{\pi^2\alpha w^2}{\beta_{R}}+\frac{\beta_{R}}{\alpha(1-w^2)}-\beta_0+\left(1-\frac{T}{A}\right)\frac{4\pi^2}{\beta_0}\right]}.
	\end{split}
\end{equation}
Here we used $T/A\geqslant\alpha$ (see \eqref{assum1}) and \eqref{def:betaLprime:largec}. We see that both factors vanish in the HDLC$_w$ limit. Therefore we have
\begin{equation}
	\begin{split}
		\abs{I_{\pm,T\leqslant h\leqslant A-\varepsilon_1}}\stackrel{{\rm HDLC}_w}{\lesssim}\max\limits_{x}\abs{\phi_{\pm}(x)}e^{-\varepsilon_1(\beta_{L}'-\beta_{L})}\left(\frac{4\pi^2}{\beta_{L}'}\right)^{3/2}\beta_{R}^{-1/2}e^{A\left(\frac{4\pi^2}{\beta_L'}+\frac{4\pi^2}{\beta_R}\right)}.
	\end{split}
\end{equation}
Dividing above by $I^{\rm dual}_{\pm,\rm vac}$ (see \eqref{Ipmvac:asymp:largec}), we get
\begin{equation}\label{eq:low:largec}
	\begin{split}
		\abs{\frac{I_{\pm,T\leqslant h\leqslant A-\varepsilon_1}}{I^{\rm dual}_{\pm,\rm vac}}}\stackrel{{\rm HDLC}_w}{\lesssim}\frac{\max\limits_{x}\abs{\phi_{\pm}(x)}}{8\pi^{5/2}(\alpha(1-w^2/2))^{3/2}\hat{\phi}_\pm(0)}\left(A\beta_{L}^3\right)^{1/2}e^{-\varepsilon_1\left(\beta_{L}'-\beta_{L}\right)-A\left(\frac{4\pi^2}{\beta_{L}}-\frac{4\pi^2}{\beta_{L}^{\prime}}\right)}.
	\end{split}
\end{equation}
The $A\left(\frac{4\pi^2}{\beta_{L}}-\frac{4\pi^2}{\beta_{L}^{\prime}}\right)$ term in the exponential factor is irrelevant because it goes to 0 in the HDLC$_w$ limit. For the remaining factors that depend on $\beta_{L}$, $\beta_{R}$ and $A$, we have
\begin{equation}
	\begin{split}
		\sqrt{\frac{A\beta_{L}^3}{\beta_{R}\beta_{L}^{\prime3}}}e^{-\varepsilon_1\left(\beta_{L}'-\beta_{L}\right)}\leqslant&\sqrt{\frac{A}{\beta_{R}}}e^{-\varepsilon_1\left(\beta_{L}'-\beta_{L}\right)} \\
		=&\left[\sqrt{\frac{A}{\beta_{R}}}e^{-\frac{2\pi^2\alpha w^2\varepsilon_1}{\beta_{R}}}\right]\times e^{-\varepsilon_1\left(\frac{4\pi^2\alpha(1-w^2)}{\beta_{R}}-\beta_{L}\right)}.
	\end{split}
\end{equation}
Here in the first line we used $\beta_{L}\leqslant\beta_{L}'$, and in the second line we used \eqref{def:betaLprime:largec}. In the HDLC$_w$ limit, the first factor vanishes by lemma \ref{lemmaLargec}, and the second factor vanishes by \eqref{def:HDLC}. Therefore, we get
\begin{equation}
	\begin{split}
		\abs{\frac{I_{\pm,T\leqslant h\leqslant A-\varepsilon_1}}{I^{\rm dual}_{\pm,\rm vac}}}\stackrel{{\rm HDLC}_w}{\longrightarrow}0\quad(\varepsilon_1\ {\rm fixed}).
	\end{split}
\end{equation}
So we have finished the proof of the second and third equations of (\ref{modular:sketch2c}) in the case of fixed $\varepsilon_1$ and $\varepsilon_2$.

Based on our estimates on various terms in the direct and dual channels of $I_{\pm}$ (recall their definitions in \eqref{Ipm:direct}, \eqref{modularinv:split} and \eqref{def:allIpm}) we establish the statement in theorem \ref{theorem:largec} for fixed $\varepsilon_1$ and $\varepsilon_2$.  The bound $\varepsilon_i<1-1/2w$ comes from similar consideration as in section \ref{section:modular:tauberian}. As a final step, we would like to let $\varepsilon_i$ also go to zero in the HDLC$_w$ limit. This will be the subject of the next subsection.

\subsection{Shrinking the $(A-\varepsilon_1,A+\varepsilon_2)$ window}\label{section:epsilonwindow:largec}
In section \ref{section:epsilonwindow}, we observed that the $\varepsilon$-window around $h=A$ can approach zero in the DLC$_w$ limit. Now, we will revisit that analysis and study the rate at which $\varepsilon_1$ and $\varepsilon_2$ can tend to zero in the HDLC$_w$ limit.

In our previous analysis, $\varepsilon_1$ and $\varepsilon_2$ only affected two aspects: (a) the range of $\phi_{\pm,\delta_\pm}$ that we could choose, and (b) the estimate of high- and low-twist contributions (as seen in \eqref{eq:high:largec} and \eqref{eq:low:largec}). The issue regarding (a) is identical to points 1 and 2 described in section \ref{section:epsilonwindow}, and we have already resolved it. Therefore, let us now reconsider the issue regarding (b).

According to \eqref{eq:high:largec} and \eqref{eq:low:largec}, for $I_{\pm,h\geqslant A+\varepsilon_2}$ and $I_{\pm,T\leqslant h\leqslant A-\varepsilon_1}$ to be subleading in the HDLC$_w$ limit as $\varepsilon_1$ and $\varepsilon_2$ tend to 0, we require the following conditions:
\begin{equation}\label{epsiloncondition:largec}
	\begin{split}
		\left(\frac{\beta_{L}}{\beta_{R}}\right)^{3/2}A^{1/2} e^{-\varepsilon_2\left(\beta_L-\frac{4\pi^2}{\beta_0}\right)+A\beta_0},\ \left(A\beta_{L}^3\right)^{1/2}e^{-\varepsilon_1\left(\frac{4\pi^2\alpha(1-w^2/2)}{\beta_{R}}-\beta_{L}\right)}\rightarrow0.
	\end{split}
\end{equation}
Here we used the definition of $\beta_{L}'$ (see \eqref{def:betaLprime:largec}) and neglected some irrelevant exponential factors that tend to 1 in the HDLC$_w$ limit. Comparing \eqref{epsiloncondition:largec} to \eqref{epsiloncondition}, we can see that their structures are quite similar. However, in \eqref{epsiloncondition:largec}, there are additional dependence on $A$. This distinction is the primary motivation for considering separate parameters $\varepsilon_1$ and $\varepsilon_2$.

Let us consider $\varepsilon_1$-term first. We express it as follows:
\begin{equation}
	\begin{split}
		\left(A\beta_{L}^3\right)^{1/2}e^{-\varepsilon_1\left(\frac{4\pi^2\alpha(1-w^2/2)}{\beta_{R}}-\beta_{L}\right)}\stackrel{{\rm HDLC}_w}{\lesssim}&\left(A\beta_{L}^3\right)^{1/2}\ e^{-\varepsilon_1\frac{2\pi^2\alpha w^2}{\beta_{R}}}. \\
	\end{split}
\end{equation}
Here we used the fact that eventually $\beta_{L}\leqslant\frac{4\pi^2\alpha(1-w^2)}{\beta_{R}}$ in the HDLC$_w$ limit. For the r.h.s. of the above equation to vanish in the HDLC$_w$ limit, we can choose any $\varepsilon_1$ satisfying
\begin{equation}\label{epsilon1:choice}
	\begin{split}
		\varepsilon_1\geqslant\dfrac{\beta_{R}}{2\pi^2\alpha w^2}\left[\frac{3}{2}\log\beta_{L}+\left(\frac{1}{2}+\nu\right)\log A\right],
	\end{split}
\end{equation}
where $\nu$ is an arbitrary fixed positive constant. 

Then let us consider the $\varepsilon_2$-term. Similarly to the case of $\varepsilon_1$-term, here we can choose any $\varepsilon_2$ satisfying
\begin{equation}\label{epsilon2:choice}
	\begin{split}
		\varepsilon_2\geqslant\left(\beta_L-\frac{4\pi^2}{\beta_0}\right)^{-1}\left[ A\beta_0+\frac{3}{2}\log\left(\frac{\beta_{L}}{\beta_{R}}\right)+\left(\frac{1}{2}+\nu\right)\log A\right],
	\end{split}
\end{equation}
where $\nu$ is again an arbitrary fixed positive constant.

One can check explicitly that the choice of $\varepsilon_1$ and $\varepsilon_2$ given by the r.h.s. of \eqref{epsilon1:choice} and \eqref{epsilon2:choice} vanish in the HDLC$_w$ limit.\footnote{Here we added the small number $\nu$ to the $\log A$ term because it is the term with slowest growth compared to $A$, $\log\beta_{L}$ and $\log\beta_{R}$. It is of course not yet the optimal choice. For example, one can replace the $\nu \log A$ term with $\nu\log\log A$.} 

Now we choose $\beta_0=3\pi$ and $\nu=1$ in \eqref{epsilon1:choice} and \eqref{epsilon2:choice}, and recall that $\beta_{R}=2\pi\frac{A}{J}$. This gives
\begin{equation}
	\begin{split}
		\varepsilon_1\geqslant&\frac{3}{2\pi\alpha w^2}\sqrt{\frac{A}{J}}\log\left(\beta_{L} A\right), \\
		\varepsilon_2\geqslant&\left(\beta_{L}-\frac{4\pi}{3}\right)^{-1}\left[3\pi A+\frac{3}{2}\log\left(\frac{\beta_{L}\sqrt{AJ}}{2\pi}\right)\right]. \\
	\end{split}
\end{equation}
Then the last part of theorem \ref{theorem:largec} follows. This finishes the whole proof of the theorem \ref{theorem:largec}.

\small
	
\bibliography{tauberian}
\bibliographystyle{utphys}
	
\end{document}